\documentclass[journal,final,twocolumn]{IEEEtran}

\usepackage{hyperref}
\usepackage{latexsym,amsmath, bm}
\usepackage{enumerate}
\usepackage{dsfont}
\usepackage{amsfonts}
\usepackage{amssymb}
\usepackage{latexsym}
\usepackage{fixmath}
\usepackage{faktor}
\usepackage{float}
\usepackage{mathtools}
\usepackage[noend]{algorithmic}

\usepackage{multirow}
\usepackage{makecell}
\newcommand{\STAB}[1]{\begin{tabular}{@{}c@{}}#1\end{tabular}}

\usepackage[T1]{fontenc}
\usepackage{fourier}
\usepackage{bbm}
\usepackage{color}
\usepackage{caption}
\usepackage{pgfplots}
\usepackage{tikz}
\usetikzlibrary{shapes,decorations,arrows,calc,arrows.meta,fit,positioning}
\usetikzlibrary{shadows,shadings,shapes.symbols}
\tikzset{
    double color fill/.code 2 args={
        \pgfdeclareverticalshading[%
        tikz@axis@top,tikz@axis@middle,tikz@axis@bottom%
        ]{diagonalfill}{100bp}{%
            color(0bp)=(tikz@axis@bottom);
            color(50bp)=(tikz@axis@bottom);
            color(50bp)=(tikz@axis@middle);
            color(50bp)=(tikz@axis@top);
            color(100bp)=(tikz@axis@top)
        }
        \tikzset{shade, left color=#1, right color=#2, shading=diagonalfill}
    }
}
\usepgfplotslibrary{fillbetween}
\usepackage{graphicx}
\usetikzlibrary{decorations.pathreplacing, matrix}

\usepackage[boxed, linesnumbered, noend, noline]{algorithm2e}
\SetKwInput{KwData}{Input}
\SetKwInput{KwResult}{Output}

\renewcommand\log{\ln}

\newcommand\dist{\mathrm{dist}}

\newcommand\DD{{DD}}
\newcommand\SCOMP{{SCOMP}}

\newcommand{\vA}{\vec A}

\newcommand{\vX}{\vec X}
\newcommand{\vY}{\vec Y}
\newcommand{\vR}{\vec R}
\newcommand{\vH}{\vec H}

\renewcommand{\epsilon}{\eps}

\newcommand\vm{\vec m}

\renewcommand{\vec}[1]{\boldsymbol{#1}}

\newcommand\SIGMA{\vec\sigma}

\newtheorem{definition}{Definition}[section]
\newtheorem{claim}[definition]{Claim}

\newtheorem{theorem}[definition]{Theorem}
\newtheorem{lemma}[definition]{Lemma}
\newtheorem{proposition}[definition]{Proposition}
\newtheorem{corollary}[definition]{Corollary}

\newtheorem{fact}[definition]{Fact}

\newcommand\cA{\mathcal{A}}
\newcommand\cB{\mathcal{B}}
\newcommand\cC{\mathcal{C}}
\newcommand\cD{\mathcal{D}}
\newcommand\cF{\mathcal{F}}
\newcommand\cG{\mathcal{G}}
\newcommand\cE{\mathcal{E}}

\newcommand\cK{\mathcal{K}}

\newcommand\cP{\mathcal{P}}

\def\cE{{\mathcal E}}

\newcommand\vv{\vec v}

\newcommand\eul{\mathrm{e}}
\newcommand\eps{\varepsilon}

\newcommand\ZZ{\mathbb{Z}}
\newcommand\NN{\mathbb{N}}

\newcommand\Erw{\mathbb{E}}
\newcommand{\vecone}{\vec{1}}

\newcommand{\Bin}{{\rm Bin}}
\newcommand{\Mult}{{\rm Mult}}

\newcommand\bc[1]{\left({#1}\right)}
\newcommand\cbc[1]{\left\{{#1}\right\}}

\newcommand\brk[1]{\left\lbrack{#1}\right\rbrack}

\newcommand\abs[1]{\left|{#1}\right|}

\newcommand{\whp}{w.h.p.}

\newcommand\pr{\mathbb{P}} 
\renewcommand\Pr{\pr} 

\newcommand\Lem{Lemma}
\newcommand\Prop{Proposition}
\newcommand\Thm{Theorem}

\newcommand\Cor{Corollary}
\newcommand\Sec{Section}

\newcommand{\cGg}{\cG_{\Gamma}}
\newcommand{\cGd}{\cG_{\Delta}}

\newcommand{\floor}[1]{\left\lfloor#1\right\rfloor}
\newcommand{\ceil}[1]{\left\lceil#1\right\rceil}

\def\pr{{\mathbb P}}

\newcommand{\remove}[1]{}

\newcommand{\one}{V_1}
\newcommand{\zero}{V_0}
\newcommand{\zerominus}{V_{0-}}
\newcommand{\oneminusminus}{V_{1--}}
\newcommand{\zeroplus}{V_{0+}}
\newcommand{\oneplus}{V_{1+}}

\newcommand\mcount{m_{\mathrm{count}}}
\newcommand\mada{m_{\mathrm{ada}}}

\newcommand\minf{m_{\mathrm{inf}}}

\newcommand\mDD{m_{\mathrm{DD}}}

\newcommand{\mhk}[1]{\textcolor{black}{#1}}
\newcommand{\geb}[1]{\textcolor{black}{#1}}

\newcommand{\be}{\begin{equation}}
    \newcommand{\bel}[1]{\begin{equation}\lab{#1}\ }
        \newcommand{\ee}{\end{equation}}
    \newcommand{\bea}{\begin{eqnarray}}
        \newcommand{\eea}{\end{eqnarray}}
    \newcommand{\bean}{\begin{eqnarray*}}
        \newcommand{\eean}{\end{eqnarray*}}

\graphicspath{{..}}
\pgfplotsset{compat=1.14}

\usepackage{abstract} 

\begin{document} 
 
\title{Near-Optimal Sparsity-Constrained Group Testing: Improved Bounds and Algorithms}


\author{Oliver Gebhard, Max Hahn-Klimroth, Olaf Parczyk, Manuel Penschuck, \\ Maurice Rolvien, Jonathan Scarlett, and Nelvin Tan 
\IEEEcompsocitemizethanks{\IEEEcompsocthanksitem Oliver Gebhard,  oliver.gebhard@tu-dortmund.de, Faculty of Computer Science, TU Dortmund University, Dortmund, Germany, 44227.\protect\\
\IEEEcompsocthanksitem  Max Hahn-Klimroth,  maximilian.hahnklimroth@tu-dortmund.de, Faculty of Computer Science, TU Dortmund University, Dortmund, Germany, 44227.\protect\\
\IEEEcompsocthanksitem Olaf Parczyk,  parczyk@mi.fu-berlin.de, Department of Mathematics and Computer Science, FU Berlin, Berlin, Germany, 14195. \protect\\
\IEEEcompsocthanksitem Manuel Penschuck,  manuel@ae.cs.uni-frankfurt.de, Institute of Computer Science, Goethe University Frankfurt, Frankfurt, Germany, 60325. \protect\\
\IEEEcompsocthanksitem Maurice Rolvien,  maurice.rolvien@tu-dortmund.de, Faculty of Computer Science, TU Dortmund University, Dortmund, Germany, 44227.\protect\\
\IEEEcompsocthanksitem Jonathan Scarlett, scarlett@comp.nus.edu.sg, Department of Computer Science, National University of Singapore, Singapore, 117418. \protect\\
\IEEEcompsocthanksitem Nelvin Tan,  tcnt2@cam.ac.uk, Department of Engineering, University of Cambridge, UK, CB2 1PZ.\protect\\
}
\thanks{The authors are listed alphabetically.  This work was presented in part at the IEEE International Symposium on Information Theory (ISIT), 2020 \cite{tan_2020} and is accepted for publication at IEEE Transactions on Information Theory. Copyright (c) 2021 IEEE. Personal use of this material is permitted.  However, permission to use this material for any other purposes must be obtained from the IEEE by sending a request to pubs-permissions@ieee.org.}}

\maketitle

\begin{abstract}
    Recent advances in noiseless non-adaptive group testing have led to a precise asymptotic characterization of the number of tests required for high-probability recovery in the sublinear regime $k = n^{\theta}$ (with $\theta \in (0,1)$), with $n$ individuals among which $k$ are infected.  However, the required number of tests may increase substantially under real-world practical constraints, notably including bounds on the maximum number $\Delta$ of tests an individual can be placed in, or the maximum number $\Gamma$ of individuals in a given test.  While previous works have given recovery guarantees for these settings, significant gaps remain between the achievability and converse bounds.  In this paper, we substantially or completely close several of the most prominent gaps.  In the case of $\Delta$-divisible items, we show that the definite defectives (DD) algorithm coupled with a random regular design is asymptotically optimal in dense scaling regimes, and optimal to within a factor of $\eul$ more generally; we establish this by strengthening both the best known achievability and converse bounds.  In the case of $\Gamma$-sized tests, we provide a comprehensive analysis of the regime $\Gamma = \Theta(1)$, and again establish a precise threshold proving the asymptotic optimality of SCOMP (a slight refinement of DD) equipped with a tailored pooling scheme.  Finally, for each of these two settings, we provide near-optimal adaptive algorithms based on sequential splitting, and provably demonstrate gaps between the performance of optimal adaptive and non-adaptive algorithms.
\end{abstract}

\section{Introduction}

The group testing problem, originally introduced by Dorfman \cite{Dorfman_1943}, is a prominent example of a classical inference problem that has recently regained considerable attention \cite{Aldridge_2019_2, Coja_2019_2, Du_1993}.  Briefly, the problem is posed as follows: Among a population of $n$ individuals, a small subset of $k$ individuals is infected with a rare disease.  We are able to test groups of individuals at once, and each test result returns positive if (and only if) there is at least one infected individual in the test group.  The challenge is to develop strategies for pooling individuals into tests such that the status of every individual can be recovered reliably from the outcomes, and to do so using as few tests as possible.

While the preceding terminology corresponds to medical applications, group testing also has many other key applications \cite[Sec.~1.7]{Aldridge_2019_2}, ranging from DNA sequencing \cite{Kwang_2006,Ngo_2000} to protein interaction experiments \cite{Mourad_2013, Thierry_2006}.  Particular attention has been paid to group testing as a tool for the containment of an epidemic crisis. On the one hand, mass testing appears to be an essential tool to face pandemic spread \cite{Cheong_2020}, while on the other hand, the capability of efficiently identifying infected individuals fast and at a low cost is indispensable \cite{Madhav_2017}. For the sake of pandemic control, risk surveillance plans aim at an early, fast and efficient identification of infected individuals to prevent diseases from spreading \cite{EU_2009,US_2017,WHO_2009}.

The group testing problem includes many variants, depending on the presence/absence of noise, possible adaptivity of the tests, recovery requirements, and so on.  Our focus in this paper is on the following setup, which has been the focus of numerous recent works (see \cite{Aldridge_2019_2} for a survey):
\begin{itemize}
    \item The tests are {\em non-adaptive}, meaning they must all be designed in advance before observing any outcomes.  This is highly desirable in applications, as it permits the tests to be implemented in parallel.
    \item The tests are {\em noiseless}; this assumption is more realistic in some applications than others, but serves as an important starting point for understanding the problem.
    \item The goal is {\em high-probability} identification of each individual's defectivity status (i.e., probability  approaching one as $n \to \infty$).  While a deterministic (probability-one) recovery guarantee is also feasible in the noiseless setting \cite{Du_1993}, it requires considerably more tests, incurring a $k^2$ dependence on the number of infected individuals \geb{(whenever $k \le O(\sqrt{n})$)} instead of $k$.
    \item The number of infected individuals $k$ is taken to equal $n^{\theta}$ for some $\theta \in (0,1)$,\footnote{To simplify notation, we assume that $k = n^{\theta}$ exactly, but all of our analysis and results extend easily to the more general case that $k = cn^{\theta}$ for any $c = \Theta(1)$.} i.e., the {\em sublinear regime}.  Heaps' law of epidemics \cite{Benz_2008, Wang_2011} indicates that this regime is of major interest.  In addition, recent hardness results preclude non-trivial recovery guarantees in the linear regime $k = \Theta(n)$ \cite{Aldridge_2019}, at least under the most widely-adopted recovery criterion.
\end{itemize}
Under this setup, Coja-Oghlan et al.~\cite{Coja_2019,Coja_2019_2} recently established the exact information-theoretic threshold on the number of tests, in an asymptotic sense including the implied constant.  This threshold was originally attained using a \textit{random regular testing} design \cite{Coja_2019} (see also \cite{Johnson_2019}), improving on earlier results for \textit{Bernoulli testing} \cite{Aldridge_2014,Scarlett_2016}.  While the recovery algorithm used in \cite{Coja_2019} is not computationally efficient, the subsequent work \cite{Coja_2019_2} attained the same threshold using a \textit{spatially coupled random regular design} and a computationally efficient recovery algorithm.  

All of the preceding test designs have in common that each individual takes part in $O(\log n)$ tests, and each test contains $O(n/k)$ individuals.  As a result, these designs face limitations in real-world applications. Firstly, one may face dilution effects: If an infected individual gets tested within a group of many uninfected individuals, the signal of the infection (e.g., concentration of the relevant molecules) might be too low. For instance, a testing scheme for HIV typically should not contain more than 80 individual samples per test \cite{Wein_1996}. More recently, evidence was found that certain laboratory tests allow pooling of up to 5 individuals \cite{Goethe_2020} or 64 individuals \cite{Technion_2020} per test for reliably detecting COVID-19 infections.  Secondly, it is often the case that each individual can only be tested a certain number of times, due to the limited volume of the sample taken.  More generally, test designs with few tests-per-individual and/or individuals-per-test may be favorable due to resource limitations, difficulties in manually placing samples into tests, and so on.

In light of these practical issues, there is substantial motivation to study the group testing problem under the following constraints on the test design:
\begin{itemize}
    \item Under the {\em $\Delta$-divisible items constraint} (or \textit{bounded resource model}), any given individual can only be tested at most $\Delta$ times;
    \item Under the {\em $\Gamma$-sized tests constraint} (or \textit{bounded test-size model}), any given test can only contain at most $\Gamma$ individuals.
\end{itemize}
Previous studies of group testing under these constraints \cite{Gandikota_2016, Inan_2017, Macula_1996, tan_2020} are surveyed in Section \ref{Sec_RelatedWorkGT}.  We note that some of the above practical motivations may warrant more sophisticated models (e.g., random noise models for dilution effects), but nevertheless, noiseless group testing under the preceding constraints serves as an important starting point towards a full understanding.  In addition, as with previous works, we only consider the above two constraints separately, though the case that both are present simultaneously may be of interest for future studies.


\subsection{Related Work} \label{Sec_RelatedWorkGT}

As outlined above, the asymptotically optimal performance limits are well-understood in the case of unconstrained test designs, with optimal designs placing each item in $\Delta = \Theta(\log n)$ tests, and each test containing $\Gamma = \Theta\big(\frac{n}{k}\big)$ items.  We refer the reader to \cite{Aldridge_2019_2} for a more detailed survey, and subsequently focus our attention on the (much more limited) prior work considering the constrained variants with $\Delta =o(\log n)$ and $\Gamma = o\big(\frac{n}{k}\big)$.

\begin{table*} [t]
\centering
\begin{tabular}{|c|c|c|}
\hline
& Reference & Number of tests \\
\hline \hline
\multirow{5}{*}{\STAB{\rotatebox[origin=c]{90}{\makecell{$\quad$ $\Delta$-div.}}}}
& Lower Bound \cite{Gandikota_2016} & $\Delta k^{1+(1-\theta)/(\Delta\theta)}$   \\
& Lower Bound (Theorem \ref{thm_inf_theory_non_ada}) & $\max\big\{e^{-1}\Delta k^{1+(1-\theta)/(\Delta\theta)},\Delta k^{1+1/\Delta}\big\}$   \\
& COMP \cite{Gandikota_2016} & $e \Delta k n^{\frac{1}{\Delta}}$   \\
& DD (Theorem \ref{thm_DD_achievability_delta}) & $\max\big\{\Delta k^{1+(1-\theta)/(\Delta\theta)},\Delta k^{1+1/\Delta}\big\}$  \\
\hline
\hline
\multirow{7}{*}{\STAB{\rotatebox[origin=c]{90}{\makecell{$\qquad\quad$ $\Gamma$-sized}}}}
& Lower Bound \cite{Gandikota_2016} & $\frac{n}{\Gamma}$   \\
& Lower Bound (Theorem \ref{thm_gsparse_universal_informationtheory}) & $\max\big\{\big(1+\big\lfloor\frac{\theta}{1-\theta}\big\rfloor\big)\frac{n}{\Gamma},\frac{2n}{\Gamma+1}\big\}$   \\
& COMP \cite{Gandikota_2016} & $\big\lceil\frac{1}{1-\theta}\big\rceil \big\lceil\frac{n}{\Gamma}\big\rceil$ \\
& \geb{SCOMP} (Theorems \ref{Thm_DDg} and \ref{thm_dd_gamma_sparse_optimal}) & \makecell{$\max\big\{\big(1+\big\lfloor\frac{\theta}{1-\theta}\big\rfloor\big)\frac{n}{\Gamma},\frac{2n}{\Gamma+1}\big\}$}\\  
\hline
\end{tabular}
\caption{\geb{Overview of noiseless non-adaptive sparsity-constrained group testing results under the scaling $k = n^{\theta}$ ($\theta \in (0,1)$).  For the setting of $\Gamma$-sized tests, this table only corresponds to $\Gamma = \Theta(1)$, and in both settings we neglect higher-order terms and the dependence on the error probability.  See the main text for more complete and precise statements. }}
\label{tab:sparse_algo_summary}
\vspace*{-1ex}
\end{table*}
The most relevant prior work is that of Gandikota et al.~\cite{Gandikota_2016}, who gave information-theoretic lower bounds on the number of tests under both kinds of constraint, as well as upper bounds via the simple COMP algorithm \cite{Chan_2011}.\footnote{The COMP algorithm declares any individual in a negative test as uninfected, and all other individuals as infected. It is called Column Matching Algorithm in \cite{Gandikota_2016}.}  The main results therein are summarised as follows, assuming the sublinear regime $k = n^{\theta}$ with $\theta \in (0,1)$ throughout (we sometimes refer to $\theta$ as the \emph{density parameter}):
\begin{itemize}
    \item \underline{$\Delta$-divisible items setting}:
    \begin{itemize}
        \item {\bf (Converse)} For $\Delta = o(\log n)$, any non-adaptive design with error probability at most $\xi$ requires $m \ge \Delta k\big( \frac{n}{k} \big)^{\frac{1-5\xi}{\Delta}}$, for sufficiently small $\xi$ and sufficiently large $n$ . \geb{(Theorem 4.1 in \cite{Gandikota_2016})}
        \item {\bf (Achievability)} Under a suitably-chosen random test design and the COMP algorithm, the error probability is at most $\xi$ provided that $m \ge \lceil e \Delta k \big(\frac{n}{\xi}\big)^{\frac{1}{\Delta}} \rceil$. \geb{(Theorem 4.2 in \cite{Gandikota_2016})}
    \end{itemize}
    \item \underline{$\Gamma$-sized tests setting}:
    \begin{itemize}
        \item {\bf (Converse)} For  $\Gamma = \Theta\big( \big(\frac{n}{k}\big)^{\beta} \big)$ with $\beta \in [0,1)$, any non-adaptive design with error probability at most $\xi$ requires $m \ge \frac{1-6\xi}{1-\beta} \cdot \frac{n}{\Gamma}$, for sufficiently large $n$. \geb{(Theorem 4.5 in \cite{Gandikota_2016})}
        \item {\bf (Achievability)} Under a suitably-chosen random test design and COMP recovery, for $\geb{\Gamma} = \Theta\big( \big(\frac{n}{k}\big)^{\beta} \big)$ with $\beta \in [0,1)$ and $\xi = n^{-\zeta}$ with $\zeta > 0$, the error probability is at most $\xi$ when $m \ge \lceil\frac{1+\zeta}{(1-\geb{\theta})(1-\beta)} \rceil \cdot \lceil\frac{n}{\Gamma}\rceil$. \geb{(Theorem 4.6 in \cite{Gandikota_2016})}
    \end{itemize}
\end{itemize}
A sizable gap remains between the achievability and converse bounds in the case of $\Delta$-divisible items, since typically $\big(\frac{1}{\xi}\big)^{\frac{1}{\Delta}} \gg \big(\frac{1}{k}\big)^{\frac{1}{\Delta}}$.  For $\Gamma$-sized tests, the bounds match to within a constant factor, but the optimal constant remains unknown.  In particular, the two differ by at least a multiplicative $\frac{1}{1-\geb{\theta}}$ factor, and even for $\geb{\theta}$ close to zero, the two can differ by a factor of $2$ due to the rounding in the achievability part.

As we outline further below, we nearly completely close these gaps for $\Delta$-divisible items, and we close them completely for $\Gamma$-sized tests in the special case $\beta = 0$ (i.e., $\Gamma = \Theta(1)$) for all $\theta \in (0,1)$.   \geb{We achieve these results using both the DD and SCOMP algorithms introduced in \cite{Aldridge_2014}.}  
While the regime $\beta \in (0,1)$ is also of interest, it appears to require different techniques, and is deferred to future work.

Gandikota et al.~\cite{Gandikota_2016} additionally gave explicit designs (i.e., test matrices that can be deterministically constructed in polynomial time), but these give worse scaling laws, and are therefore of less relevance to our results based on random designs.  In a distinct but related line of works, Macula \cite{Macula_1996} and Inan et al.~\cite{Inan_2017,Inan_2020} developed designs for the much stronger guarantee of {\em uniform recovery}, i.e., a single test matrix that uniquely recovers any infected set of size at most $k$, without allowing any error probability.  This stronger guarantee comes at the price of requiring considerably more tests, and we thus omit a direct comparison and refer the interested reader to \cite{Macula_1996,Inan_2017,Inan_2020} for details.

\subsection{Contributions}

Our main contributions are informally outlined as follows (with $k = n^{\theta}$ for $\theta \in (0,1)$, and $\epsilon$ being an arbitrarily small constant throughout), with ``\whp'' meaning probability approaching one as $n \to \infty$.  The formal statements are given in the theorems referenced.  \geb{The results are also summarised in Table \ref{tab:sparse_algo_summary} (non-adaptive only), and exemplified in Figure~\ref{fig_bounds_illustration_delta_divisible} ($\Delta$-divisible) and Figure~\ref{fig_bounds_illustration_gamma_sparse} ($\Gamma$-sparse)}.
\begin{itemize}
    \item \underline{$\Delta$-divisible items setting.} Assuming that $\Delta = (\log n)^{1-\Omega(1)}$ (and in some cases, any $\Delta = o(\log n)$ is allowed), we have the following:
    \begin{itemize}
        \item {\bf (General converse  -- Theorem \ref{thm:converse})} If $m \le (1-\epsilon) \eul^{-1} \Delta k^{1 + \frac{1-\theta}{\Delta\theta}}$, then \whp~any (possibly adaptive) group testing strategy fails.\footnote{These expressions are obtained after substituting $k=n^{\theta}$.  In the more general case that $k$ equals a positive constant times $n^{\theta}$, the results remain unchanged upon replacing $k^{1 + \frac{1-\theta}{\Delta\theta}}$ by $k\big(\frac{n}{k}\big)^{\frac{1}{\Delta}}$ everywhere. \geb{Note also that the achievability bounds may exceed $n$ in some scaling regimes, but in such cases $m = n$ tests still suffice, since one can instead resort to one-by-one testing.}}
        \item {\bf (Non-adaptive converse  -- Theorem \ref{thm_inf_theory_non_ada})} Under any non-adaptive test design, if $\Delta > \theta/(1-\theta)$ and  $m \le (1-\epsilon) \Delta k^{1+\frac{1}{\Delta}}$, then \whp~any inference algorithm fails.  Combining with the general lower bound, the same holds for $m \le (1-\epsilon) \max \cbc{ \eul^{-1}  \Delta k^{1+\frac{1-\theta}{\Delta \theta}}, \Delta k^{1+\frac{1}{\Delta}} }$.
        \item {\bf (Non-adaptive achievability via DD -- Theorem \ref{thm_DD_achievability_delta})} Under a random regular test design, DD succeeds when $m \ge (1+\epsilon)\max \cbc{ \Delta k^{1+\frac{1-\theta}{\Delta \theta}}, \Delta k^{1+\frac{1}{\Delta}} }$ (\whp~when $\Delta = \omega(1)$, and with probability \geb{$\Omega(1)$} when $\Delta = \Theta(1)$).
        \item {\bf (DD-specific converse -- Theorem \ref{thm_DD_converse_delta})} Under random regular testing, DD fails when $m$ is slightly below the achievability bound (\whp~when $\Delta = \omega(1)$, and with $\Omega(1)$ probability when $\Delta = \Theta(1)$).
        \item {\bf (Adaptive achievability  -- Theorem \ref{thm:gamma_upperbound_ours})} There exists an efficient adaptive algorithm succeeding with probability one when $m \ge (1+\epsilon) \Delta k^{1 + \frac{1-\theta}{\Delta\theta}}$.
    \end{itemize}
    \item \underline{$\Gamma$-sized tests setting}: Assuming that $\Gamma = \Theta(1)$ in the non-adaptive setting (whereas the adaptive results allow general $\Gamma = o\big(\frac{n}{k}\big)$), we have the following:
    \begin{itemize}
        \item {\bf (Non-adaptive converse -- Theorem \ref{thm_gsparse_universal_informationtheory})} If $m \le (1-\epsilon)\max\big\{\big(1+\big\lfloor\frac{\theta}{1-\theta}\big\rfloor\big)\frac{n}{\Gamma},\frac{2n}{\Gamma+1}\big\}$ and  $\Gamma \ge 1 + \floor{\frac{\theta}{1-\theta}}$, then any non-adaptive group testing strategy fails (\whp~if $\frac{\theta}{1-\theta}$ is non-integer, and with $\Omega(1)$ probability if $\frac{\theta}{1-\theta}$ is an integer).
        \item {\bf (Non-adaptive achievability via SCOMP -- Theorems \ref{Thm_DDg} and \ref{thm_dd_gamma_sparse_optimal})} Under a suitably-chosen random test design, SCOMP succeeds \whp~when $m \ge \max\big\{\big(1+\big\lfloor\frac{\theta}{1-\theta}\big\rfloor\big)\frac{n}{\Gamma},\frac{2n}{\Gamma+1}\big\}$.  We use different test designs and analyses for the dense regime $\theta \ge \frac{1}{2}$ (Theorem \ref{Thm_DDg}) and sparse regime $\theta < \frac{1}{2}$ (Theorem \ref{thm_dd_gamma_sparse_optimal}), and combine the two results to get the overall condition in $m$ in Section \ref{sec:pieces}.  \geb{For the dense regime, our analysis shows that DD has the same guarantee, whereas for the sparse regime, we crucially require the refined SCOMP algorithm.}
        \item {\bf (Adaptive achievability -- Theorem \ref{thm_gamma_adaptive})} There exists an efficient adaptive algorithm succeeding with probability one when $m \ge (1+\epsilon)\frac{n}{\Gamma} + k\mathrm{log}_2\Gamma$.  In particular, when $\Gamma = o\big( \frac{n}{k \log n} \big)$, it suffices that $m \ge (1+\epsilon)\frac{n}{\Gamma}$.
        \item {\bf (General converse -- Theorem \ref{thm:gamma_simple_converse})} If $m \le (1-\epsilon)\frac{n}{\Gamma}$, then the error probability is bounded away from zero for any (possibly adaptive) group testing strategy.
    \end{itemize}
\end{itemize}

\begin{figure*}[!htb]
    \centering
    \begin{minipage}{.47\textwidth}
        \centering
        \includegraphics[width=0.99\linewidth]{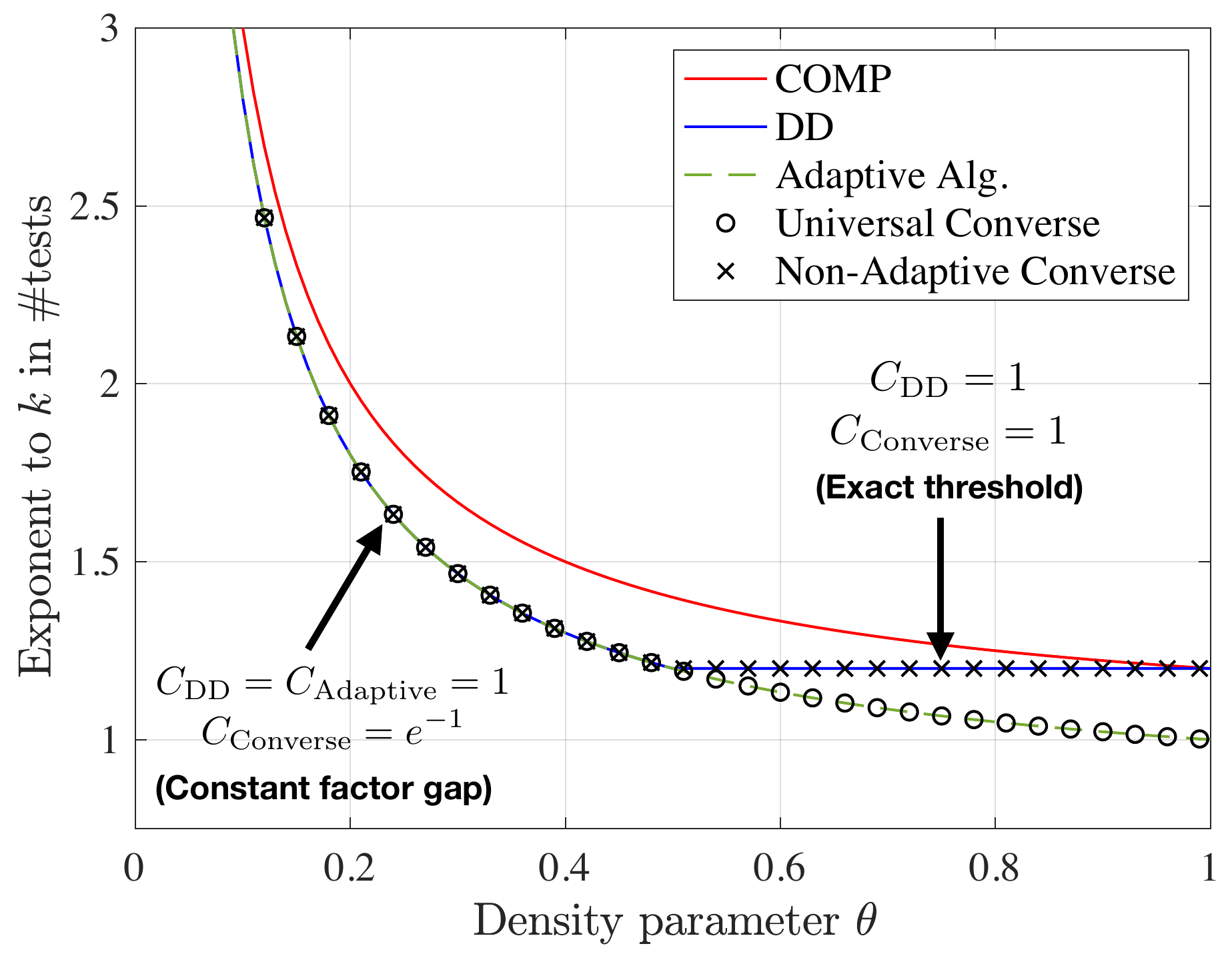}
        \caption{Illustration of values of $\eta$ (vertical axis) and $C$ (labeled with text) such that $m = C \Delta k^{\eta} (1+o(1))$ under $\Delta$-divisible item constraints, with $\Delta = 5$.}
        \label{fig_bounds_illustration_delta_divisible}
    \end{minipage}%
    \hfill
    \begin{minipage}{0.47\textwidth}
        \centering
        \includegraphics[width=0.99\linewidth]{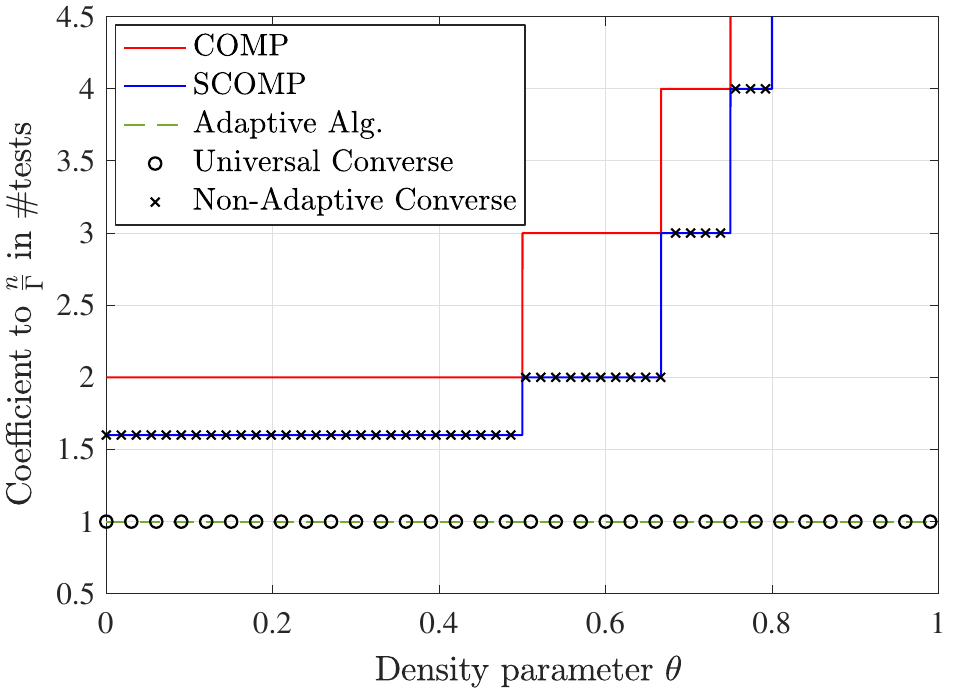}
        \caption{Illustration of threshold $C$ such that $m = (C+o(1)) \frac{n}{\Gamma}$ under $\Gamma$-sized test constraints, with $\Gamma = 4$.}
        \label{fig_bounds_illustration_gamma_sparse}
    \end{minipage}
\end{figure*}

    

    

These results have several interesting implications, which we discuss as follows.  In the $\Delta$-divisible setting, our first converse bound strengthens that of \cite{Gandikota_2016} (removing the $-5\xi$ term in the exponent) and extends it to the adaptive setting, and our second converse provides a further improvement for non-adaptive designs.  Our DD achievability result scales as $O\big( \Delta k \big(\max\big\{k,\frac{n}{k}\big\}\big)^{\frac{1}{\Delta}}\big)$, which is strictly better than the $O\big( \Delta k \big( \frac{n}{\xi} \big)^{\frac{1}{\Delta}}\big)$ scaling of COMP \cite{Gandikota_2016} for all $\theta \in (0,1)$.  In fact, for $\theta > \frac{1}{2}$ and $\Delta = \omega(1)$, our results demonstrate that DD is asymptotically optimal among non-adaptive strategies, with a precise phase transition between success and failure at $m \approx \Delta k^{1+\frac{1}{\Delta}}$.  For $\theta < \frac{1}{2}$, while establishing a precise phase transition remains an open problem, our results establish DD's optimality up to a multiplicative factor of $\eul$, and demonstrate that one cannot reduce the number of tests further under DD and the random regular design.  Finally, our results prove a strict adaptivity gap for $\theta > \frac{1}{2}$, and demonstrate that our adaptive algorithm is optimal to within a factor of $\eul$ for all $\theta \in (0,1)$.

In the $\Gamma$-sized tests setting, our results provide an exact asymptotic threshold on the number of tests in the $\Gamma = \Theta(1)$ regime, and we establish the asymptotic optimality of \geb{SCOMP} in all such cases.  To achieve this, we adopt novel analysis techniques specific to this scaling, including a novel test design in the case $\theta < \frac{1}{2}$, as described in the next section.  \geb{This case of $\theta < \frac{1}{2}$ also has the interesting feature that using SCOMP instead of DD appears to be crucial, in stark contrast with other settings in which the two algorithms tend to have identical asymptotic performance \cite{Coja_2019}. }
 We note that the distinction between integer and non-integer valued $\frac{\theta}{1-\theta}$ arises due to rounding issues in the analysis, e.g., counting the number of individuals appearing in at most $\big\lfloor\frac{\theta}{1-\theta}\big\rfloor$ tests.  Our results again demonstrate a strict adaptivity gap (this time for all $\theta \in (0,1)$), and we provide a precise phase transition at $\frac{n}{\Gamma}$ for adaptive algorithms under most scalings of $\Gamma$. Finally, in Section~\ref{simulation}, we present numerical results for small population sizes to support our theoretical findings.

\section{Fundamentals of Non-Adaptive Group Testing} \label{sec:fundamentals}

\subsection{General Notation}
Given the number of individuals $n$, the number of infected individuals $k \sim n^{\theta} (\theta \in (0,1))$, and the number of tests $m$, we let $\cG = (V \cup F, E)$ be a random bipartite (multi-)graph with $\abs{F} = m$ \textit{factor} nodes $(a_1,...,a_m)$ and $\abs{V} = n$ variable nodes $(x_1,...,x_n)$. The variable nodes represent individuals, the factor nodes represent tests, and an edge between individual $x_i$ and test $a_j$ indicates, that $x_i$ takes part in test $a_j$.
Furthermore, let $(\partial_{\cG} a_1,...,\partial_{\cG} a_m)$ and $(\partial_{\cG} x_1,...,,\partial_{\cG} x_n)$ denote the neighbourhoods in $\cG$. Whenever the context clarifies what $\cG$ is, we will drop the subscript.  The test-node degrees are given by $\Gamma_i(\cG)=\abs{\partial_{\cG} a_i}$, and the individual-node degrees by $\Delta_i(\cG)= \abs{\partial_{\cG} x_i}$. 
We can visualise any non-adaptive group testing instance by a \textit{pooling scheme} in the form of such a graph $\cG$. 

We indicate the infection status of each individual of the population by $\SIGMA\in\lbrace 0,1 \rbrace^n$, a uniformly chosen vector of Hamming weight $k$. Formally, $\SIGMA_x = 1$ iff $x$ is infected.  Then, we let $\hat{\SIGMA} = \hat \SIGMA(\cG, \SIGMA) \in \cbc{0,1}^m$  denote the sequence of test results, such that $\hat \SIGMA_a = 1$ iff test $a$ contains at least one infected individual, that is $$ \hat \SIGMA_a = \max_{x \in \partial a} \SIGMA_x. $$

Throughout the paper, we use standard Landau notation, e.g., $o(1)$ is a function converging to $0$ while $\omega(1)$ stands for an arbitrarily slowly diverging function. Moreover, we say that a property $\cP$ holds \textit{with high probability (\whp)}, if $\Pr \bc{ \cP} = 1 - o(1)$ as $n \to \infty$.

\subsection{Pooling Schemes}
The random (almost-)regular bipartite pooling scheme is known to be information-theoretically optimal in the unconstrained variant of group testing \cite{Coja_2019}, and is conceptually simple and easy to implement.  In this work, depending on the setup, we sometimes require less standard schemes, as described in the following.  It is important to note that in each of these designs, we are constructing a multi-graph rather than a graph, and every multi-edge is counted when referring to a node degree. In the following we will define our choices of the restricted pooling scheme and denote them $\cGd$ and $\tilde \cGg$
\subsubsection{$\Delta$-divisible}\label{pool_delta}

In this setup, we adopt the design of \cite{Coja_2019,Johnson_2019}, but with fewer tests per individual in accordance with the problem constraint:  Each individual chooses $\Delta$ tests uniformly at random with replacement; thus, an individual may be placed in the same test more than once. By construction of $\cGd$, any individual has degree \textit{exactly} $\Delta$, whereas the test degrees fluctuate. We denote by $\vec \Gamma\bc{\cGd} = \left\lbrace \vec \Gamma_1\bc{\cGd}, \dots, \vec \Gamma_m\bc{\cGd} \right\rbrace$ the (random) sequence of test-degrees. 

\subsubsection{$\Gamma$-sparse}\label{pool_gamma}

In the $\Gamma$-sparse case, our choice of pooling scheme requires additional care; we define $\tilde{\cGg}(\theta)$ separately for two cases:
\begin{align}\label{Def_Gg}
   \tilde{ \cGg}(\theta)=
    \begin{cases}
    \cGg \quad \text{ if } \theta \geq 1/2 \\
    \cGg^* \quad \text{ otherwise }
    \end{cases}
\end{align}
with $\cGg$ and $\cGg^*$ defined in the following.  Throughout the paper, we will always clarify which of the cases we assume, and we will therefore refer to $\tilde{\cGg}(\theta)$ as $\tilde{\cGg}$. 
Starting with $\cGg$, we employ the \textit{configuration model} \cite{Janson_2011}. Given $n, m, \Gamma$, set $\Delta = m \Gamma / n$ and create for each individual $x\in [n]$ exactly $\Delta$ clones $\cbc{x} \times \cbc{1}, \dots, \cbc{x} \times \cbc{\Delta}$. We assume throughout, that $\Delta, \Gamma, n, m$ are integers, thus all divisibility requirements are fulfilled.\footnote{It will turn out in due course that $m \Gamma / n$ is an integer under the choice of $\Gamma$ used in the analysis.} Analogously, create $\Gamma$ clones $\cbc{a} \times \cbc{1} \dots \cbc{a} \times \cbc{\Gamma}$ for each test $a \in [m]$.  Then, choose a perfect matching uniformly at random between the individual-clones and the test-clones and construct a random multi-graph by merging the clones to vertices and adding an edge $(x,a)$ whenever there are $i \in [\Delta], j \in [\Gamma]$ such that the edge $(\cbc{x} \times \cbc{i}, \cbc{a} \times \cbc{j})$ is part of the perfect matching \geb{(in other words, the edge $(x,a)$ exists in the graph as a result of the $i$-th clone of $x$ and the $j$-th clone of $a$ being matched)}. We denote by $\cGg$ the random regular multi-graph that comes from this procedure.

For $\cGg^* $, we adopt a different approach. First, we select $\gamma \leq \frac{2 n}{\Gamma + 1}$ individuals randomly and put them apart for the moment (denote by $X = \cbc{ x_1 \ldots x_{\gamma} }$ the set of those vertices).  The precise $\gamma$ value is chosen such that we can create a random bipartite regular graph on the remaining vertices with each individual having degree 2 and each test having degree $\Gamma - 1$ (thus, an instance of $\cG_{\Gamma -1}$).  By a simple comparison of degrees, this is only possible if $m \geq 2 \frac{n}{\Gamma + 1}$. Now, we draw a uniformly random matching between the tests (of degree $\Gamma-1$) and the remaining individuals $x_1 \ldots x_{\gamma}$. By definition, each of those individuals takes part in exactly one test. 

In both cases above, $\cGg^*$ is an almost-regular bipartite graph with each test comprising at most $\Gamma$ individuals.  

\subsection{Choice of recovery algorithm}

We make use of the definite defectives (\DD) \geb{and sequential combinatorial orthogonal matching pursuit (SCOMP)} algorithms \cite{Aldridge_2014}, which are described as follows.  \geb{Note that SCOMP amounts to running DD and then performing greedy improvements.}

\begin{algorithm}[h]
Declare every individual $x$ that appears in a negative test as non-infected; remove all such individuals. \\
Declare all individuals that are now the sole individual in a (positive) test as infected. \\
Proceed as follows depending on the algorithm:
\begin{itemize}
    \item For DD, declare all remaining individuals as uninfected.
    \item For SCOMP, repeat the following step until no unexplained\footnotemark  positive tests remain: \\ Declare as infected the (previously undeclared) individual in the largest number of unexplained positive tests.
\end{itemize}
\caption{The \DD~ {and SCOMP} algorithms as defined by \cite{Aldridge_2014}.}
\label{dd_algorithm}
\end{algorithm}
\footnotetext{A positive test is unexplained if it does not contain any individuals that have already been marked as infected.}

\subsection{The combinatorics behind group testing}

	\begin{figure}[h]
	\captionsetup{margin=0.01cm}
		\centering

			\begin{minipage}[t]{0.3 \textwidth}
    \begin{tikzpicture}[scale=0.55]
			\node[rectangle] (x0) at (0, 1.1) {\small $x_1 \in \zerominus$};
			\node[circle, draw, minimum width=0.5cm,text=white, fill=blue!80] (x1) at (0, 0) {$x_1$};
			
			\node[rectangle, draw, minimum width=0.5cm, minimum height=0.5cm, fill=blue!20] (a1) at (-3, -1.5) {$ $};
			\node[rectangle, draw, minimum width=0.5cm, minimum height=0.5cm, fill=red!20] (a2) at (0,-1.5) {$ $};
			\node[rectangle, draw, minimum width=0.5cm, minimum height=0.5cm, fill=red!20] (a3) at (3, -1.5) {$ $};
			
			\node[circle, draw, minimum width=0.5cm, fill=blue!80] (x2) at (-4, -3) {$ $};
			\node[circle, draw, minimum width=0.5cm, fill=blue!80] (x3) at (-3, -3) {$ $};
			\node[circle, draw, minimum width=0.5cm, fill=blue!80] (x4) at (-2, -3) {$ $};
			
			\node[circle, draw, minimum width=0.5cm, fill=blue!80] (x5) at (-1, -3) {$ $};
			\node[circle, draw, minimum width=0.5cm, fill=blue!20] (x6) at (0, -3) {$ $};
			\node[circle, draw, minimum width=0.5cm, fill=red!20] (x7) at (+1, -3) {$ $};
			
			\node[circle, draw, minimum width=0.5cm, fill=blue!80] (x8) at (2, -3) {$ $};
			\node[circle, draw, minimum width=0.5cm,  fill=red!20] (x9) at (3, -3) {$ $};
			\node[circle, draw, minimum width=0.5cm,  fill=red!20] (x10) at (4, -3) {$ $};
			\node[circle, draw=white, minimum width=0.5cm,  fill=white] (help1) at (4, -5) {$ $};
			\node[circle, draw=white, minimum width=0.5cm,  fill=white] (help2) at (3, -5) {$ $};
			\node[circle, draw=white, minimum width=0.5cm,  fill=white] (help3) at (2, -5) {$ $};
			\node[circle, draw=white, minimum width=0.5cm,  fill=white] (help4) at (1, -5) {$ $};
			\node[circle, draw=white, minimum width=0.5cm,  fill=white] (help5) at (0, -5) {$ $};
			\node[circle, draw=white, minimum width=0.5cm,  fill=white] (help6) at (-1, -5) {$ $};
			\node[circle, draw=white, minimum width=0.5cm,  fill=white] (help7) at (-2, -5) {$ $};
			\node[circle, draw=white, minimum width=0.5cm,  fill=white] (help8) at (-3, -5) {$ $};
			\node[circle, draw=white, minimum width=0.5cm,  fill=white] (help9) at (-4, -5) {$ $};
			
			\path[draw] (x1) -- (a1);
			\path[draw] (x1) -- (a2);
			\path[draw] (x1) -- (a3);
			
			\path[draw] (a1) -- (x2);
			\path[draw] (a1) -- (x3);
			\path[draw] (a1) -- (x4);
			
			\path[draw] (a2) -- (x5);
			\path[draw] (a2) -- (x6);
			\path[draw] (a2) -- (x7);
			
			\path[draw] (a3) -- (x8);
			\path[draw] (a3) -- (x9);
			\path[draw] (a3) -- (x10);
			
			\draw[dashed] (x10) -- (help1);
			\draw[dashed] (x9) -- (help2);
			\draw[dashed] (x8) -- (help3);
			\draw[dashed] (x7) -- (help4);
			\draw[dashed] (x6) -- (help5);
			\draw[dashed] (x5) -- (help6);
			\draw[dashed] (x4) -- (help7);
			\draw[dashed] (x3) -- (help8);
			\draw[dashed] (x2) -- (help9);
			\end{tikzpicture}
			\end{minipage}
			\hspace{0.1cm}
			\begin{minipage}[t]{0.3 \textwidth}
    \begin{tikzpicture}[scale=0.55]
			\node[rectangle] (x0) at (0, 1.1) {\small $x_2 \in \oneminusminus$};
			\node[circle, draw, minimum width=0.5cm, fill=red!20] (x1) at (0, 0) {$x_2$};
			
			\node[rectangle, draw, minimum width=0.5cm, minimum height=0.5cm, fill=red!20] (a1) at (-3, -1.5) {$ $};
			\node[rectangle, draw, minimum width=0.5cm, minimum height=0.5cm, fill=red!20] (a2) at (0,-1.5) {$ $};
			\node[rectangle, draw, minimum width=0.5cm, minimum height=0.5cm, fill=red!20] (a3) at (3, -1.5) {$ $};
			
			\node[circle, draw, minimum width=0.5cm, fill=blue!80] (x2) at (-4, -3) {$ $};
			\node[circle, draw, minimum width=0.5cm, fill=blue!80] (x3) at (-3, -3) {$ $};
			\node[circle, draw, minimum width=0.5cm, fill=blue!80] (x4) at (-2, -3) {$ $};
			
			\node[circle, draw, minimum width=0.5cm, fill=blue!20] (x5) at (-1, -3) {$ $};
			\node[circle, draw, minimum width=0.5cm, fill=blue!80] (x6) at (0, -3) {$ $};
			\node[circle, draw, minimum width=0.5cm, fill=red!20] (x7) at (+1, -3) {$ $};
			
			\node[circle, draw, minimum width=0.5cm, fill=blue!20] (x8) at (2, -3) {$ $};
			\node[circle, draw, minimum width=0.5cm,  fill=blue!20] (x9) at (3, -3) {$ $};
			\node[circle, draw, minimum width=0.5cm,  fill=blue!20] (x10) at (4, -3) {$ $};
			
						\node[circle, draw=white, minimum width=0.5cm,  fill=white] (help1) at (4, -5) {$ $};
			\node[circle, draw=white, minimum width=0.5cm,  fill=white] (help2) at (3, -5) {$ $};
			\node[circle, draw=white, minimum width=0.5cm,  fill=white] (help3) at (2, -5) {$ $};
			\node[circle, draw=white, minimum width=0.5cm,  fill=white] (help4) at (1, -5) {$ $};
			\node[circle, draw=white, minimum width=0.5cm,  fill=white] (help5) at (0, -5) {$ $};
			\node[circle, draw=white, minimum width=0.5cm,  fill=white] (help6) at (-1, -5) {$ $};
			\node[circle, draw=white, minimum width=0.5cm,  fill=white] (help7) at (-2, -5) {$ $};
			\node[circle, draw=white, minimum width=0.5cm,  fill=white] (help8) at (-3, -5) {$ $};
			\node[circle, draw=white, minimum width=0.5cm,  fill=white] (help9) at (-4, -5) {$ $};
			
			\path[draw] (x1) -- (a1);
			\path[draw] (x1) -- (a2);
			\path[draw] (x1) -- (a3);
			
			\path[draw] (a1) -- (x2);
			\path[draw] (a1) -- (x3);
			\path[draw] (a1) -- (x4);
			
			\path[draw] (a2) -- (x5);
			\path[draw] (a2) -- (x6);
			\path[draw] (a2) -- (x7);
			
			\path[draw] (a3) -- (x8);
			\path[draw] (a3) -- (x9);
			\path[draw] (a3) -- (x10);
			
			\draw[dashed] (x10) -- (help1);
			\draw[dashed] (x9) -- (help2);
			\draw[dashed] (x8) -- (help3);
			\draw[dashed] (x7) -- (help4);
			\draw[dashed] (x6) -- (help5);
			\draw[dashed] (x5) -- (help6);
			\draw[dashed] (x4) -- (help7);
			\draw[dashed] (x3) -- (help8);
			\draw[dashed] (x2) -- (help9);
			\end{tikzpicture}
			\end{minipage}
			\hspace{0.1cm}
			\begin{minipage}[t]{0.3 \textwidth}
    \begin{tikzpicture}[scale=0.55]
			\node[rectangle] (x0) at (0, 1.1) {\small $x_3 \in \zeroplus \cup \oneplus$};
			\node[circle, draw, minimum width=0.5cm, double color fill={red!20}{blue!20}, shading angle=45] (x1) at (0, 0) {$x_3$};
			
			\node[rectangle, draw, minimum width=0.5cm, minimum height=0.5cm, fill=red!20] (a1) at (-3, -1.5) {$ $};
			\node[rectangle, draw, minimum width=0.5cm, minimum height=0.5cm, fill=red!20] (a2) at (0,-1.5) {$ $};
			\node[rectangle, draw, minimum width=0.5cm, minimum height=0.5cm, fill=red!20] (a3) at (3, -1.5) {$ $};
			
			\node[circle, draw, minimum width=0.5cm, fill=red!20] (x2) at (-4, -3) {$ $};
			\node[circle, draw, minimum width=0.5cm, fill=blue!80] (x3) at (-3, -3) {$ $};
			\node[circle, draw, minimum width=0.5cm, fill=blue!20] (x4) at (-2, -3) {$ $};
			
			\node[circle, draw, minimum width=0.5cm, fill=blue!80] (x5) at (-1, -3) {$ $};
			\node[circle, draw, minimum width=0.5cm, fill=blue!80] (x6) at (0, -3) {$ $};
			\node[circle, draw, minimum width=0.5cm, fill=red!20] (x7) at (+1, -3) {$ $};
			
			\node[circle, draw, minimum width=0.5cm, fill=blue!20] (x8) at (2, -3) {$ $};
			\node[circle, draw, minimum width=0.5cm,  fill=red!20] (x9) at (3, -3) {$ $};
			\node[circle, draw, minimum width=0.5cm,  fill=red!20] (x10) at (4, -3) {$ $};
			
						\node[circle, draw=white, minimum width=0.5cm,  fill=white] (help1) at (4, -5) {$ $};
			\node[circle, draw=white, minimum width=0.5cm,  fill=white] (help2) at (3, -5) {$ $};
			\node[circle, draw=white, minimum width=0.5cm,  fill=white] (help3) at (2, -5) {$ $};
			\node[circle, draw=white, minimum width=0.5cm,  fill=white] (help4) at (1, -5) {$ $};
			\node[circle, draw=white, minimum width=0.5cm,  fill=white] (help5) at (0, -5) {$ $};
			\node[circle, draw=white, minimum width=0.5cm,  fill=white] (help6) at (-1, -5) {$ $};
			\node[circle, draw=white, minimum width=0.5cm,  fill=white] (help7) at (-2, -5) {$ $};
			\node[circle, draw=white, minimum width=0.5cm,  fill=white] (help8) at (-3, -5) {$ $};
			\node[circle, draw=white, minimum width=0.5cm,  fill=white] (help9) at (-4, -5) {$ $};
			
			\path[draw] (x1) -- (a1);
			\path[draw] (x1) -- (a2);
			\path[draw] (x1) -- (a3);
			
			\path[draw] (a1) -- (x2);
			\path[draw] (a1) -- (x3);
			\path[draw] (a1) -- (x4);
			
			\path[draw] (a2) -- (x5);
			\path[draw] (a2) -- (x6);
			\path[draw] (a2) -- (x7);
			
			\path[draw] (a3) -- (x8);
			\path[draw] (a3) -- (x9);
			\path[draw] (a3) -- (x10);
			
		\draw[dashed] (x10) -- (help1);
			\draw[dashed] (x9) -- (help2);
			\draw[dashed] (x8) -- (help3);
			\draw[dashed] (x7) -- (help4);
			\draw[dashed] (x6) -- (help5);
			\draw[dashed] (x5) -- (help6);
			\draw[dashed] (x4) -- (help7);
			\draw[dashed] (x3) -- (help8);
			\draw[dashed] (x2) -- (help9);
			
			\end{tikzpicture}
			\end{minipage}
		\caption{Rectangles represent tests and circles individuals. Dark blue individuals are elements of $\zerominus$ and can be easily identified as uninfected. Light blue individuals are elements of $\zeroplus$, and even if uninfected themselves, they only appear in positive tests and might be hard to identify. Infected individuals (red) that appear only in such tests are impossible to identify. Finally, infected individuals of $\oneminusminus$ appear in at least one test with only elements of $\zerominus$. Thus, after identifying all elements of $\zerominus$, they can be identified. \geb{The dashed lines represent the fact that the individuals may also participate in other tests; these may include negative tests classifying their participants as uninfected (elements of $\zerominus$) even though the particular test displayed is positive.}}
		\label{figure_types_of_individuals}
	\end{figure}
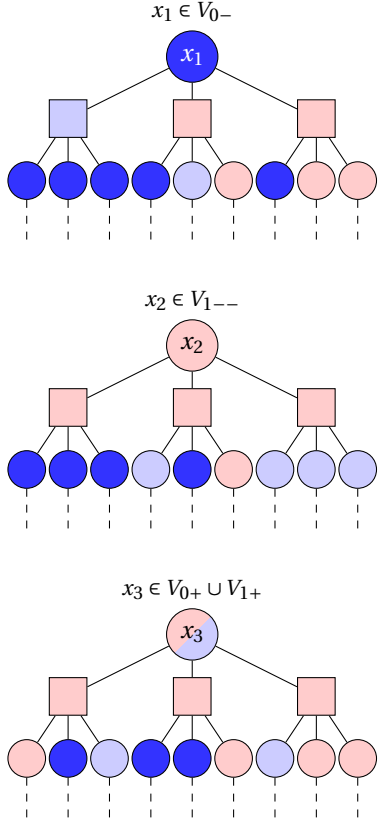

\label{Sec_Individualtypes}In this section, we introduce four types of individuals (see Figure \ref{figure_types_of_individuals}) that might appear in any group testing instance and which the student can make use of. It turns out that the sizes of the sets of these individuals are the key to understanding group testing combinatorially.
Given a pooling scheme $\cG$, let 
\begin{align*}
    &\zero(\cG) = \cbc{x \in V(\cG): \SIGMA_x = 0}\\
    \text{and}\qquad &\one(\cG) = \cbc{ x \in V(\cG) : \SIGMA_x = 1 }
\end{align*}
$$  \qquad  $$ 
be the uninfected and infected individuals, respectively.
Then we can define \textit{easy uninfected} individuals to be the uninfected individuals that appear in a negative test -- clearly, they can easily be identified. We will call the set of such individuals $\zerominus$; formally,
\begin{align}\label{Def_zerominus}
    \zerominus(\cG) = \cbc{ x \in \zero(\cG): \exists a \in \partial_{\cG} x: \hat \SIGMA_a = 0  }.
\end{align}
Then, there the \textit{easy infected} individuals \geb{(sometimes referred to as \emph{definitive defectives})}. These are those infected individuals that appear in at least one test with only easy uninfected individuals. Thus, upon removing the easy uninfected individuals, there will be at least one positive test with exactly one undeclared individual, and this individual has to be infected. We call this set
\begin{align}\label{Def_oneminusminus}
  \oneminusminus(\cG) = \cbc{ x \in \one(\cG): \exists a \in \partial_{\cG} x: (\partial_{\cG} a \setminus \{x\} )\subset \zerominus(\cG) }.  
\end{align}
Subsequently, there might be \textit{disguised uninfected} individuals, that are uninfected themselves but only appear in positive tests. It is well known \cite{Aldridge_2016, Coja_2019, Coja_2019_2} that since the prior probability of being uninfected is very large, a group testing instance can tolerate a certain number of individuals of this type. Formally,
 \begin{align}\label{def0+}
    \zeroplus(\cG) = \cbc{ x \in \zero(\cG): \forall a \in \partial_{\cG} x: \hat \SIGMA_a = 1 }.  
 \end{align}

Finally, there might be \textit{disguised infected} individuals, thus infected individuals appearing only in tests that contain at least one more infected individual. Formally, 
\begin{align}\label{def1+}
\oneplus(\cG) = \cbc{ x \in \one(\cG): \forall a \in \partial_{\cG} x \,:\, (\partial_{\cG} a \setminus \cbc{x}) \cap \one(\cG) \neq \emptyset }.    
\end{align}
\geb{While the above types of individuals are not exhaustive, we will see in Section~\ref{Sec_Nishimori} that they are the relevant types for the information-theoretic and algorithmic analyses.}
\subsubsection{Remarks on information-theoretic and combinatorial bounds}
It turns out that in the sparse group testing problem -- as well as in the unrestricted version \cite{Aldridge_2014, Coja_2019_2}
-- the non-adaptive information-theoretic phase transition comes \geb{in} two installments. First, there are universal information-theoretic bounds, e.g., counting bounds, that account for the fact that a given number of tests can carry only a certain amount of information.  Such bounds directly apply to the non-adaptive as well as the adaptive setting. Second, there are combinatorial / graph theoretical restrictions:  Given that there exist a large number of disguised infected individuals (i.e., individuals such that in each of its tests there is a second infected individual), any non-adaptive algorithm fails with high (conditional) probability \cite{Coja_2019,Coja_2019_2}. This non-adaptivity gap becomes stronger if we increase the \geb{infection} density parameter $\theta$, because for larger $\theta$, the chance of finding multiple infected individuals in a small neighborhood increases as well. In this section we deal with the combinatorial part.  In our setting, the transition where the combinatorial bound dominates the information-theoretic bound happens at $k \sim \sqrt{n}$, i.e., at the point where we find multiple infected individuals in a bounded neighborhood \whp.

\subsection{The Nishimori property}\label{Sec_Nishimori}
Given a pooling scheme $\cG$, a ground truth infection status vector $\SIGMA$ (drawn uniformly from the vectors of Hamming weight $k$) and a sequence of test results $\hat \SIGMA$, we denote by $S_k(\cG, \SIGMA)$ the set of all colorings (i.e., infection status assignments) of individuals $\tau \in \cbc{0,1}^n$ that would have led to the test outcomes $\hat \SIGMA$ (clearly including $\SIGMA$ itself). Furthermore, we define $Z_k(\cG, \SIGMA) = \abs{S_k(\cG, \SIGMA)}$.  The following proposition states that all sets in $S_k(\cG, \SIGMA)$ are equally likely given the test outcomes.

\begin{proposition}{{\em [Corollary 2.1 of \cite{Coja_2019}]}} \label{prop:Nishimori}
For all $\tau\in\lbrace 0,1 \rbrace^{n}$ we have 
$$\Pr(\SIGMA=\tau|\cG, \hat \SIGMA)=\frac{\mathbb{1}\lbrace \tau\in S_k(\cG,\SIGMA)\rbrace}{Z_k(\cG,\SIGMA)}.$$
\end{proposition}
\noindent This immediately implies the following corollary.
\begin{corollary} \label{Cor_Nishi}
If $Z_k(\cG,\SIGMA) \geq \ell$ \whp, then any inference algorithm recovers $\SIGMA$ from $(\cG, \hat \SIGMA)$ with probability at most $\ell^{-1}(1+o(1))$.
\end{corollary}
\geb{In other words, as soon as multiple satisfying assignments exist, one cannot do any better than selecting one uniformly at random, as no further information is included in $\cG$ and $\hat \SIGMA$ \cite{lenka_stat_inf}.}
The following claims will also be useful.

\begin{claim}
\label{Claim_1+0+}
For any test design, we have $Z_k(\cG,\SIGMA) \ge \abs{\oneplus(\cG)} \abs{\zeroplus(\cG)}$.  Hence, conditioned on the sets $\oneplus(\cG)$ and $\zeroplus(\cG)$, any inference algorithm fails with probability at least $1 - \frac{1}{\abs{\oneplus(\cG)} \abs{\zeroplus(\cG)}}$.
\end{claim}
\begin{IEEEproof}
The first statement is straightforward and was already given in \cite[Fact 3.3]{Coja_2019}, and the second statement follows directly from Corollary \ref{Cor_Nishi}.
\end{IEEEproof}


Finally, we have the following well-known result on the \DD~algorithm.

\begin{claim}
\label{claim_dd_individualtypes}
The \DD~ algorithm succeeds if and only if $\one(\cG) = \oneminusminus(\cG)$.
\end{claim}
\begin{IEEEproof}
By definition, \DD~ first classifies all $x\in \zerominus(\cG)$ correctly. In the second step, \DD~ classifies those individuals $x$ as infected, which belong to a positive test $a$ such that $\partial a \setminus \cbc{x} \subset \zerominus(\cG)$. Thus, \DD~ finds all $x\in \one \cap \oneminusminus(\cG)$. As \DD~ classifies the remaining individuals as uninfected, it fails as soon as there exists an individual $x \in \one \setminus \oneminusminus(\cG)$.
\end{IEEEproof}
\geb{We note that even if $\one(\cG) \setminus \oneminusminus(\cG) \neq \emptyset$, the \DD-algorithm does not produce any false positives but only false negatives.  In addition, if \DD~succeeds then SCOMP is guaranteed to succeed \cite{Aldridge_2017a}, but unlike \DD, in general SCOMP may produce both false positive and false negatives.}

\subsection{The two-round exposure technique}\label{two_round_exp}
A key tool to deal with an arbitrary test design is to introduce certain levels of independent randomness.
For example, the only randomness in $(\cG, \SIGMA)$ is the infection status of each individual. We will see in due course we can study an independent infection model (denoted by $\SIGMA^*$) instead of dealing with exactly $k$ infected individuals, specifically considering each individual as being infected independently from all others with probability $p=\frac{k-\sqrt{k} \log n}{n}$ (see Corollary~\ref{cor_bernoulli_ok}).
\geb{For the purposes of establishing a converse, the main step is to show that $\oneplus(\cG) \neq \emptyset$, and} we will establish this in two steps. 
We denote by $V_+(\cG)$ the set of \textit{disguised} individuals, i.e.,~all tests containing this individual $x$ contain at least one other individual (differing from $x$) that is infected, and hence
\[V_+(\cG) = \oneplus(\cG) \cup \zeroplus(\cG).\]
Once we find a large enough set $\abs{V_+(\cG)} \gg n/k$, there will be some infected individuals in $V_+(\cG)$ \whp.
The main challenge is that in order to find the set of disguised individuals, one uses infected individuals, therefore the events $\abs{V_+(\cG)}$ exceeding a specific size and infected individuals existing in $V_+(\cG)$ are not independent in $(\cG, \SIGMA^*)$.
This is where the two-round exposure technique, used very prominently in the study of random graphs \cite{Janson_2011}, comes into account.

More specifically, our analysis will take the following steps in which individuals are randomly infected:

\begin{enumerate}
    \item We first mark each individual as infected with probability $\alpha k/n$ for some fixed constant $\alpha \in (0,1)$ and find a set $\cK_1$ of infected individuals whose neighbourhood (the tests they belong to) has certain properties.
    \item Next, we mark the remaining individuals in the second neighbourhood of $\cK_1$ (hence, we look at the individuals that are contained in the tests together with the vertices of $\cK_1$) as infected independently with probability $(1 - 2\alpha)k/n$ for establishing the property of being disguised.
    \item After the previous step, each individual has been infected with probability at most $\alpha k/n + (1-\alpha k/n) (1-2\alpha) k/n < p$. \geb{To attain the desired final distribution of $\SIGMA^*$, we independently mark each individual $i \in [n]$ as infected with probability $p - p_i$, where $p_i$ is the probability already incurred from the first two steps.  By doing so, the overall distribution of $\SIGMA^*$ is i.i.d.~with probability $p$, as desired.  While these extra infections are not actually analyzed, the idea is that} they produce the desired overall distribution, while only enlarging (or keeping unchanged) the set of individuals that are disguised.
\end{enumerate}

\section{Non-Adaptive Group Testing with $\Delta$-Divisible Individuals}

In this section, we formally state and prove our main results regarding non-adaptive group testing with $\Delta$-divisible individuals.

\subsection{Model}\label{model}

As we highlighted earlier, optimal unconstrained designs are known that place each individual in $\Theta(\log n)$ tests.  Accordingly, we only consider the regime $\Delta = o(\log n)$, and specifically suppose that $\Delta \le \log^{1- \delta} n$ for some constant $\delta \in (0,1)$.

\subsection{Results}
Define
\geb{
\begin{align} \label{minf_nonadaptive}
    &\minf(\Delta) = \Delta k \max \cbc{e^{-1}k^{\frac{(1-\theta)}{\Delta \theta}},  k^{\frac{1}{\Delta}} },\notag \\ &  \mDD(\Delta) = \Delta k \max \cbc{ k^{\frac{(1-\theta)}{\Delta \theta}}, k^{\frac{1}{\Delta}} },
\end{align}
}
which will represent the information theoretic converse bound for any non-adaptive group testing scheme and the algorithmic barrier for \DD, respectively.

In the following, we assume that 
\geb{$\Delta \ge 2$ and $\Delta > \theta/(1-\theta)$}. 
If the latter inequality is reversed, then we find that $\mDD(\Delta) = \omega(n)$, in which case one is better off resorting to one-by-one testing.

Our first main result provides a simple counting-based converse bound for any adaptive or non-adaptive test design.  This result, and all subsequent results, will be proved throughout the rest of the section. \geb{An overview of the proof strategy will be provided in Section~\ref{Overview_Delta}}

\begin{theorem} \label{thm:converse} 
    Fix $\epsilon\in(0,1)$, and suppose that $k = n^{\theta}$ with $\theta \in (0,1)$ and $\Delta = o(\log(n))$.  Then, if $m \le (1-\epsilon) \eul^{-1} \Delta k^{1+\frac{(1-\theta)}{\Delta \theta}}$ for fixed $\epsilon > 0$, we have w.h.p.~that any (possibly adaptive) group testing procedure that tests each individual at most $\Delta$ times fails to recover $\SIGMA$.
\end{theorem} 

This bound recovers the first term of $\max\{\cdot,\cdot\}$ appearing in the definition of $\minf(\Delta)$ above, which is dominant for $\theta \le 1/2$.  For the second term (which is dominant for $\theta \ge 1/2$), we require a more sophisticated argument that only holds for non-adaptive designs; as we will see in Section \ref{sec:adaptive}, adaptive designs can in fact go beyond this threshold. \geb{The proof of Theorem~\ref{thm:converse} is given in Section~\ref{proof_thm_converse}.}

\begin{theorem}
\label{thm_inf_theory_non_ada}
Given any non-adaptive pooling scheme $\cG$ where any individual gets tested at most $\Delta$ times (with $\theta/(1-\theta) < \Delta \le (\ln n)^{1-\delta}$ for some $\delta > 0$), if $m \le (1 - \eps) \Delta k^{1 + 1/\Delta}$ for some $\eps \in (0,1)$, any algorithm (efficient or not) fails at inferring $\SIGMA$ from $(\cG, \hat \SIGMA)$, \geb{with probability $1-o(1)$ if $\Delta = \omega(1)$, and with probability $\Omega(1)$ if $\Delta = O(1)$.}

\end{theorem}

Combining these results, we find that any non-adaptive group testing strategy using at most $(1-\epsilon)\minf(\Delta)$ tests fails w.h.p.~if $\Delta = \omega(1)$, and fails with constant non-zero probability if $\Delta = O(1)$. 
\geb{We provide the proof of Theorem~\ref{thm_inf_theory_non_ada} in Section~\ref{sec_proof_inf_theory_non_ada}.}
Next, we state our main upper bound, corresponding to the random regular design and the DD algorithm.



\begin{theorem}\label{thm_DD_achievability_delta}
Suppose that $m = (1 + \eps) \mDD(\Delta)$ for some $\epsilon > 0$. Then, under the random regular design with parameter $\Delta$, \DD~ recovers $\SIGMA$ from $(\cGd, \hat \SIGMA)$ with probability at least $1 - \bc{1 + \eps}^{- \Delta}(1+o(1)) - O(n^{-\Omega(1)})$.
\end{theorem}

Note that the success probability tends to one as $\Delta \to \infty$; if $\Delta = O(1)$ then we need to take $\epsilon \to \infty$ for the probability to approach one (but it can be close to one for finite $\epsilon$). \geb{The proof of Theorem~\ref{thm_DD_achievability_delta} is given in Section~\ref{proof_thm_DD_achievability_delta}.}
Comparing this result with Theorem \ref{thm:converse}, we find that DD is asymptotically optimal for $\theta \geq 1/2$.  On the other hand, a gap between $\minf(\Delta)$ and $\mDD(\Delta)$ remains for $\theta < \frac{1}{2}$.  In principle, this could be due to a weakness in the converse, a fundamental limitation of DD, or a weakness in our analysis of DD.  However, the following theorem rules out the latter of these.

\begin{theorem}\label{thm_DD_converse_delta}
Let $\theta < 1/2$. Given the random regular pooling scheme $\cGd$ on $m = (1 - \eps) \mDD(\Delta)$ tests for fixed $\epsilon \in (0,1)$, we have the following:
\begin{enumerate}
    \item If $\Delta = \Theta(1)$, then \DD~fails with positive probability bounded way from zero.
    \item If $\Delta = (\log n)^{1-\delta}$ for $\delta \in (0,1)$, then \DD~ fails \whp.
\end{enumerate} 
\end{theorem}

Thus, \Thm~\ref{thm_DD_converse_delta} settles a coarse phase transition of \DD~ in the random regular model when there are finitely many tests-per-individual, and a sharp phase transition when the number of tests-per-individual is diverging. 
\geb{The proof of \Thm~\ref{thm_DD_converse_delta} is provided in Section~\ref{proof_thm_DD_converse_delta}.}
We expect that DD is in fact provably suboptimal for $\theta < \frac{1}{2}$, but leave this as an open problem.
\geb{
\subsubsection{Overview of proofs}\label{Overview_Delta}
Before proving Theorems~\ref{thm:converse}--\ref{thm_DD_converse_delta}, we provide a brief overview:
\begin{itemize}
    \item To prove Theorem~\ref{thm:converse}, we establish an upper bound on the probability that an arbitrary inference algorithm recovers $\SIGMA$ correctly based on the amount of information provided by the test results (which is inherently limited due to the testing constraints). This already suffices to show that as soon as the number of tests crosses a certain lower bound, any inference algorithm must have an error probability approaching one.
    \item Theorem~\ref{thm_inf_theory_non_ada} deals with non-adaptive designs, which can be represented as a bipartite graph. The main argument is that when there are too many disguised infected and disguised uninfected individuals, perfect recovery becomes impossible, since interchanging these two types of individuals would not impact the test results.  We carefully analyse the number of occurrences of these disguised individuals by the means of local structures in the graph (see Figure \ref{figure_types_of_individuals}). 
    \item Theorem~\ref{thm_DD_achievability_delta} provides performance guarantees for the \DD-algorithm in the $\Delta-$divisible setting. As this algorithm succeeds if and only if all infected individuals appear in one test containing only definitive uninfected individuals (c.f., Sections~ \ref{Sec_Individualtypes} and \ref{Sec_Nishimori}), it suffices to analyse a carefully-chosen pooling scheme and pinpoint the number of tests required such that all infected individuals exhibit this property.
    \item Finally, we prove Theorem~\ref{thm_DD_converse_delta} by showing that as soon as the number of tests is too small, there exists a large number of infected individuals that fail to participate in any tests containing only definitive uninfected individuals.
\end{itemize}
}

\subsection{Universal counting-based converse: Proof of Theorem \ref{thm:converse}}\label{proof_thm_converse}
We first prove a counting-based upper bound on the success probability for any test design and inference algorithm.
Afterwards, we will use this bound on the success probability 
to prove our main converse bound, providing a lower bound on $m$ for attaining a given target error probability.\\

Let $\cA(\cG,\hat \SIGMA,k)$ be the output of a group testing inference algorithm with input $\cG$ (pooling scheme), $\hat \SIGMA$ (test results), and $k$ (number of infected individuals). The inference algorithm is successful if $\cA(\cG,\hat \SIGMA,k)=\SIGMA$, and $\Pr\left(\cA(\cG,\hat \SIGMA,k)=\SIGMA\right)$ is the success probability.  We first prove the following non-asymptotic counting-based bound via a similar approach to \cite{Baldassini_2013} with suitable adjustments, and also using the Nishimori property similarly to \cite{Coja_2019}.

\begin{lemma} \label{lem:counting_bound} 
  Under the preceding setup, for any pooling scheme $\cG$ and inference algorithm $\cA(\cG,\hat \SIGMA,k)$, we have
\begin{align}
    \mathbb{P}(\cA(\cG,\hat \SIGMA,k)=\SIGMA)\leq\frac{\sum_{i=0}^{\Delta k}{m\choose i}}{{n\choose k}}. \label{eq:prob_suc_bound}
\end{align}
\end{lemma}
\begin{IEEEproof}
    Any given pooling scheme 
    can be viewed as a deterministic mapping from an infection status vector $\SIGMA \in \{0,1\}^n$ to an outcome vector $\hat\SIGMA \in \{0,1\}^m$.  Recall that in Proposition \ref{prop:Nishimori}, $S_k(\cG, \SIGMA)$ is the set of all colorings of individuals that lead to the testing sequence $\hat \SIGMA$, and $Z_k(\cG,\SIGMA)$ is its cardinality.  In the following, we additionally let $\hat{Z}_k(\cG,\hat\SIGMA)$ denote $Z_k(\cG, \SIGMA)$ when the test outcomes produced by $(\cG, \SIGMA)$ are equal to $\hat\SIGMA$, and let $\hat{S}_k(\cG,\hat\SIGMA)$ be the set of all $\SIGMA$ sequences that produce test outcomes $\hat\SIGMA$.
    
    Proposition \ref{prop:Nishimori} shows that the optimal inference algorithm outputs an arbitrary element of $S_k(\cG, \SIGMA)$, and is correct with probability (conditioned on $\SIGMA$) equal to $\frac{1}{Z_k(\cG, \SIGMA)}$.   Thus, averaging over the $n \choose k$ possible $k$-sparse vectors $\SIGMA$, we have the following:
    \begin{align*}
        \mathbb{P}(\cA(\cG,\hat \SIGMA,k)=\SIGMA)
        &=\frac{1}{{n\choose k}} \sum_{\SIGMA} \frac{1}{Z_k(\cG, \SIGMA)} \\
        &=\frac{1}{{n\choose k}} \sum_{\hat\SIGMA \,:\, \hat{Z}_k(\cG,\hat\SIGMA) \ge 1} \sum_{\SIGMA \in \hat{S}_k(\cG,\hat\SIGMA)} \frac{1}{\hat{Z}_k(\cG,\hat\SIGMA)} \\
        &\stackrel{(a)}{\leq}\frac{|\{\hat{\boldsymbol{\sigma}}\in\{0,1\}^m \,:\, \hat{Z}_k(\cG,\hat\SIGMA) \ge 1\}|}{{n\choose k}}\\
        &\stackrel{(b)}{\leq}\frac{|\{\text{$\hat{\boldsymbol{\sigma}}$ with at most $\Delta k$ ones}\}|}{{n\choose k}} \\
        &=\frac{\sum_{i=0}^{\Delta k}{m\choose i}}{{n\choose k}},
    \end{align*}
    where (a) follows since there are $\hat{Z}_k(\cG,\hat\SIGMA)$ terms in the second summation, thus canceling the $\frac{1}{\hat{Z}_k(\cG,\hat\SIGMA)}$ term, and (b) uses the fact that at most $\Delta k$ test outcomes can be positive, even in the adaptive setting; this is because adding another infected individual always introduces at most $\Delta$ additional positive tests.
\end{IEEEproof}
We now use the result in \eqref{eq:prob_suc_bound} to prove Theorem~\ref{thm:converse} - \ref{thm:converse}. In the following we want to provide a short overview of how we obtain these results

\begin{IEEEproof}[Proof of Theorem \ref{thm:converse}]
\geb{Let $\mcount(\Delta) = \eul^{-1} \Delta k^{1+\frac{(1-\theta)}{\Delta \theta}}$ denote the threshold in the theorem statement.} It suffices to prove the claim for $m = (1-\epsilon) \mcount(\Delta)$, since the inference algorithm could choose to ignore tests. We use the non-asymptotic bound in Lemma \ref{lem:counting_bound}, and upper bound the sum of binomial coefficients via \cite[Section~4.7.1]{Ash90} to obtain the following for a fixed target success probability of $1-\xi$ (for some $\xi \in (0,1)$):
    \begin{align}
       \Pr\left(\cA(\cG,\hat \SIGMA,k)=\SIGMA\right) \leq\frac{e^{mh(\frac{\Delta k}{m})}}{{n\choose k}} \equiv 1-\xi, \label{eq:weak_counting_bound}
    \end{align}
    where $h(\cdot)$ is the binary entropy function in nats (logs to base $\eul$). From \eqref{eq:weak_counting_bound}, we have $e^{mh(\frac{\Delta k}{m})}/{n\choose k}=1-\xi$, which implies that
    \begin{align}
        \log\bigg((1-\xi){n\choose k}\bigg)&=mh\Big(\frac{\Delta k}{m}\Big)\notag\\
        &=\Delta k\log{\frac{m}{\Delta k}}+(m-\Delta k)\log\frac{1}{1-\frac{\Delta k}{m}}\notag\\
        &\stackrel{(a)}{=}\Delta k\log\frac{m}{\Delta k}+\Delta k(1+o(1)), \label{eq:simplified_entropy}
    \end{align}
    where (a) uses a Taylor expansion and the fact that $\frac{\Delta k}{m}\in o(1)$ (due to $\Delta = o(\ln n)$ and $m = (1-\geb{\eps}) \mcount(\Delta)$). Hence, we have $(1-\frac{\Delta k}{m})^{-1}=\exp(\frac{\Delta k}{m})(1+o(1))$ which is used to obtain the simplification. Rearranging \eqref{eq:simplified_entropy}, we obtain
    \begin{align*}
        \log\frac{m}{\Delta k}=\frac{1}{\Delta k}\log\left((1-\xi){n\choose k}\right)-(1+o(1)),
    \end{align*}
    which gives
    \begin{align}
        m&=e^{-(1+o(1))}\Delta k\left((1-\xi){n\choose k}\right)^{\frac{1}{\Delta k}} \nonumber \\
        &\stackrel{(a)}{\geq}e^{-(1+o(1))}(1-\xi)^{\frac{1}{\Delta k}}\Delta k^{1+\frac{1-\theta}{\theta \Delta}}, \label{eq:tests}
    \end{align}
    where (a) follows from the fact that ${n\choose k}\geq\big(\frac{n}{k}\big)^k$ and $k = n^{\theta}$.
    
    \geb{Finally, we note that $(1-\xi)^{1/(\Delta k)}\rightarrow1$ for any fixed $\xi \in (0,1)$, since $k \to \infty$ by assumption.  This means that $m$ must be at least $(1-o(1)) \eul^{-1} \Delta k^{1+\frac{(1-\theta)}{\Delta \theta}}$ to obtain any arbitrarily small success probability, and hence, if $m$ is instead a $(1-\geb{\eps})$ factor below this threshold (as we have assumed) then the success probability must tend to zero. }
\end{IEEEproof}


\subsection{Universal converse for non-adaptive designs: Proof of Theorem \ref{thm_inf_theory_non_ada}}
\label{sec_proof_inf_theory_non_ada}

\geb{It suffices to prove the assertion of the theorem for $m = (1 - \eps)\Delta k^{1 + 1/\Delta}$, since extra tests can only help (or can be ignored).} 
Let $\eps, \theta, \delta \in (0,1)$, and $\theta/(1-\theta) \le \Delta \le \log^{1- \delta} n$.
Furthermore, let $\cG$ be an arbitrary non-adaptive pooling scheme with $V(\cG)$ the set of $n$ individuals and $F(\cG)$ the set of $m = (1-\eps) \Delta k^{1+1/\Delta}$ tests such that each individual is tested at most $\Delta$ times.
Let
\begin{equation}
    \bar \ell = \frac{1}{1 - \eps} k^{- \frac{1}{\Delta}} \qquad \text{and} \qquad \bar \Gamma = \frac{1}{m} \sum_{a \in F(\cG)} \Gamma_a = \frac{n \Delta}{m} \ge \bar \ell \frac{n}{k}. \label{eq:l_def}
\end{equation}
Thus, $\bar \Gamma$ represents the average degree of the tests in $F(\cG)$, where $\Gamma_a$ is the size of test $a$.
We pick a set of $k$ infected individuals uniformly at random and let $\SIGMA$ be the $\{0,1\}$-vector representing them.
We introduce $p = \frac{k - \sqrt{k} \log n}{n}$ and $\SIGMA^*$ as a binomial $\cbc{0,1}$-vector, such that each entry represents one individual and equals 1 with probability $p$ independently of the others.  Our next result relates $\SIGMA$ and $\SIGMA^*$.
\geb{As in \cite{Aldridge_2019,Coja_2019_2} the way to establish a lower bound is to establish that the underlying graph structure always contains a certain number of disguised infected as well as disguised uninfected individuals. We note that due to the $\Delta$-divisibility condition, a straightforward application of the FKG inequality does not appear to provide a sufficiently strong bound, since the variances of the random variables of interest may become too large.}

\begin{corollary}
\label{cor_bernoulli_ok}
Under the preceding setup, for fixed $\eps \in (0,1)$ and $n$ large enough, if there is a non-negative integer $C$ (possibly $C = 0$) such that
\begin{align*}
    &\Pr \bc{ \abs{\oneplus(\cG, \SIGMA^*)} >  2 C} \geq 1 - \eps\\
    \text{and} \qquad &\Pr \bc{ \abs{\zeroplus(\cG, \SIGMA^*)} > 2 C} \geq 1 - \eps,
\end{align*}

then it also holds that
\begin{align*}
    &\Pr \bc{ \abs{\oneplus(\cG, \SIGMA)} > C } \geq 1 - \eps - o(1)\\
     \text{and} \qquad &\Pr \bc{ \abs{\zeroplus(\cG, \SIGMA)} > C} \geq 1 - \eps - o(1).
\end{align*}
\end{corollary}
\begin{IEEEproof}
The proof follows along the lines of the proof of \cite[Lemma~3.6]{Coja_2019_2}. Let $\cB$ be the event that $|\SIGMA^*| \in [ k - 2 \sqrt{k} \log n, k ]$. Then a standard application of the Chernoff bound guarantees that $\Pr \bc{ \cB } = 1 - o(1)$. 

Given $\cB$, we couple $\SIGMA^*$ and $\SIGMA$ by flipping at most $2 \sqrt{k} \log n$ uninfected individuals in $\SIGMA^*$ to infected, uniformly at random. 
\geb{This yields the correct distribution, since by definition the set $I_1 = \cbc{ i: \SIGMA^*_i = 1 }$ is a uniform subset of size $\abs{\SIGMA^*}$ (conditioned on $\abs{\SIGMA^*}$).  Hence, when we infect another random subset of size $k - |I_1|$ uniformly at random, the overall infected set is uniform over the subsets of size $k$.}
Clearly, the number of disguised infected individuals can only increase, and hence 
\begin{equation}
    \abs{\oneplus(\cG, \SIGMA^*)} \leq \abs{\oneplus(\cG, \SIGMA)}. \label{eq:first_claim}
\end{equation}
However, it might happen that previously disguised uninfected individuals do now contribute to $\abs{\oneplus(\cG, \SIGMA)}$ instead of $\abs{\zeroplus(\cG, \SIGMA)}$. Let $$\vec V := \abs{ \abs{\zeroplus(\cG, \SIGMA)} - \abs{\zeroplus(\cG, \SIGMA^*)} }.$$ 
By the above coupling argument, we have
$$ \Erw[ \vec V \mid \cB ] \leq \frac{2 \sqrt{k} \log n}{n - k}  \abs{\zeroplus(\cG, \SIGMA^*)} < n^{ - (1 - \theta) } \abs{\zeroplus(\cG, \SIGMA^*)}. $$
Therefore, Markov's inequality implies
\begin{align}
    &\Pr \bc{ \abs{\zeroplus(\cG, \SIGMA)} \le \abs{\zeroplus(\cG, \SIGMA^*)}/2 \mid \cB } \notag\\
    &\hspace{3cm} \leq \Pr \bc{ \vec V \ge \frac{\Erw \brk{\vec V \mid \cB}}{2 n ^{-(1-\theta)}} \mid \cB } = o(1). \label{eq:second_claim}
\end{align}
The desired result now follows directly from \eqref{eq:first_claim}, \eqref{eq:second_claim}, and $\Pr \bc{ \cB } = 1 - o(1)$. 
\end{IEEEproof}

\begin{corollary}
\label{cor_bernoulli_ok2}
Under the preceding setup, we have the following:
\begin{itemize}
    \item[(i)] If $\Pr \bc{ \abs{\oneplus(\cG, \SIGMA^*)} > 0} = 1-o(1)$, then it also holds that $\Pr \bc{ \abs{\zeroplus(\cG, \SIGMA)} > \log n} = 1-o(1)$.
    \item[(ii)] If $\Pr \bc{ \abs{\oneplus(\cG, \SIGMA^*)} > 0} = \Omega(1)$, then it also holds that $\Pr \bc{ \abs{\zeroplus(\cG, \SIGMA)} > \log n} = \Omega(1)$.
\end{itemize}
\end{corollary}
\begin{IEEEproof}
\mhk{
We use the fact that the property of being disguised is independent of the infection status. Indeed, given the number of disguised individuals $\abs{V_+(\cG, \SIGMA^*)}$, we have $\abs{\oneplus(\cG, \SIGMA^*)} \sim \Bin(\abs{V_+(\cG, \SIGMA^*)}, k/n)$ and $\abs{\zeroplus(\cG, \SIGMA^*)} \sim \Bin(\abs{V_+(\cG, \SIGMA^*)}, 1- k/n)$.  Let $\delta > 0$ be such that, by assumption, $\Pr \bc{ \abs{\oneplus(\cG, \SIGMA^*)} > 0 } = 1-\delta$. 
 Therefore, 
 \begin{align}
     \label{eq_lem37_0} \delta &= \sum_{ n' = 1 }^n \Pr \bc{ \abs{V_+(\cG, \SIGMA^*)} = n' }\notag \\& \hspace{2.5cm}\cdot\Pr \bc{ \abs{\oneplus(\cG, \SIGMA^*)} = 0 \mid \abs{V_+(\cG, \SIGMA^*)} = n'}\notag \\& = \sum_{n' = 1}^n \Pr \bc{ \abs{V_+(\cG, \SIGMA^*)} = n' } \bc{ 1 - \frac{k}{n} }^{n'}. 
 \end{align}
Observe that if $n' < \frac{n}{k \log n}$, then we have $\bc{ 1 - k/n }^{n'} = 1 - o(1)$. Therefore, due to \eqref{eq_lem37_0}, we require
\begin{align}
   \label{eq_lem37_2}\sum_{ n' = 1 }^{n/(k \log n)} \Pr \bc{ \abs{V_+(\cG, \SIGMA^*)} = n' } \leq \delta + o\bc{1}
\end{align}
and we conclude that $\abs{V_+(\cG, \SIGMA^*)} = \tilde \Omega \bc{ \frac{n}{k} } = n^{\Omega(1)}$ with probability at least $1-\delta - o(1)$.  Moreover, conditioned on $\abs{V_+(\cG, \SIGMA^*)} = n^{\Omega(1)}$, the Chernoff bound yields \whp~that $\abs{\zeroplus(\cG, \SIGMA^*)} = n^{\Omega(1)} > 2\log n$.  The desired result then follows directly from \Cor~\ref{cor_bernoulli_ok}, distinguishing between $\delta = o(1)$ and $\delta \in (\Omega(1), 1-\Omega(1))$.}
\end{IEEEproof}

By adopting the two-round exposure technique from \Sec~\ref{two_round_exp}, Theorem~\ref{thm_inf_theory_non_ada} will follow from the next lemma, which establishes the conditions in Corollary~\ref{cor_bernoulli_ok2} regarding $\SIGMA^*$.
\begin{lemma} \label{lem:V1+}
For any $\eps, \theta, \delta \in (0,1)$ the following holds.  \geb{Consider the i.i.d.~infection model $\SIGMA^*$, and let $\cG$} be a test design such that any of the $n=V(\cG)$ individuals is tested at most $\Delta$ times (with $\theta/(1-\theta) < \Delta \le (\ln n)^{1-\delta}$) and $m=\abs{F(\cG)} = (1 - \eps) \Delta k^{1 + 1/\Delta}$, where $k=n^{\theta}$.  
\mhk{Then, if $\Delta = \omega(1)$ we have \whp~that $\abs{\oneplus \bc{\cG,\SIGMA^*}} > 0$, whereas if $\Delta = O(1)$, we have with $\Omega(1)$ probability that $\abs{\oneplus \bc{\cG,\SIGMA^*}} > 0$.}
\end{lemma}
\begin{IEEEproof}
\geb{
    We first give a brief overview of the proof:
    \begin{itemize}
        \item We first establish that there must be no tests in $\cG$ with too few individuals (Claim \ref{claim_maxmindegree}).
        \item Second, we apply the two-round exposure technique described in Section \ref{sec:fundamentals} to create a set $\cK_1$ of infected individuals of size roughly $\alpha k$.
        \item Third, we remove any tests that already contain two infected individuals, since individuals of $\cK_1$ are disguised if and only if they are disguised upon the removal of such tests (Fact \ref{fact_tworound}). 
        \item Next, we show that, upon applying the second stage of the two-round exposure technique to the second neighbourhood of the individuals of $\cK_1$ in the remaining graph, the probability an individual $x \in \cK_1$ being disguised is minimised in the case that its tests are disjoint (Claim \ref{claim_positive_correlation}).
        \item The preceding result is used to lower bound the average probability of being disguised by employing a hypothetical model in which all tests are mutually disjoint and therefore independent (Claim \ref{claim_AM_GM}).
        \item Finally, carefully applied concentration results are used complete the proof.
    \end{itemize}}
\noindent Proceeding more formally, we first show that $\cG$ satisfies certain degree properties, namely, there cannot be any tests that are too small.

\begin{claim}\label{claim_maxmindegree}
For any fixed integer $D$, we can assume without loss of generality (for proving Lemma \ref{lem:V1+}) that, for $n$ large enough, every test has size at least $D$.
\end{claim}
\begin{IEEEproof}[Proof of Claim~\ref{claim_maxmindegree}]
We obtain an alternative design $\cG'$ from $\cG$ by iteratively deleting a test of size less than $D$ and all individuals contained in the test, until all tests have size at least $D$.
In each step, we remove one test, between one and $D$ individuals, and at most $\Delta D$ edges.
Without loss of generality, assume that in $\cG$ there are only $o(n)$ individuals that are not contained in any tests (otherwise, the error probability would trivially tend to one).
Therefore, the test-design $\cG'$ contains at least $(1-o(1))n - m \Delta D  = (1-o(1))n$ edges, and since the individual degree is still at most $\Delta$, its number of individuals $n'=|V(\cG')|$ satisfies $n'\ge (1-o(1))n/\Delta$.  This lower bound on $n'$ along with the assumption $\Delta \le (\ln n)^{1-\delta}$ additionally imply that $\Delta \le (\ln n')^{1-\delta/2}$ when $n$ is sufficiently large.

As for the remaining number of tests $m'=|F(\cG')|$, we claim that for all large enough $n$,
\begin{equation}
    m' \le (1-\eps) \Delta n^{\theta+\theta/\Delta} - (n-n')/D \le (1-\eps/2) \Delta (n')^{\theta+\theta/\Delta}. \label{eq:desired}
\end{equation}
Indeed, the first inequality follows since $m \le (1-\eps) \Delta k^{1+1/\Delta} = (1-\eps) \Delta n^{\theta+\theta/\Delta}$ and the fact that we delete at least one test per $D$ deleted individuals.
For the second inequality, let $\zeta := \theta+\theta/\Delta$, which yields $\zeta < 1$ by our assumption $\Delta > \theta/(1-\theta)$.  Then, we distinguish two cases:
\begin{itemize}
    \item If $n-n' \ge \sqrt{n}$, then we have the following:
    \begin{align*}
        (n - n')/D &\geq \Delta (n - n')^\zeta \\ &\geq \Delta \bc{n^\zeta - (n')^\zeta}\\
        &\geq (1-\eps) \Delta \bc{n^\zeta - (n')^\zeta},
    \end{align*}
  
    where the first inequality holds for sufficiently large $n$ since $D$ is constant, $\zeta \in (0,1)$, and $\Delta$ is at most logarithmic, and the second inequality holds because the function $f(x) = x^{\zeta}$ (for $\zeta \in (0,1)$) is concave and monotone, so for any $\delta > 0$ it holds that $f(x+\delta) - f(x)$ is largest when $x=0$.  Substituting the above finding yields the desired second inequality in \eqref{eq:desired}.
    \item On the other hand, if $n-n' < \sqrt{n}$, we have the following for large enough $n$:
    \begin{align*}
        (1-\eps) \Delta n^{\zeta} &< (1-\eps) \Delta (n')^\zeta \cdot (1+\sqrt{n}/n')^\zeta \\ &\le (1-\eps) \Delta (n')^{\zeta} \cdot (1+ \sqrt{n}/n') \\ &\le (1-\eps/2) \Delta (n')^{\zeta},
    \end{align*}
    since $n-n' < \sqrt{n}$ implies that $\sqrt{n}/n' = o(1)$.  Hence, in this case we get the desired result even after trivially bounding $(n-n')/D$ by zero.
\end{itemize}
Since $\oneplus(\cG')\subseteq \oneplus(\cG)$, we can continue working with $\cG'$ and the desired claim holds.
\end{IEEEproof}

Recall that in the multi-step argument in Section \ref{two_round_exp}, for some $\alpha>0$, the first step is to infect each individual independently with probability $\alpha k/n$, and denote the resulting set of infected individuals by $\cK_1$.  We seek to characterize the number of disguised individuals in $\cK_1$ following a second step of infections, in which each previously-uninfected individual is infected with probability $(1-2 \alpha)k/n$.  Given $\cK_1$, let $\vX_v^\ast$ be the probability that $v \in \cK_1$ is disguised after this second step, and let $\vX^\ast = \sum_{v \in \cK_1} \vX_v^\ast$.
To prove that $\vX^\ast$ is large, we need the following two statements.

\begin{fact}\label{fact_tworound}
Let $a$ be a test such that $\abs{ \partial a \cap \cK_1} \geq 2$.
Then any individual in $\cK_1$ is disguised if and only if it is disguised when removing the test $a$.
\end{fact}

This fact is immediate as any infected individual is disguised in $a$ by definition.
Furthermore, to get a handle on the subtle dependencies between overlapping tests, we prove that the probability for an individual to be disguised in two tests is minimised when the tests are disjoint.
For this, denote by $\partial^{(x)} a = \partial a \setminus \cbc{x}$ the individuals in test $a$ without $x$.

\begin{claim}
\label{claim_positive_correlation} 
Consider marking each individual in $\partial^{(x)} a \cup \partial^{(x)} a'$ as infected with some probability $q$ independent of the others.
Then, for any integer $z > 0$, any individual $x \in V(\cG)$ and any two tests $a, a' \in \partial x$, we have
\begin{align*}
    \Pr  &\Bigg( \partial^{(x)} a \cap V_1(\cG) \neq \emptyset,  \partial^{(x)} a' \cap V_1(\cG) \neq \emptyset \mid \partial^{(x)} a \cap \partial^{(x)} a' = \emptyset\Bigg) \\ &\leq \Pr \Bigg( \partial^{(x)} a \cap V_1(\cG) \neq \emptyset, \\&\hspace{2.5cm}\partial^{(x)} a' \cap V_1(\cG) \neq \emptyset \mid \abs{\partial^{(x)} a \cap \partial^{(x)} a'} = z \Bigg).
\end{align*}
\end{claim}

\begin{IEEEproof}
We first note that
\begin{align}
    \label{eq_pos_cor_1} \Pr &\bc{ \partial^{(x)} a \cap V_1(\cG) \neq \emptyset,  \partial^{(x)} a' \cap V_1(\cG) \neq \emptyset \mid \partial^{(x)} a \cap \partial^{(x)} a' = \emptyset}\notag\\& = \bc{ 1 - (1 - q)^{\abs{\partial^{(x)}a }} } \bc{ 1 - (1 - q)^{\abs{\partial^{(x)} a' }} },
\end{align}
as the infected individuals in the two tests are independent due to the conditioning event.

On the other hand, suppose that $ \abs{ \partial^{(x)}a \cap \partial^{(x)} a' } = z > 0$.
In order to make both tests contain at least one infected individual that is not $x$, we can either have at least one of the $z$ common individuals which is infected (happening with probability $ \bc{1 - (1-q)^{z}} $), or we need both tests to contain an infected individual outside of the intersection.
Hence, 
\begin{align}
      \notag\Pr & \bc{ \partial^{(x)} a \cap V_1(\cG) \neq \emptyset,  \partial^{(x)} a' \cap V_1(\cG) \neq \emptyset \mid \abs{\partial^{(x)} a \cap \partial^{(x)} a'} = z } \\ 
     & = \bc{ 1 - \bc{1-q}^z } + \bc{1 - q}^z \bc{ 1 - (1 - q)^{\abs{\partial^{(x)} a } - z } }\notag\\& \hspace{3.5cm}\cdot\bc{ 1 - (1 - q)^{\abs{\partial^{(x)} a' } - z} } \label{eq_pos_cor_2}
\end{align}
Using \eqref{eq_pos_cor_1} and \eqref{eq_pos_cor_2}, we conclude the proof with a short calculation:
\begin{align*}
    &\Pr  \bc{ \partial^{(x)} a \cap V_1(\cG) \neq \emptyset,  \partial^{(x)} a' \cap V_1(\cG) \neq \emptyset \mid \abs{\partial^{(x)} a \cap \partial^{(x)} a'} = z } \\ &  - \Pr \bc{ \partial^{(x)} a \cap V_1(\cG) \neq \emptyset,  \partial^{(x)} a' \cap V_1(\cG) \neq \emptyset \mid \partial^{(x)} a \cap \partial^{(x)} a' = \emptyset} \\ 
    & = \bc{ 1 - \bc{1-q}^z } \\&\hspace{0.5cm}+ \bc{1 - q}^z \bc{ 1 - (1 - q)^{\abs{\partial^{(x)} a } - z } } \bc{ 1 - (1 - q)^{\abs{\partial^{(x)} a' } - z} } \\&\hspace{0.5cm}-  \bc{ 1 - (1 - q)^{\abs{\partial^{(x)}a }} } \bc{ 1 - (1 - q)^{\abs{\partial^{(x)} a' }} } \notag \\
    & = \bc{ 1 - \bc{1-q}^z } \bc{1-q}^{\abs{\partial^{(x)} a } + \abs{\partial^{(x)} a' } -z } \ge 0,
\end{align*}
where the last step follows by expanding and simplifying.
\end{IEEEproof}

With this in mind, we can consider a simplified model in which the test degrees are unchanged, but the {\em tests are all disjoint}.\footnote{This suggests an increase in the number of individuals, but the total number of individuals does not play a role in this part of the analysis.} More precisely, we define the following: 
Given an infection rate $q \in (0,1)$, we let $\vY_a = \vY_a(q) := \bc{ 1 - \bc{1 - q}^{\Gamma_a - 1}}$ be the probability that in a test $a$ of size $\Gamma_a$ with one fixed individual $x$, there is at least one infected individual that is not $x$.
For any individual $v$, we then denote by $\vX_v = \vX_v(q) := \prod_{a \in \partial v} \vY_a(q)$ the probability that $v$ is disguised in this model, where all tests are mutually disjoint.
Observe that, by Claim~\ref{claim_positive_correlation}, $\vX_v^\ast \ge \vX_v$, and therefore, $\vX^\ast \ge \vX$.
The advantage is that in this model, $\vX_v$ and $\vX_u$ are independent for $v \not= u$.
Recall that $$\bar \ell = \frac{1}{1 - \eps} k^{-1/\Delta} = o(1), $$ because $\Delta = O ( \ln
^{1-\delta}n)$ and $k=n^{\theta}$, and let $\ell_a = \Gamma_a k/n$.

\geb{Note that $\vec X_v$ describes the probability of being disguised for one individual; we proceed by considering the entire set of individuals.} The following lemma provides a useful lower bound on $n^{-1} \sum_{v \in V(\cG)} \vX_v$.
\begin{claim} \label{claim_AM_GM}
Under the preceding setup with $q = (1-2\alpha) k/n$, we have
\[n^{-1} \sum_{v \in V(\cG)} \vX_v \geq (1 - \exp(- (1-3\alpha) \bar \ell)) ^ \Delta.\]
\end{claim}
\begin{IEEEproof}
By the inequality of arithmetic and geometric means, we have
\begin{align}\label{eq_amgm_1}
   n^{-1} \sum_{v \in V(\cG)} \vX_v \ge \prod_{v \in V(\cG)} \bc{ \prod_{a \in \partial v} \vY_a}^{1/n} = \prod_{a \in F(\cG)} \vY_a^{\Gamma_a / n}.
\end{align}
Furthermore, by Claim \ref{claim_maxmindegree}, we may assume that $\Gamma_a \ge  (3 \alpha)^{-1}$, and we deduce that
\[\vY_a \ge 1 - \exp \bc{ -q \bc{\Gamma_a-1}} \ge 1 - \exp \bc{- (1-3\alpha)\ell_a}.\]
Hence, \eqref{eq_amgm_1} yields
\begin{align}\label{eq_amgm_2}
   n^{-1} \sum_{v \in V(\cG)} \vX_v \geq \prod_{a \in F(\cG)} \bc{1 - \exp(- (1-3\alpha)\ell_a)}^{ \ell_a / k}.
\end{align}
\geb{Next, we note that $\sum_{a \in F(\cG)} \Gamma_a \le \Delta n$ by the $\Delta$-divisibility constraint, which further implies $\sum_{a \in F(\cG)} \ell_a \le k\Delta$.  The choice $m = (1 - \eps) \Delta k^{1 + 1/\Delta}$ also implies $\bar \ell = \Delta k m^{-1}$, and we can characterise the logarithm of the right-hand side of \eqref{eq_amgm_2} as follows:
{\small
\begin{equation} 
\begin{split}
    \label{eq_amgm_3}
    k^{-1} & \sum_{a \in F(\cG)} \ell_a \log \bc{ 1 - \exp( - (1-3\alpha) \ell_a ) } \\
    & = m k^{-1} \sum_{a \in F(\cG)} m^{-1} \bc{\ell_a \log \bc{ 1 - \exp( - (1-3\alpha) \ell_a ) }} \\
    &\geq m k^{-1} \bc{\sum_{a \in F(\cG)} m^{-1} \ell_a} \log \bc{ 1 - \exp \bc{ - (1-3\alpha) m^{-1} \sum_{a \in F(\cG)} \ell_a} } \\
    & \geq \Delta \log \bc{ 1 - \exp \bc{- (1-3\alpha) \bar \ell } },
\end{split}
\end{equation}
}
where the first inequality applies Jensen's inequality applied to the convex function $f(x) = x \log (1 - \exp(-(1-3\alpha)x))$ on $(0,1)$, and the second inequality uses $\bar \ell \ge m^{-1} \sum_{a \in F(\cG)} \ell_a$ (by the above calculations regarding $\bar \ell$ and $\ell_a$ above), along with the fact that $\bar \ell = \Delta k m^{-1} = o(1)$ and $f(x)$ is a decreasing function for small enough $x$.  Finally, the assertion of the claim follows from~\eqref{eq_amgm_2} and~\eqref{eq_amgm_3}.}
\end{IEEEproof}

We note from this claim that if we let $\vv$ be a uniformly random individual, we have (also using $\bar \ell = o(1)$) that
\begin{align*}
   \Erw \brk{ \vX_{\vv} } & \geq \bc{ 1 - \exp \bc{- (1-3\alpha) \bar \ell} }^\Delta \ge (1-4\alpha)^{\Delta} \bar \ell^{\Delta} \\
   & = \frac{(1-4\alpha)^{\Delta}}{(1-\eps)^{\Delta} k} \ge (1-\eps/2)^{-\Delta} k^{-1}, 
\end{align*}
provided that $\alpha \le \eps/8$.

Now, recall that $\vX = \sum_{v \in \cK_1} \vX_v$, and that each individual is in $\cK_1$ with probability $\alpha k/n$.
Then we deduce from the above that
\[ \Erw[\vX] = \alpha k\, \Erw[\vX_{\vv}] \ge \alpha (1-\eps/2)^{-\Delta}.\]
As $\vX_v$ and $\vX_u$ are independent for $v \not=u$, we can apply the Chernoff bound (Lemma~\ref{Lem_Chernoff}, or more precisely a one-sided version that saves a factor of 2) to obtain
\begin{equation}
   \Pr \bc{ \vX < \alpha (1- \eps/2)^{-\Delta}/2} \le \exp (-\alpha (1- \eps/2)^{-\Delta}/12). \label{Xbound} 
\end{equation}

Now, as described earlier, consider infecting any uninfected individual with probability $q=(1-2\alpha) k/n$ independent of all the others.
Then, as $\sum_{v \in \cK_1} \Pr (v \in \oneplus(\cG)) = \vX^\ast \ge \vX$, we find that conditioned on $\cK_1$ and $\vX$, it holds with probability at least 
 \begin{align*}
     1 - \prod_{v \in \cK_1} (1-\Pr (v \in \oneplus(\cG)))
 \ge 1 - \bc{1-\frac{\vX}{|\cK_1|}}^{|\cK_1|}
 \ge \frac{\vX}{1+\vX}
 \end{align*}
 that at least one individual from $\cK_1$ is disguised.
 Here we used the inequality of arithmetic and geometric means to upper bound the product, and the last step \geb{uses Bernoulli's inequality to write $(1-x/c)^c \le 1-x \le \frac{1}{1+x}$.}
 With $\alpha=\eps/8$ and the upper bound \eqref{Xbound} on the probability that $\vX < \eps (1- \eps/2)^{-\Delta}/16$, it follows that there exists a disguised individual in $\cK_1$ with probability at least
 \[ \label{eq:lemV1+} \nu = \nu(\Delta, \eps) := (1-\exp(-\eps (1- \eps/2)^{-\Delta}/96)) \] 
which yields the statement of Lemma \ref{lem:V1+}; \mhk{note that $\nu = 1 - o(1)$ when $\Delta = \omega(1)$, and that $\nu = \Omega(1)$ when $\Delta = O(1)$.  The latter assertion holds via the Taylor expansion $1 - \exp(-x) = x + \Theta(x^2)$ as $x \to 0$.}

Recall that $p = \frac{k - \sqrt{k} \log n}{n}$, and note that any individual is infected with probability at most 
\[\tilde p = \alpha k/n + (1-\alpha k/n)(1-2 \alpha)k/n < p,\]
independent of all the others.
As discussed in Section~\ref{two_round_exp} we can in hindsight raise the infection probability of each individual to $p$, which can only increase the size of the set $\oneplus(\cG)$ \geb{(i.e., the number of disguised infected individuals)}.  This yields the assertion of Lemma \ref{lem:V1+} for the i.i.d.~infection model.
\end{IEEEproof}

\begin{IEEEproof}[Proof of Theorem~\ref{thm_inf_theory_non_ada}]
The theorem now follows easily by combining Lemma \ref{lem:V1+} with Corollary \ref{cor_bernoulli_ok2}:  With at least one disguised infected individual and at least $\log n$ disguised uninfected individuals, the conditional error probability is $1-o(1)$ due to Claim \ref{Claim_1+0+}. 

\end{IEEEproof}

\subsection{Algorithmic achievability on the random regular model: Proof of Theorem \ref{thm_DD_achievability_delta}}\label{proof_thm_DD_achievability_delta}

\subsubsection{Further notation} \label{sec_further_not}
Recall the random regular model $\cGd$ from \Sec~\ref{pool_delta}.  We let $(\vec \Gamma_1, \dots, \vec \Gamma_m)$ be the (random) sequence of test-degrees, which satisfies the following by construction:
\begin{align}\label{Deg_eq}
    \sum_{i=1}^{m}\vec \Gamma_i=n\Delta.
\end{align}
Furthermore, given the sequence $(\vec \Gamma_i)_{i \in [m]}$, we define
$$ \Gamma_{\min} = \min_{i \in [m]} \vec \Gamma_i, \qquad \bar \Gamma = \frac{1}{m} \sum_{i=1}^m \vec \Gamma_i = \frac{n\Delta}{m} $$
and $$ \Gamma_{\max} = \max_{i \in [m]} \vec \Gamma_i.$$ We stress at this point that the construction of $\cGd$ allows for multi-edges, and hence one individual might take part in a test multiple times and contribute more than one to its degree.

Moreover, we parametrise the average degree as $\bar \Gamma = \ell n/k$, such that $\ell$ denotes the expected number of infected individuals a test would contain in a binomial random bipartite graph.  The definition of $\bar \Gamma$ implies $\ell = \frac{k\Delta}{m}$, and substituting $m = (1+\epsilon)\mDD$ yields
\begin{equation}
    \ell = (1+\epsilon)^{-1} \bc{ \min \cbc{n^{-(1-\theta)/\Delta}, n^{- \theta/\Delta}} }. \label{eq:l_choice}
\end{equation}
Note that with $\frac{\theta}{1 - \theta} < \Delta \le  (\log n)^{1-\Omega(1)}$, we have $\omega(n^{1-\theta}) \leq \ell \leq o(1)$.  \geb{We will make use of a stronger version of the left inequality stating that $\frac{\ell}{n^{1-\theta}} \ge n^{\Omega(1)}$, which follows from $\Delta > \frac{\theta}{1 - \theta}$ and checking both cases of which term in \eqref{eq:l_choice} attains the minimum.}

We first argue that each test degree is tightly concentrated with high probability, defining the concentration event $\cC_{\Gamma}$ as follows:
\begin{align}
    \cC_{\Gamma} = \Big\{ \bc{ 1 - O\bc{ n^{-\Omega(1)}  } } \frac{\ell n }{k} & \leq \Gamma_{\min} \leq \bar \Gamma  \leq \Gamma_{\max} \notag \\ 
    & \leq \bc{ 1 + O\bc{ n^{-\Omega(1)}  } } \frac{\ell n }{k} \Big\}. \label{Eq_Event_Gamma}
\end{align}

\begin{lemma} \label{lem_conc_gamma_in_gdelta} For $\ell$ given in \eqref{eq:l_choice}, we have $\Pr(\cC_{\Gamma})=1-\tilde O(n^{-3})$.
\end{lemma}

\begin{IEEEproof}
Each individual chooses $\Delta$ tests with replacement. Hence, each individual has a chance of picking a given test $\Delta$ times independently, yielding 
$$\vec \Gamma_i = \sum_{j=1}^n \sum_{h=1}^{\Delta} \vecone \cbc{x_j \text{ chooses } a_i \text{ in  } h\text{-th selection}}$$ 
and $$\vec \Gamma_i \sim \Bin \bc{ n \Delta, 1 / m }.$$
Thus, we have $\Erw \brk{ \vec \Gamma_i } = \ell n / k$, which scales as $\omega(1)$ since we have established $\ell \ge \omega(n^{1-\theta})$.

Applying the Chernoff bound (\Lem~\ref{Lem_Chernoff}) \geb{and the above-established fact $\frac{\ell}{n^{1-\theta}} \ge n^{\Omega(1)}$}, we obtain
$$ \Pr \bc{ \vec \Gamma_i < (1 - t) \ell n / k} \leq \exp \bc{ - t^2 \ell n^{1-\theta} / 3 } \leq \exp \bc{ - \Omega \bc{ t^2 n^{ \Omega(1) } } }.$$ 
Hence, we can choose $t$ of the form $O(n^{-\Omega(1)} \log n) = O(n^{-\Omega(1)})$ to attain
\begin{align}
    \label{eq_conc_i} \Pr \bc{ \vec \Gamma_i < (1 - t) \ell n / k} = \tilde O( n^{-4} ).
\end{align}
An analogous calculation shows
\begin{align}
    \label{eq_conc_ii} \Pr \bc{ \vec \Gamma_i > (1 + t) \ell n / k} = \tilde O( n^{-4} ).
\end{align}
Therefore, the lemma follows from \eqref{eq_conc_i}, \eqref{eq_conc_ii}, and a union bound over all $m \leq n$ tests. 
\end{IEEEproof}

\subsubsection{Analysis of the different types of individuals} \label{sec_types_sizes_delta}
Let $\vY_i$ denote the number of infected individuals (including all multi-edges) in test $a_i$ (for $i =1 \dots m$).  These variables are not mutually independent, as a single individual takes part in multiple tests. Luckily, it turns out that the family of the $\vY_i$ can be approximated by a family of mutually independent random variables sufficiently well. Given $\vec \Gamma_1 \dots \vec \Gamma_m$, let $(\vX_i)_{i \in [m]}$ be a sequence of mutually independent $\Bin \bc{ \vec \Gamma_i, k/n }$ variables. Furthermore, let 
\begin{align}\label{XapproxY}
   \cE_{\Delta} = \cbc{ \sum_{i=1}^m \vX_i = k \Delta}  
\end{align}
be the event that the sequence $(\vX_i)$ renders the correct number of infected individuals.  Stirling's approximation (\Lem~\ref{Lem_Stirling_Approx}) guarantees that 
$\cE_{\Delta}$ is not too unlikely; specifically, $\pr\bc{\cE_{\Delta} \mid (\vec \Gamma_i)_i} = \Omega( (n \Delta)^{-{1/2}} )$.
Furthermore, the $\vX_i$ are indeed a good local approximation to the correct distribution, as stated in the following known result.

\begin{lemma}\label{lem_x_y_delta}
{\em \cite[Appendix B.2]{Coja_2019}}
Conditioned on $(\vec \Gamma_i)_i$ and $\cE_{\Delta}$, the sequences $(\vY_i)_{i \in [m]}$ and $(\vX_i)_{i \in [m]}$ are identically distributed. \hfill $\blacksquare$
\end{lemma}
Next, we establish that the number of negative tests $\vm_0 = \vm_0(\cGd, \SIGMA)$ and the number of positive tests $\vm_1 = m - \vm_0$ are highly concentrated.
\begin{lemma}
\label{lem_m0_gdelta}
With probability at least $1 - o(n^{-2})$ we have
$$ \vm_0 = \bc{1 + {O \bc{ n^{-\Omega(1)} }}} m \exp(- \ell)$$ 
and $$\vm_1 = \bc{1 + O \bc{n^{-\Omega(1)} }} m \bc{1 -\exp(- \ell)}.$$
\end{lemma}
\begin{IEEEproof}
Let $\vm_0'=\abs{\left\{(\vX_i)_{i \in [m]}:\vX_i=0\right\}}$. Combining the definition of $\vX_i$ with \eqref{Prob_nr_inf_cor}, we get
\begin{align*} 
\Erw \brk{\vm_0' \mid  (\vec \Gamma_i)_i} = \sum_{i=1}^m \Pr \bc{\vX_i = 0 \mid \vec \Gamma_i} = \sum_{i = 1}^m \bc{ 1 - k/n }^{\vec \Gamma_i},
\end{align*}
which represents the expected number of negative tests approximated through $(\vX_i)_i$.  Hence, when $(\vec \Gamma_i)_i$ satisfies the concentration event defining $\cC_{\Gamma}$ (see \eqref{Eq_Event_Gamma}), a second order Taylor expansion (\Lem~\ref{Bernoulli}) yields
\begin{align}
    \label{Eq_m0_1} \Erw \brk{\vm_0' \mid (\vec \Gamma_i)_i} = \bc{1 + {O \bc{n^{-\Omega(1)}}}} m \exp(- \ell).
\end{align}
\geb{Then, conditioned on $(\vec \Gamma_i)_i$, the Chernoff bound implies implies with probability at least $1 - o(n^{-10})$ that
\begin{equation}
    \vm_0' =  \Erw \brk{\vm_0' \mid (\vec \Gamma_i)_i} \bc{1 + O(m^{-1/4})}.  \label{Eq_m0_2}
\end{equation}}
The first assertion of the lemma now follows from \eqref{Eq_m0_1}, \eqref{Eq_m0_2}, \Lem~\ref{lem_conc_gamma_in_gdelta}, \Lem~\ref{lem_x_y_delta}, and the fact that 
$\cE_{\Delta}$ has probability $\Omega( (n \Delta)^{-{1/2}} )$: \geb{Letting $\cA$ be the above probability-$o(n^{-10})$ event, we simply write $\Pr(\cA|\cE_{\Delta}) \le \frac{\Pr(\cA)}{\Pr(\cE_{\Delta})}$, and substitute the upper bound on the numerator and lower bound on the denominator.}

\geb{For the second assertion of the lemma, we need to additionally take note of the fact that $\ell = o(1)$ and hence $m(1-e^{-\ell}) = O(m \ell) \ll m$.  But since $m\ell = k\Delta$, this only amounts to replacing $m^{-1/4}$ by $k^{-1/4}$ in the counterpart of \eqref{Eq_m0_2}, and otherwise has no impact.}
\end{IEEEproof}

Next, we provide a characterization of the size of $\zeroplus(\cGd)$, i.e., the number of disguised uninfected individuals.
\begin{lemma}
\label{lem_size_0+_delta}We have with probability at least $1 - O \bc{ n^{-\Omega(1)} }$ that $$ \abs{\zeroplus(\cGd)} = \bc{1 + {O\bc{n^{-\Omega(1)}}}} n \bc{1 - \exp(-\ell)}^{\Delta}.$$
\end{lemma}
\begin{IEEEproof}
Without loss of generality, given $\vm_1$ and $\cC_{\Gamma}$, we suppose that tests $a_1 \dots a_{\vm_1}$ are the positive tests. By the degree bounds in \eqref{Eq_Event_Gamma} and \Lem~\ref{lem_m0_gdelta}, the total number of edges connected to a positive test is \whp~ given by
\begin{align}
     \sum_{i=1}^{\vm_1} \vec \Gamma_i = \bc{1 + O{\bc{ n^{-\Omega(1)} }}} m \bar \Gamma \bc{1 - \exp(- \ell)}. \label{Eq_pos_halfedges_delta}
\end{align}
We need to calculate the probability that a given uninfected individual belongs to $\zeroplus\bc{\cGd}$, i.e., each of its $\Delta$ edges is connected to a positive test.  By a counting argument, we have
\begin{align*}
\Pr_{\cGd} & \bc{ x \in \zeroplus(\cGd) \mid x \in \zero(\cGd),  \vm_1, \cC_\Gamma, (\vec \Gamma_i)_i} \\
& =  \binom{\sum_{i=1}^{\vm_1} \vec \Gamma_i}{\Delta} \binom{\sum_{i=1}^{m} \vec \Gamma_i }{\Delta}^{-1} \\
& = \bc{1 + {O \bc{ n^{-\Omega(1)} }} }\bc{1 - \exp \bc{- \ell}}^{\Delta},\end{align*}
where the simplification follows via Claim~\ref{stirling_applied} along with \eqref{Eq_pos_halfedges_delta} and $\sum_{i=1}^{m} \vec \Gamma_i = m\bar{\Gamma}$. 

Therefore,
\begin{align}
    \Erw_{\cGd} \brk{ \abs{\zeroplus(\cGd)} \mid \cC_{\Gamma}} &= \bc{1 + {O \bc{ n^{-\Omega(1)} }}} n \bc{1 - \exp(-\ell)}^{\Delta}. \label{Eq_FirstMoment_0+}
\end{align}
Analogously, the second moment turns out to be
\begin{align}
 \Erw_{ \cGd} & \brk{ \abs{\zeroplus(\cGd)}^2 \mid \cC_{\Gamma}} \le \frac {\binom{n-k}{2} 
\binom{\bc{1 + {O \bc{ n^{-\Omega(1)} }}} m \bar \Gamma \bc{1 - \exp(- \ell)}}{2\Delta}} {\binom{\bc{1 + {O \bc{ n^{-\Omega(1)} }}} m \bar \Gamma }{2\Delta}} \notag \\
& = \bc{1 + {O \bc{ n^{-\Omega(1)} }}} n^2 \bc{1 - \exp(-\ell)}^{2 \Delta}. \label{Eq_SecondMoment_0+}
\end{align}
\geb{The idea of the first line of \eqref{Eq_SecondMoment_0+} is to consider pairs of uninfected individuals whose $2\Delta$ combined edges only participate in positive tests.\footnote{\geb{The contribution of ``self-pairs'' where a individual just chooses its own $\Delta$ edges from the corresponding set is strictly smaller, which is why the expression given is an upper bound rather than an equality.}} The second line of \eqref{Eq_SecondMoment_0+} follows from Stirling's approximation in the form of Claim~\ref{stirling_applied}.
We lemma is now obtained using \eqref{Eq_FirstMoment_0+}, \eqref{Eq_SecondMoment_0+}, and Chebyshev's inequality, and noting that $n(1-e^{-\ell})^{\Delta} = n^{-\Omega(1)}$ (which is seen by using $\ell = o(1)$ to approximate $(1-e^{-\ell})^{\Delta}$ by $\ell^{\Delta}$, and applying \eqref{eq:l_choice}).}
\end{IEEEproof}

Let $\vec{A}$ denote the number of infected individuals that do not belong to the easy uninfected set  $\oneminusminus(\cGd)$. The following lemma allows us to bound its size.
\begin{lemma} \label{lem_size_1--_delta}If $m = (1 + \eps) \mDD(\Delta)$, then $\vec A = 0$ with probability at least $1 -  (1 + \eps)^{- \Delta}(1+o(1)) - O(n^{-\Omega(1)})$.
\end{lemma}

\begin{IEEEproof}
We can split \eqref{eq:l_choice} into two cases, depending on the sparsity level $\theta$:
\begin{equation}
    \ell = \begin{cases} (1+\eps)^{-1}n^{-(1 - \theta)/\Delta}, & \text{if} \quad \theta \leq 1/2 \\
(1+\eps)^{-1}k^{- 1 / \Delta}, & \text{if} \quad \theta > 1/2.
\end{cases} \label{eq:ell_cases}
\end{equation}
Recall that $\vm_1$ is the number of positive tests, and define
\begin{align}
    \cF_{\Delta} &= \cbc{\vm_1 = \bc{1 + O\bc{n^{-\Omega(1)}}} m \bc{ 1 - \exp \bc{- \ell}}} \notag \\ & \qquad \cap \cbc{\abs{\zeroplus(\cGd)} = \bc{1 + O\bc{n^{-\Omega(1)} }} n \bc{ 1 - \exp(-\ell)}^\Delta} \label{Eq_F_delta}
\end{align}
as the event that both the number of positive tests as well as the size of $\zeroplus(\cGd)$ behave as expected. \Lem s  \ref{lem_m0_gdelta} and \ref{lem_size_0+_delta} guarantee that $\cF_{\Delta}$ is a high probability event, namely, $\Pr \cbc{ \cF_{\Delta} } \geq 1 - \tilde O(n^{-1})$. Given $\vm_1$, we suppose without loss of generality that $a_{1} \dots a_{\vm_1}$ are the tests rendering a positive result.

We describe the number of occurrences of different types of individuals by introducing two sequences of random variables. Define $\vR_i = (\vR^1_i, \vR^{0+}_i, \vR^{0-}_i)_{i \in [\vm_1]}$ as the number of infected individuals, disguised uninfected individuals of $\zeroplus(\cGd)$, and non-disguised uninfected individuals (those of $\zerominus(\cGd)$) appearing in test $i$, respectively.  By construction, we have $\vR_i^{0-} = \Gamma_i - \vR_i^{0+} - \vR_i^{1}$. 

Given $\abs{\zeroplus(\cGd)}$ and $\vm_1$, we approximate these variables by a sequence of mutually independent multinomials. Specifically, let 
\begin{align}\label{Multinom}
\vH_i &= (\vH^1_i, \vH^{0+}_i, \vH^{0-}_i)_{i \in [\vm_1]} \\ 
&\stackrel{{\rm i.i.d.}}{\sim} \Mult_{\geq (1,0,0)} \bc{ \vec \Gamma_i, \bc{ \frac{k}{n}, \frac{\abs{\zeroplus(\cGd)}}{n}, 1 - \frac{k +  \abs{\zeroplus(\cGd)}}{n} } } \notag,
\end{align}
where $\Mult_{\geq (1,0,0)}$ means multinomial  conditioned on the first coordinate being at least one.
We introduce the event $$\cD_{\Delta} = \cbc{ \sum_{i=1}^{\vm_1} \vH_i^{1} = k \Delta, \quad \sum_{i=1}^{\vm_1} \vH_i^{0+} = \abs{\zeroplus(\cGd)} \Delta }, $$ 
and make use of the following. 
\begin{claim}\label{Cor_DistEqual}
Given $\vec (\Gamma_i)_i$, $\abs{\zeroplus(\cGd)}$, and $\vm_1$, the distribution of $\vR_i$ equals the distribution of $\vH_i$ given $\cD_{\Delta}$. Furthermore, $\Pr \bc{\cD_{\Delta}} \geq \Omega(n^{-2})$. 
\end{claim}
\begin{IEEEproof}  Let $(r_i)_{i\in [\vm_1]}$ be a sequence with $r_i=(r_i^{1},r_i^{0+},r_i^{0-})$ satisfying $$S_1 := \sum_{i=1}^{\vm_1} r_i^{1} = k \Delta, \qquad S_{0+} := \sum_{i=1}^{\vm_1} r_i^{0+} = \abs{\zeroplus(\cGd)} \Delta$$ and  $$r_i^{0-} = \Gamma_i - r_i^{1} - r_i^{0+}. $$ 
In addition, let $$ S_{0-} := \sum_{i=1}^{\vm_1} r_i^{0-}$$ denote the number of connections from individuals in $\zerominus(\cGd)$ to positive tests.
Then, a counting argument gives
\begin{align*}
    \Pr_{\cGd} & (\forall i\in [\vm_1]:\vR_i=r_i\mid \vec (\Gamma_i)_i,\abs{ \zeroplus( \cGd ) },\vm_1) \\ &=\frac{\binom{S_1}{r_1^1...r_{\vm_1}^1}\binom{S_{0+}}{r_1^{0+}...r_{\vm_1^{0+}}}\binom{S_{0-}}{r_1^{0-}...r_{\vm_1}^{0-}}}{\binom{S_1 + S_{0+} + S_{0-}}{\Gamma_1,...,\Gamma_{\vm_1}}} \\ &=\binom{S_1 + S_{0+} + S_{0-}}{S_1,S_{0+},S_{0-}}^{-1}\prod_{i=1}^{\vm_1}\binom{\Gamma_i}{r_i^{1},r_i^{0+},r_i^{0-}}.
\end{align*}
Letting $(r^{'}_i)_{i\in [\vm_1]}$ be a second sequence as above, it follows that
\begin{align}\label{Eq_ratioYd}&\frac{\Pr_{\cGd}(\forall i\in [\vm_1]:\vR_i=r_i\mid \vec (\Gamma_i)_i,\abs{ \zeroplus( \cGd ) },\vm_1)}{\Pr_{\cGd}(\forall i\in [\vm_1]:\vR_i=r^{'}_i\mid \vec (\Gamma_i)_i,\abs{ \zeroplus( \cGd ) },\vm_1)}\\&=\prod_{i=1}^{\vm_1}\frac{\binom{\Gamma_i}{r^{1}_i r^{0+}_i r^{0-}_i}}{\binom{\Gamma_i}{r^{1'}_i r^{0+'}_i r^{0-'}_i}}.\end{align}
Next, define  $$ R_1 = \sum_{i=1}^{\vm_1} r_i^1, \qquad R_+ = \sum_{i=1}^{\vm_1} r_i^{0+}, \qquad \text{and} \qquad R_- = \sum_{i=1}^{\vm_1} r_i^{0-} $$ and analogously for $R'_1, R'_+, R'_-$. By definition, we have $$ R_1 = R'_1, \qquad R_+ = R'_+ \qquad \text{and} \qquad R_- = R'_-.$$
Then, by the definition of $\vH$, we have
\begin{align}
    & \frac{\Pr_{\cGd}(\forall i\in [\vm_1]:\vH_i=r_i\mid (\vec \Gamma_i)_i,\abs{ \zeroplus( \cGd ) },\vm_1,\cD_{\Delta})}{\Pr_{\cGd}(\forall i\in [\vm_1]:\vH_i=r^{'}_i\mid (\vec \Gamma_i)_i,\abs{ \zeroplus( \cGd ) },\vm_1,\cD_{\Delta})} \notag \\
    & = \frac{(k/n)^{R_1} (\abs{\zeroplus(\cGd)}/n)^{R_+} (1 - k/n - \abs{\zeroplus(\cGd)}/n )^{R_-} }{(k/n)^{R'_1} (\abs{\zeroplus(\cGd)}/n)^{R'_+} (1 - k/n - \abs{\zeroplus(\cGd)}/n )^{R'_-} } \notag \\  
    & \qquad \cdot \prod_{i=1}^{\vm_1} \frac{ \binom{\Gamma_i}{ r_i^{1},r_i^{0+},r_i^{0-} } }{ \binom{\Gamma_i}{r_i^{' 1},r_i^{' 0+}, r_i^{' 0-}}} = \prod_{i=1}^{\vm_1} \frac{\binom{\Gamma_i}{r_i^{1},r_i^{0+},r_i^{0-}}}{\binom{\Gamma_i}{r_i^{' 1},r_i^{' 0+}, r_i^{' 0-}}} \label{Eq_Xd_ratio}.
\end{align}
Thus, \geb{the first statement of} Claim~\ref{Cor_DistEqual} follows from \eqref{Eq_ratioYd} and \eqref{Eq_Xd_ratio}, and \geb{the second statement follows from Claim~\ref{Prob_nr_inf_cor}}
\end{IEEEproof}

We now introduce a random variable that counts (positive) tests featuring only one infected individual and no disguised uninfected individuals. Formally, let
\begin{align} \label{def_B}
   &\vec B = \sum_{i=1}^{\vm_1} \vecone \cbc{ \vR_i^{1} + \vR_i^{0+} = 1} \notag \\ \text{and} \qquad &\vec B' = \sum_{i=1}^{\vm_1} \vecone \cbc{ \vH_i^{1} + \vH_i^{0+} = 1}.
\end{align}
By the definition of $\vH_i$ (see \eqref{Multinom}), we have
\begin{align}
    \Erw_{\cGd} &\brk{\vec B' \mid (\vec \Gamma_i)_i,\abs{ \zeroplus( \cGd ) },\vm_1  } \notag \\
    &= \sum_{i=1}^{\vm_1} \binom{\vec{\Gamma_i}}{1,0,\Gamma_i-1}  \frac{k/n (1 - k/n - \abs{\zeroplus(\cGd)}/n)^{\vec \Gamma_i - 1} }{1 - \bc{1 - k/n}^{\vec \Gamma_i}}.  \label{Eq_ErwA_delta1}
\end{align}
In the following, we suppose that $\vec \Gamma_i$ satisfies the concentration around $\bar \Gamma$ defining event $\cC_{\Gamma}$ (see \eqref{Eq_Event_Gamma}), and $\vm_1$ and $\abs{ \zeroplus( \cGd ) }$ satisfy the concentration defining event $\cF_{\Delta}$ (see \eqref{Eq_F_delta}).  Using the concentration of $\vec \Gamma_i$ and the asymptotic expansion $(1 - k/n)^{\bar \Gamma} = \exp(- \ell (1+O(n^{-\Omega(1)})) )$, we find that
\begin{align}
    \notag \Erw_{\cGd} & \brk{\vec B' \mid  \cC_{\Gamma},\abs{ \zeroplus( \cGd ) },\vm_1} \\
    \notag &= \bc{1 + O \bc{ n^{-\Omega(1)} }} \vm_1 \bar \Gamma \notag\\ &\hspace{2.5cm}\frac{n^{-(1 - \theta)} (1 - n^{-(1 - \theta)} - \abs{\zeroplus(\cGd)}/n)^{\bar \Gamma} }{1 - \exp(- \ell (1+O(n^{-\Omega(1)})))},
\end{align}
and further applying $\bar\Gamma = \frac{n\Delta}{m}$, $k = n^{\theta}$, and the concentration of $\vm_1$, we obtain
\begin{align}
    \Erw_{ \cGd} & \brk{\vec B' \mid \cC_{\Gamma},\abs{ \zeroplus( \cGd ) },\vm_1} \notag \\
    &= \bc{1 + O\bc{n^{-\Omega(1)}}} k \Delta \bc{1 - \frac{k + \abs{\zeroplus(\cGd)}}{n}} ^{\bar \Gamma} .  \label{Eq_ErwA_delta2}
\end{align}
Now, let us distinguish between the cases $\theta \leq 1/2$ and $\theta > 1/2.$
\\ 
\noindent \textbf{Case 1: $\theta > 1/2$}:
In this case, we have $n/k = o(k)$, and $\ell = {(1+\eps)^{-1}} k^{- 1 / \Delta}$. We recall the event $\cF_{\Delta}$ from \eqref{Eq_F_delta} that gives a concentration condition for $\abs{\zeroplus(\cGd)}$ and $\vm_1$. Substituting $\ell$ into \eqref{Eq_F_delta}, we find that given $\cF_{\Delta}$, there is some $\gamma \in(0,1)$ such that 
    $$\abs{ \zeroplus( \cGd ) } = \Theta\big( (1 + \eps)^{-\Delta} n / k \big) = O \bc{k^{1 - \gamma}}. $$ 
Hence, using \eqref{Eq_ErwA_delta2} and applying $\bar \Gamma = \ell n/k$, we obtain
{\small
\begin{align}
    \notag \Erw_{\cGd} \brk{\vec B' \mid  \cC_{\Gamma}, \cF_{\Delta}} &= \bc{1 + O{{\bc{n^{-\Omega(1)}}}}} k \Delta \bc{1 - \frac{ \bc{1 + {O\bc{n^{-\Omega(1)}}}} k}{n}} ^{\bar \Gamma} \\
    &= \bc{1 + O{{\bc{n^{-\Omega(1)}}}}} k \Delta \exp \bc{- \ell} \notag \\
    &= \bc{1+O{{\bc{n^{-\Omega(1)}}}}}k\Delta (1-\ell + O(\ell^2)), \label{size_B'-case1}
\end{align}
}
by a second-order Taylor expansion of $e^{-\ell}$.
Now, $\vec B'$ is  a binomial random variable with a random number of trials and a random probability parameter. Clearly, when conditioning on a specific number of trials and a specific probability, $\vec B'$ is a binomial random variable. 
Therefore, recalling the expression for $\ell$ in \eqref{eq:ell_cases}, the Chernoff bound guarantees that under the concentration events $\cC_{\Gamma}$ and  $\cF_{\Delta}$, we have 
    $$ \vec B'= \bc{1 + O\bc{ n^{-\Omega(1)} }} \Delta k \cdot (1-(1+\eps)^{-1} k^{- 1/\Delta} + O(k^{-2/\Delta})) $$ 
with probability at least $o(n^{-10})$.  \geb{Then, similar to the proof of Lemma \ref{lem_m0_gdelta},} Claim~\ref{Cor_DistEqual} yields that
\begin{align}\label{eq_sizeB_2}
    \vec B = \bc{1 + O\bc{ n^{-\Omega(1)} }} \Delta k \cdot (1-(1+\eps)^{-1} k^{- 1/\Delta} + O(k^{-2/\Delta})) 
\end{align}
with probability $1-O(n^{-\Omega(1)})$.
Thus, we can calculate the probability of an infected individual not belonging to $\oneminusminus(\cG)$ \geb{(i.e., not being in the easily-identified infected set)} as follows. Such an individual has to choose all of its $\Delta$ edges out of the $k\Delta - \vec B$ edges that would lead to a test in which the individual could be identified by \DD.  Hence, we have
\begin{align}
    \notag \Pr \bc{ x \not \in \oneminusminus(\cGd) \mid x \in \one(\cG), \vec B} &= \binom{k \Delta - \vec B}{\Delta} \binom{k\Delta}{\Delta}^{-1}  \\ & = \bc{1+o(1)} \bc{ (1+\eps)^{-1} k^{-1/\Delta} }^{\Delta}, \label{Bcomb}
\end{align}
where the simplification holds using \eqref{eq_sizeB_2} and Claim~\ref{stirling_applied}.\footnote{The $O(k^{-2 / \Delta})$ term in \eqref{eq_sizeB_2} amounts to multiplying by $(1+O(k^{-1/\Delta}))^{\Delta}$ in \eqref{Bcomb}.  This simplifies to $1+o(1)$, since $k^{1/\Delta} = \omega(\Delta)$ due to our assumptions $\Delta \le (\log n)^{1-\Omega(1)}$ and $k = n^{\theta}$ (this is verified by comparing the logarithms).}
Interpreting the average of $\vec A$ as a sum of $k$ probabilities, it follows that
\begin{align}
    \label{Erw_A_case1} \Erw_{\cGd} \brk{\vec A \mid \vec B}  \leq \geb{\bc{1+o(1)}} {(1+\eps)^{-\Delta}}.
\end{align}
\\
\noindent \textbf{Case 2: $\theta \leq 1/2$}:
In this case, we have $\ell = (1+\eps)^{-1} n^{- (1 - \theta) / \Delta}$. Hence, given $\cF_{\Delta}$,
\begin{align}
    \label{size_vo+_m1_case2}&\abs{\zeroplus(\cGd)} = \bc{1 - O\bc{n^{-\Omega(1)}}} k (1+\eps)^{-\Delta}.
\end{align}
In contrast to the first case, here we find that the influence of the size of disguised uninfected individuals does not vanish asymptotically in relation to the number of infected individuals in \eqref{Eq_ErwA_delta2}. 

By a similar argument as the first case, \eqref{size_vo+_m1_case2} and \eqref{Eq_ErwA_delta2} imply
{\small
\begin{align}
    \notag \Erw_{\cGd} &\brk{\vec B' \mid  \cC_{\Gamma},\cF_{\Delta} } \notag \\&= \bc{1 + O\bc{ n^{-{\Omega(1)}}}} k \Delta \bc{1 - {\frac{k}{n}\bc{1-(1+\eps)^{-\Delta}-O\bc{\frac{n^{-{\Omega(1)}}}{(1+\eps)^{\Delta}}}}}} ^{\bar \Gamma} \\
    \notag &= \bc{1 + O\bc{n^{-\Omega(1)}}} k \Delta \exp \bc{{- \bc{1-(1+\eps)^{-\Delta}-O\bc{\frac{n^{-{\Omega(1)}}}{(1+\eps)^{\Delta}}}} \ell}} \\
    &= \bc{1 + O\bc{n^{-{\Omega(1)}}}} \notag \\ & \qquad \cdot \Delta k \bc{1-{ \bc{1-(1+\eps)^{-\Delta}-O\bc{n^{-\Omega(1) } {(1+\eps)^{-\Delta}}}} (\ell + O(\ell^2))}}, \label{size_B'-case2}
\end{align}
}
and similarly to \eqref{eq_sizeB_2}, combining this with the Chernoff bound and Claim~\ref{Cor_DistEqual} yields that
\begin{align}
    \vec B = \bc{1 + O\bc{n^{-\Omega(1)}}} \Delta k \cdot (1-{(1+\eps)^{-1}} n^{- (1 - \theta)/\Delta} + O(n^{-2(1-\theta)/\Delta})) \label{eq_size_b_case_2}
\end{align}
with probability $1-O(n^{-\Omega(1)})$.
Therefore, the probability of an infected individual not belonging to $\oneminusminus(\cG)$ satisfies the following analog of \eqref{Bcomb}:
\begin{align*}
    \Pr \bc{ x \not \in \oneminusminus(\cGd) \mid x \in \one(\cG), \vec B} & = \binom{k \Delta - \vec B}{\Delta} \binom{k\Delta}{\Delta}^{-1} \\ & = \geb{\bc{1+o(1)}}{(1+\eps)^{-\Delta}}  n^{-(1 - \theta)}.
\end{align*}
Since $2\theta-1 \leq 0$ by assumption, it follows that
\begin{align}
    \label{Erw_A_case2} \Erw \brk{\vec A \mid \vec B} =  \geb{\bc{1+o(1)}} (1+\eps)^{-\Delta}  n^\theta n^{-(1 - \theta)} \leq \geb{\bc{1+o(1)}} {(1+\eps)^{-\Delta} }.
\end{align}
Thus, Lemma \ref{lem_size_1--_delta} follows from \eqref{Erw_A_case1} and \eqref{Erw_A_case2} followed by Markov's inequality.
\end{IEEEproof}

Theorem \ref{thm_DD_achievability_delta} now follows directly from Lemma \ref{lem_size_1--_delta} and Claim \ref{claim_dd_individualtypes}.

\subsection{A converse for \DD~ in the sparse regime: Proof of \Thm~\ref{thm_DD_converse_delta}}\label{proof_thm_DD_converse_delta}

In accordance with Claim \ref{claim_dd_individualtypes}, we first provide a lemma bounding the size of $\oneminusminus(\cGd)$, the set of infected individuals appearing in at least one test with only easy uninfected individuals.

\begin{lemma}
For $\theta < 1/2$ and $m = (1 - \eps) \mDD(\Delta)$, we have under the random regular design  that 
\begin{align*}
    \Erw & \brk{ \abs{\oneminusminus(\cGd)} } \\ &=  \geb{\bc{1+O\bc{ n^{-{\Omega(1)}}}}}k \bc{ 1 - \bc{1-\exp \bc{ - (1 - \eps)^{-\Delta} \bc{1 - 1/\Delta}}}^{\Delta} }.
\end{align*}
\end{lemma}
\begin{IEEEproof} 
We re-use the notations $\bar\ell$ and $\bar\Gamma$ in \eqref{eq:l_def}, but their expressions are modified as follows in accordance with the choice $m = (1 - \eps) \Delta k^{1+\frac{(1-\theta)}{\Delta \theta}}$ associated with $\theta < \frac{1}{2}$:
\begin{align}
    \label{eq_l_bargamma_lower} \bar\ell = (1 - \eps)^{-1} n^{- (1 - \theta)/\Delta} \qquad \text{and} \qquad \bar \Gamma = (1 - \eps)^{-1} n^{(1- \theta)(1 - 1/\Delta)}.
\end{align}
We additionally recall $\vec B$ from \eqref{def_B} as the number of tests featuring exactly one infected individual and no elements of $\zeroplus$.  By the same calculation as in \eqref{size_B'-case2} and \eqref{eq_size_b_case_2} with $\ell$ and $\bar \Gamma$ replaced by the values in \eqref{eq_l_bargamma_lower}, we obtain
\begin{align} 
    \vec B  &=  \bc{1 + O \bc{n^{-\Omega(1)}}} k \Delta \bc{ 1 - \bc{1 - \eps}^{-\Delta} k n^{-1} }^{\bar \Gamma}\notag\\
     &= \bc{1 + O \bc{n^{-\Omega(1)}}} k \Delta \exp \bc{ - (1 - \eps)^{-\Delta} \bc{1 - 1/\Delta} }\label{eq_size_b}
\end{align}
with probability at least $1 - o(n^{-8})$.
Therefore, we can calculate the probability that an infected individual does not belong to $\oneminusminus(\cGd)$ via Claim~\ref{stirling_applied} as follows:
\begin{align*}
    &\Pr \bc{ x \not \in \oneminusminus(\cGd) \mid x \in \one(\cG)} \notag\\ &= \geb{\bc{1+O\bc{ n^{-{\Omega(1)}}}}} \frac{ \binom{k\Delta - \vec B}{\Delta} }{\binom{k\Delta}{\Delta}} \notag\\&= \geb{\bc{1+O\bc{ n^{-{\Omega(1)}}}}} \bc{1-\exp \bc{ - (1 - \eps)^{-\Delta} \bc{1 - 1/\Delta}}}^{\Delta}.
\end{align*}
Since there are $k$ individuals in $x \in \one(\cG)$ by assumption, we obtain
\begin{align}
    \label{eq_dd_delta_div} &\Erw \brk{ \abs{V_1(\cG) \setminus \oneminusminus(\cGd)} } \\ &= \geb{\bc{1+O\bc{ n^{-{\Omega(1)}}}}} k { \bc{1-\exp \bc{ - (1 - \eps)^{-\Delta} \bc{1 - 1/\Delta}}}^{\Delta} }
 \end{align}
and the lemma follows using $\abs{\oneminusminus(\cGd)} = k - \abs{V_1(\cG) \setminus \oneminusminus(\cGd)}$.
\end{IEEEproof}
Knowing the expected size of $\abs{\oneminusminus(\cGd)}$, Markov's inequality leads to the following.
\begin{corollary}
\label{cor_dd_fails_delta}
Let $\theta < 1/2$ and $m = (1 - \eps) \mDD(\Delta)$ and $\Delta = \Theta(1)$. Then, with  probability at least
\begin{equation}
    1 - \frac{ 1 - { \bc{1-\exp \bc{ - (1 - \eps)^{-\Delta} \bc{1 - 1/\Delta}}}^{\Delta} }} {1 - \gamma} \label{ProbBound}
\end{equation}
there are at least $\gamma k$ infected individuals $x \in \one(\cG) \setminus \oneminusminus(\cGd)$.
\end{corollary}
Claim \ref{claim_dd_individualtypes} and Corollary \ref{cor_dd_fails_delta} immediately imply \Thm~\ref{thm_DD_converse_delta}, since \eqref{ProbBound} is always positive for sufficiently small $\gamma$, and approaches one as $\Delta \to \infty$.

\section{Non-Adaptive Group Testing with $\Gamma$-Sized Tests}

In this section, we formally state and prove our main results concerning non-adaptive group testing $\Gamma$-sized tests, namely, a universal lower bound and an algorithmic upper bound that matches the lower bound. Recall that we focus on the regime $\Gamma = \Theta(1)$. Within this section, $\cG$ denotes an arbitrary non-adaptive pooling scheme with respect to the $\Gamma$-sparsity constraint. The section contains two main parts, outlined as follows:
\geb{
\begin{itemize}
    \item Theorem~\ref{thm_gsparse_universal_informationtheory} states our universal lower bound for non-adaptive designs.  The proof is based on a careful analysis of the appearance of disguised individuals (see Section \ref{Sec_Individualtypes}), with the idea being that too many such individuals leads to failure.  For $\theta < \frac{1}{2}$, we additionally use the idea of identifying sufficiently many tests with multiple individuals of degree one, prohibiting reliable inference.
    \item Theorems~\ref{Thm_DDg} and~\ref{thm_dd_gamma_sparse_optimal} analyze the performance of the DD and SCOMP algorithms.  The proofs are again based on the idea that in the underlying pooling scheme, any infected individual appears in at least one test with only definitive uninfected individuals (elements of $\zerominus(\cG)$). We refer the reader to Sections~\ref{Sec_Individualtypes} and~\ref{Sec_Nishimori} for further insights on these properties.  The test size constraints pose additional technical challenges compared to the unconstrained setting \cite{Coja_2019}, in particular leading us to adopt a less standard matching-based test design when $\theta < \frac{1}{2}$.
\end{itemize}
}
\subsection{A universal information-theoretic bound}

The first statement that we prove is an information-theoretic converse that applies to \textit{any} non-adaptive group testing scheme with maximum test size $\Gamma$. Denote by
\begin{align}
    \label{minf_gamma}\minf{}_{, \Gamma} =  \max \cbc{ \bc{1 + \floor{\frac{\theta}{1-\theta}}} \frac{n}{\Gamma}, 2 \frac{n}{\Gamma + 1}},
\end{align}
which we will show to be the sharp information-theoretic phase transition point when $\Gamma \ge 1 + \floor{\frac{\theta}{1-\theta}}$; \geb{note that if this inequality is reversed, then $\minf{}_{, \Gamma} > n$, whereas $n$ tests trivially suffice via one-by-one testing.}  In \cite{Gandikota_2016} a lower bound of $(n/\Gamma)(1+o(1))$ was proved, and we see that in the regime $\Gamma = \Theta(1)$, our lower bound improves on this for all $\theta \in (0,1)$. 

\begin{theorem}\label{thm_gsparse_universal_informationtheory}
Let $\theta \in (0,1)$, $\Gamma \ge 1 + \floor{\frac{\theta}{1-\theta}}$, and $\delta>0$.
Furthermore, let $\cG$ be any non-adaptive pooling scheme (deterministic or randomised) with $ m = (1 - \delta) \minf{}_{, \Gamma}$ tests such that each test contains at most $\Gamma$ individuals. Then any inference algorithm $\cA$ fails in recovering $\SIGMA$ from $(\hat \SIGMA,\cG)$
\begin{itemize}
    \item with probability $1 - o(1)$ if $\theta/(1-\theta) \not\in \ZZ$,
    \item with probability $\Omega(1)$ if $\theta/(1 - \theta) \in \ZZ$.
\end{itemize}
\end{theorem}
Thus, even with unlimited computational power, there cannot be any algorithm with a maximum test size of $\Gamma$ that is able to infer the infected individuals correctly \whp~ once the number of tests drops below \eqref{minf_gamma}.  The distinction between integer vs.~non-integer values of $\theta/(1 - \theta)$ arises for technical reasons (e.g., counting the number of nodes with degree at most $\floor{\theta/(1 - \theta)}$), and we found it difficult to prove a high-probability (rather than constant-probability) failure result in the integer case.

The proof of the universal information-theoretic converse resembles the proof of \cite{Coja_2019_2} for the existence of a universal information-theoretic bound for unrestricted non-adaptive group testing, but several modifications are required to handle the test size constraint. \geb{We provide the details in the following subsection.} 

\subsection{Proof of \Thm~\ref{thm_gsparse_universal_informationtheory}}
We start by defining
\[ d^+ = 1+ \floor{\frac{\theta}{1-\theta}} \quad \text{and} \quad d^- = \floor{\frac{\theta}{1-\theta}} \, . \]
For the proof, we distinguish two different regimes for $\theta$, as stated in Proposition \ref{prop_many_disguised_universal} and Proposition \ref{prop_universal_non_divis_prop_small_theta}.  We start with the following proposition \geb{addressing the existence of disguised individuals.}

\begin{proposition}\label{prop_many_disguised_universal}
Let $1/2 \leq \theta < 1$, $\Gamma \ge d^+$, and let $\cG$ be an arbitrary pooling scheme with tests of size at most $\Gamma$. For any $\geb{\epsilon} \in (0,1)$, if $ m = (1 - \geb{\epsilon}) d^+ \frac{n}{\Gamma}$, then
\begin{itemize}
    \item $\Pr \bc{ \abs{\oneplus(\cG)} > \log n} \geq 1 - o(1)$ and $\Pr \bc{ \abs{\zeroplus(\cG)} > \log n} \geq 1 - o(1)$ if  $ \frac{\theta}{1 - \theta} \not \in \ZZ$
    \item $\Pr \bc{ \abs{\oneplus(\cG)} \geq 1} = \Omega(1)$ and $\Pr \bc{ \abs{\zeroplus(\cG)} > \log n } \geq 1 - o(1)$ if $\frac{\theta}{1 - \theta}  \in \ZZ$
\end{itemize}
\end{proposition}

\subsubsection{Proof of \Prop~\ref{prop_many_disguised_universal}}

Let $\cG$ be an arbitrary pooling scheme such that each test contains at most $\Gamma$ individuals.  We denote by $V(\cG)$ the set of individuals, and by $F(\cG)$ the set of tests in $\cG$ (by the identification of $\cG$ with a bipartite graph). Instead of analysing $(\cG, \hat \SIGMA)$, similarly to in the $\Delta$-divisible case, we analyse a related model that eliminates nuisance dependencies between the infection status of different individuals.  

Specifically, let $p = \frac{k - \sqrt{k} \log n}{n}$, and let $\SIGMA^*$ be a $\cbc{0,1}$-valued vector, where every entry is one with probability $p$. 
\Cor~\ref{cor_bernoulli_ok} guarantees that if the modified model satisfies
\begin{align*}
   &\Pr \bc{ \abs{\oneplus(\cG, \SIGMA^*)} >  2 C} \geq 1 - o(1) \\ \text{and} \qquad &\Pr \bc{ \abs{\zeroplus(\cG, \SIGMA^*)} > 2 C} \geq 1 - o(1), 
\end{align*}

then the original model satisfies
\begin{align*}
    &\Pr \bc{ \abs{\oneplus(\cG, \SIGMA)} > C} \geq 1 - o(1)\\
    \text{and} \qquad &\Pr \bc{ \abs{\zeroplus(\cG, \SIGMA)} > C} \geq 1 - o(1)
\end{align*}
Thus, working with the modified model is sufficient.  For the sake of brevity, we henceforth write $\cGg$ in place of $(\cG, \SIGMA^*)$, leaving the dependencies on $\SIGMA^*$ implicit.

We proceed by finding a set of (many) individuals, that have a high probability of being disguised.   We will apply the probabilistic method iteratively to create the desired set. Creating this set turns out to be delicate due to the dependencies in an arbitrary pooling scheme. Luckily, it will suffice for our purposes to note that whenever individuals have distance at least $6$ (\geb{i.e., the shortest path between two individuals has at least $6$ edges}) in the underlying graph, the events of being disguised are independent \cite{Coja_2019_2}.  \geb{To see this, note that we can identify whether an individual is disguised by looking at the tests it is in, and the defectivity status of all other individuals in those tests.  This procedure only reaches the second neighborhood, so a separation of 6 is enough to ensure there is no overlap when doing this for two different nodes (which implies independence under an i.i.d.~defectivity model).}

In the following, we denote the set of all disguised individuals by $$ V^+(\cG) = \zeroplus(\cG) \cup \oneplus(\cG).$$  We first present a claim establishing that we may safely assume that each individual gets tested $\Theta(1)$ times.

\begin{claim} \label{const_degree}
Given any pooling scheme $\cG'$ with $m = (1 - 2 \geb{\eps}) d^+ \frac{n}{ \Gamma}$ (for some $\geb{\eps} > 0$) such that each test contains at most $\Gamma = \Theta(1)$ individuals, there is another pooling scheme $\cG$ such that each test contains at most $\Gamma = \Theta(1)$ individuals with $m = (1 - \geb{\eps}) d^+ \frac{n}{ \Gamma}$, while also satisfying the following:
\begin{itemize}
    \item Each individual is contained in at most $C = \Theta(1)$ tests;
    \item Recovery of $\SIGMA$ from $(\cG', \hat \SIGMA')$ implies recovery from $(\cG, \hat \SIGMA)$.
\end{itemize}
\end{claim}
\begin{IEEEproof}
Given $\cG'$ and a constant $C \in \NN$, there is $C' \in \NN$ such that there are at most $n/C$ individuals of degree at least $C'$ in $\cG'$, which is an immediate consequence of $m$ being linear in $n$ (due to $\Gamma = \Theta(1)$).
Design $\cG$ such that each individual of $\cG'$ with degree larger than $C'$ gets tested individually (causing $n/C$ additional tests) and all other individuals and tests stay the same as under $\cG'$. Clearly, if recovery in $\cG'$ was possible, then it is possible in $\cG$ as well. Setting $C = \frac{\Gamma}{\geb{\eps} d^+}$, the claim follows.
\end{IEEEproof}

In addition to being able to assume there are no individuals with an overly high degree, we can also prove that there cannot be too many individuals with an overly low degree.
\begin{lemma}
\label{lem_no_integer_small_degree}
Let $\cG$ be the given pooling scheme and $m \leq (1 - \geb{\eps})d^+ \frac{n}{\Gamma}$, where $\Gamma \ge d^{+}$. If there is a constant $\alpha > 0$ such that the number of individuals of degree at most $d^-$ is $\alpha n$, then we have the following:
\begin{itemize}
    \item $\abs{\oneplus(\cG)} > 2 \log n$ \whp~if $\theta/(1-\theta) \not \in \ZZ$,
    \item $\abs{\oneplus(\cG)} > 0$ with probability $\Omega(1)$ if $\theta/(1-\theta) \in \ZZ$.
\end{itemize}
\end{lemma}
\begin{IEEEproof}
Suppose that the number of individuals with degree at most $d^-$ is $\alpha n$, and recall that $p=\frac{k-\sqrt{k}\log(n)}{n}$.
Without loss of generality, we can assume that there are no tests of degree zero or one.
Otherwise, remove them and each connected individual from the testing scheme and note that, by the assumed lower bound $\Gamma \ge d^{+}$, there are at least $\epsilon n$ individuals left.
This manipulated graph satisfies the same inequality between the number of individuals and number of tests and, clearly, if the inference of $\SIGMA$ does not succeed on this manipulated graph, then it cannot succeed in $\cG$.
Before proceeding, we introduce the following auxiliary result.

\begin{claim}\label{claim_many_disjoint}
\geb{Under the preceding setup, suppose that there exists a set $I^- \subset V$ of individuals of degree at most $d^-$ with $|I^-| \le \alpha n$ ($\alpha \in (0,1)$).  Then, there exists $\beta \in (0,\alpha)$ (depending only on $d^-$ and $\Gamma$) such that there must also exist $I^+ \subset I^-$ with $\abs{ I^+ } = \beta n$, having the property that for all pairs $x \neq y$ in $I^+$ it holds that dist$(x,y) \geq 6$.  }
\end{claim}
\begin{IEEEproof}
\geb{ First recall from Claim \ref{const_degree} that all degrees in the graph are bounded. Consider the procedure of {\em iterating through all individuals $x \in I^-$, and deleting all $y \in I^-$ of distance at most four from $x$}, and repeating until no individuals remain.  Let $I^+$ denote set of $x$'s visited by this process.  Since the degrees in the graph are finite, each removal only decreases the size of the set $I^-$ by at most a constant, and the assertion of the claim follows.}    
\end{IEEEproof} 

Let $B$ be the largest possible subset of individuals satisfying the requirements of Claim \ref{claim_many_disjoint}. Thus, $B$ is a set of $\beta n$ individuals such that for all $x \neq x' \in B$ we have 
\begin{itemize}
    \item[\textbf{(B1)}] $\deg(x) \leq d^-$
    \item[\textbf{(B2)}] $\dist(x,x') \geq 6$.
\end{itemize}
We analyze a single individual $x \in B$ using the FKG inequality (e.g., see \cite[Proposition 1]{Fortuin_1971}); as noted in \cite[Lemma 4]{Aldridge_2019}, the events of $x$ being disguised in each of its tests are increasing \geb{with respect to $\SIGMA^*$ (in the sense that marking additional individuals as infected in $\SIGMA^*$ can only increase the probability that an individual $x$ is disguised).}  Hence, the FKG inequality yields the following, recalling that we are considering the case that $\deg(a) \ge 2$ for all $a$:
$$ \Pr \bc{x \in V^+(\cG) } \geq \prod_{a \in \partial x} \bc{1 - \bc{1 - p}^{\deg(a)-1}}. $$
Then, by the fact that $\deg(x) \le d^{-} = O(1)$ within $B$, Claim~\ref{Bernoulli} guarantees that
$$ \prod_{a \in \partial x} \bc{1 - \bc{1 - p}^{\deg(a)-1}} \geq   C p^{d^-}  $$ for some constant $C$ depending on $\theta$ and $\Gamma$.

We now turn to the total number of disguised individuals in $B$.
As noted above, for two individuals $x, x' \in B$, the events 
of being disguised are independent due to the pairwise distances being at least $6$, as described above.  
Thus, the number of disguised infected individuals $\abs{\oneplus(\cG)}$ dominates a binomial random variable $\Bin(\beta n, p \cdot C p^{d^-})$. Since $np \sim k = n^{\theta}$, the mean of this binomial distribution scales as $\Theta(n^{\theta - (1-\theta) d^-})$.  In particular, when $\frac{\theta}{1-\theta}$ is non-integer, the choice $d^- = \floor{ \frac{\theta}{1 - \theta} }$ ensures that the exponent is positive, and the Chernoff bound gives \whp~that

\begin{align}
    \label{eq_gamma:azuma} \abs{\oneplus(\cG)} \geq n^{\Omega(1)}.
\end{align}
On the other hand, if $\frac{\theta}{1 - \theta}$ is integer-valued, then the mean of the binomial is $\Theta(1)$, which is enough to ensure that $\abs{\oneplus(\cG)} > 0$ with $\Omega(1)$ probability.  Combining these two cases completes the proof of Lemma \ref{lem_no_integer_small_degree}.
\end{IEEEproof}

As an immediate consequence of \Lem~\ref{lem_no_integer_small_degree}, in any group testing instance that succeeds \whp, there are at most $o(n)$ individuals of degree up to $d^-$.   However, if $m \leq (1-\geb{\eps}) d^+ n/\Gamma$ we find at least $\alpha n$ individuals of degree at most $d^-$ (for some $\alpha$ depending on $\geb{\eps}$) by the handshaking lemma \cite[\Cor~1.3]{Wu_2014}, yielding a contradiction. Therefore, \Prop~\ref{prop_many_disguised_universal} is a direct consequence of \Lem~\ref{lem_no_integer_small_degree}, with the claims regarding $\abs{\zeroplus(\cG)}$ following easily from those regarding $\abs{\oneplus(\cG)}$ in the same way as Corollary \ref{cor_bernoulli_ok2}.
\hfill $\blacksquare$\\

We now turn to the sparse regime $\theta < \frac{1}{2}$, establishing the following proposition as a stepping stone to \Thm~\ref{thm_gsparse_universal_informationtheory}.
\begin{proposition}
\label{prop_universal_non_divis_prop_small_theta}
Let $0 < \theta < 1/2$, and let $\cG$ be an arbitrary pooling scheme with tests of size at most $\Gamma$. For all $\geb{\epsilon_*} > 0$ and sufficiently large $n$,
if $ m \le (2-\geb{\epsilon}) \frac{n}{\Gamma + 1}$, then any algorithm (efficient or not) fails at recovering $\SIGMA$ from $\hat \SIGMA$ and $\cG$ \whp.
\end{proposition}

\subsubsection{Proof of \Prop~\ref{prop_universal_non_divis_prop_small_theta}}
The proof hinges on a fairly straightforward observation. We can again assume without loss of generality that there are no tests containing only one individual (otherwise, we remove them and their corresponding individuals from the testing scheme). By a simple counting argument, there can be only $o(n)$ such tests (since otherwise $m > 2 n/\Gamma$, which is a contradiction).  \geb{In addition, we can assume that there are no degree-zero individuals; if there were $\Omega(n)$ of them, high-probability correct inference would trivially be impossible, whereas with $o(n)$ of them, they can be removed and the subsequent analysis still holds for those remaining, with the $o(n)$ difference not impacting the final result.}

Then, another counting argument leads to the fact that the number of individuals of degree 1 is large when $m < 2 n/\Gamma$, as stated in the following.
\begin{lemma}
\label{lemma_many_degree_one_individuals}
If $m = (2 - \geb{\eps}) n/\Gamma$, then there are at least $\geb{\eps} n$ individuals of degree 1.
\end{lemma}
\begin{IEEEproof}Denote by $\alpha n$ the number of individuals of degree 1, i.e., $\alpha > 0$ is the proportion of such individuals. Then the lemma follows by double counting edges (on the individual side and on the test-side):
$$ (2 - \geb{\eps})n = m \Gamma \geq \sum_{a \in F(\cG)} \deg(a) = \sum_{x \in V(\cG)} \deg(x) \geq \alpha n + 2(1 - \alpha)n.$$
Solving for $\alpha$ yields $\alpha \geq \eps$, and the lemma follows.
\end{IEEEproof}

The next lemma shows that there can only be a small number of tests containing more than one individual of degree 1.
\begin{lemma}
\label{lem_tests_one_degree_ind}
If there is any algorithm recovering $\SIGMA$ from the test results with $\Omega(1)$ probability, then the number of tests containing more than one individual of degree one is below $n / \sqrt{k} = o(n)$.
\end{lemma}
\begin{IEEEproof}
Suppose that at least $n / \sqrt{k}$ tests contain at least two individuals of degree one, and consider any resulting subset of $2n / \sqrt{k}$ individuals (two per test).  The average number of infected individuals among these is $(2n / \sqrt{k}) \cdot (k/n) = 2\sqrt{k}$.  Hence, by the Chernoff bound for the hypergeometric distribution, \whp~there are at least $\sqrt{k} / \log n$ such infected individuals.  On the other hand, among these tests, the average number in which both of these degree-one individuals are infected is $(n / \sqrt{k}) \cdot O((k/n)^2) = O(k \sqrt{k} / n)$, so Markov's inequality implies that \whp~the actual number is $O(\sqrt{k} n^{-\Omega(1)})$.  

Hence, all but an $o(1)$ fraction of the above-mentioned $\sqrt{k} / \log n$ infected individuals must be in a test with both a degree-one infected and a degree-one uninfected individual.  For these tests, the inference algorithm cannot do better than guess which one is the infected one, but then the probability of all guesses being correct is $(1/2)^{\omega(1)} = o(1)$, from which the lemma follows.
\end{IEEEproof}
We are now in a position to prove Proposition \ref{prop_universal_non_divis_prop_small_theta}.  For $m = (2 - \geb{\eps}) n / \Gamma$, we find by \Lem~\ref{lemma_many_degree_one_individuals} that there are at least $\geb{\eps} n$ individuals of degree 1. By \Lem~\ref{lem_tests_one_degree_ind} and the fact that $\Gamma = \Theta(1)$, only $o(n)$ such individuals can be placed together in any tests, and hence, the total number of tests is at least $\geb{\eps} n - o(n)$. Formally,
\begin{align}
    \label{eq_m_upper} (2 - \eps) n / \Gamma = m \geq \geb{\eps} n - o(n).
\end{align}
Solving \eqref{eq_m_upper} for $\geb{\eps}$, we find $\geb{\eps} \leq \frac{2}{\Gamma + 1} + o(1)$. Hence, $$ m \geq  \bc{2 - \frac{2}{\Gamma + 1} - o(1)} \frac{n}{\Gamma} = 2 \frac{n}{\Gamma + 1} - o(n), $$
and the proposition follows. \hfill $\blacksquare$

The universal lower bound in the considered regime is a direct consequence of \Prop~\ref{prop_many_disguised_universal},  \Prop~\ref{prop_universal_non_divis_prop_small_theta}, and Claim \ref{Claim_1+0+}.  The proof of \Thm~\ref{thm_gsparse_universal_informationtheory} is thus complete.

\subsection{Algorithmic bound: Preliminaries and statement of result}

We now turn to the problem of establishing an upper bound, with a suitably-chosen test design and an efficient inference algorithm, that matches the universal lower bound. 
We start by recalling the definition of $\tilde{\cGg}$ in Section~\ref{pool_gamma}:
\begin{align}
   \tilde{ \cGg}(\theta)=
    \begin{cases}
    \cGg \quad \text{ if } \theta \geq 1/2 \\
    \cGg^* \quad \text{ otherwise }
    \end{cases}
\end{align}
We equip this pooling scheme with the efficient {\tt DD} algorithm (see Algorithm~\ref{dd_algorithm}). In the following, we will see that the combination of these tools will lead to information-theoretically optimal performance in the $\Gamma$-sparse setting with $\Gamma = \Theta(1)$.

\begin{proposition}\label{Prop_DD}
Define $$m_{\rm SCOMP}(\tilde{\cGg})=\max\cbc{ \bigg(1 + \left \lfloor\frac{\theta}{1-\theta}\right\rfloor\bigg)\frac{n}{\Gamma},2\frac{n}{\Gamma+1}}.$$
For $\Gamma=\Theta(1)$ and $m=(1+\eps) m_{\rm SCOMP}$, 
we have
$$\mathbb{P}(\cA_{\rm DD}(\tilde{\cGg},\hat \SIGMA,k)=\SIGMA)=1-o(1).$$
\end{proposition}

To prove this result, we handle the dense regime $\theta > \frac{1}{2}$ in \Thm~\ref{Thm_DDg} below, the sparse regime $\theta < \frac{1}{2}$ in \Thm~\ref{thm_dd_gamma_sparse_optimal}, and combine them in Section \ref{sec:pieces}. We observe that $m_{\rm SCOMP}= m_{\inf,\Gamma}$, i.e., the achievability and converse results match for all $\theta \in (0,1)$.

\subsection{Algorithmic feasibility I: The configuration model}\label{algo_dense}
We first show that the \DD~ algorithm succeeds with a slightly higher threshold, namely $\max \cbc{2,1+\floor{\frac{\theta}{1-\theta}} }\frac{n}{\Gamma}$, employing the configuration model $\cGg$.  We define
\begin{align} \label{Eq_mdd}
    \Delta_{\rm DD}(\theta) = \max \cbc{2,1+\floor{\frac{\theta}{1-\theta}} }, \mDD(\cGg) =  \Delta_{\rm DD}(\theta) \frac{n}{\Gamma},
\end{align}  
representing this achievability bound for \DD~in $\cGg$. 

\begin{theorem}\label{Thm_DDg}
Let $\eps > 0$ and $m \ge \mDD(\cGg)$. Then \whp~\DD~infers $\SIGMA$ from $(\cGg, \hat \SIGMA)$ correctly.
\end{theorem} 

We stress at this point that \Thm~\ref{Thm_DDg} gives a performance guarantee for the configuration model with any sparsity level, but it will turn out in due course that for $\theta < \frac{1}{2}$ a different model performs slightly better. Note also that for $\theta \ge \frac{1}{2}$, we can simplify $\max \cbc{2,1+\floor{\frac{\theta}{1-\theta}} } = 1+\floor{\frac{\theta}{1-\theta}}$.

\subsubsection{Proof of \Thm~\ref{Thm_DDg}} 
The proof of \Thm~\ref{Thm_DDg} hinges on a slightly delicate combinatorial argument. Recall from Figure \ref{figure_types_of_individuals} that $\oneminusminus$ consists of those infected individuals that appear in at least one test with only individuals that are removed in the first step of \DD~ (i.e., the easy uninfected individuals \geb{$\zerominus$}). By Claim~\ref{claim_dd_individualtypes}, \DD~ succeeds if and only if $\one(\cG) = \oneminusminus$.
\begin{lemma}\label{size_1--} Let $\vA = \abs{ \one(\cG) \setminus \oneminusminus(\cGg) }$ denote the number of infected individuals that are not identified in the second step of \DD. If $m \geq  m_{\DD}$, then it holds \whp~that $\vA = 0$.
\end{lemma}

The proof of Lemma \ref{size_1--}, while conceptually not difficult and similar to \cite{Coja_2019}, is technically challenging, as we have to deal with subtle dependencies in the pooling scheme, caused by the mutli-edges given through the configuration model. A heuristic argument with a (false) independence assumption can provide some intuition as follows:  In order for an individual $x$ to be part of a test containing no  infected individual (besides possibly $x$ itself) is roughly $(1 - k/n)^{\Gamma-1}$.  For $x$ to be disguised, thus being element of $\zeroplus( \cGg )$ or $ \oneplus(\cGg)$, $x$ may not be part of such a test. Hence, the probability of $x$ being disguised would be roughly $\bc{1 - (1 - k/n)^{\Gamma-1} }^\Delta$ if the associated $\Delta$ events were independent (recall that $\Delta = m\Gamma/n$ is the degree of each individual in the random regular design).

To formally deal with the dependencies in the graph, we proceed as follows. Denote by $(\vY_1, \dots, \vY_m)$ the number of infected individuals in the tests. There are $n\Delta$ edges connected to individuals, out of which exactly $k\Delta$ correspond to infected individuals. Each test chooses exactly $\Gamma$ individuals without replacement, and hence, the number of \geb{infected} individuals in any test follows a hypergeometric distribution. In order to get a handle on this distribution, we introduce a family $(\vX_1,...,\vX_m)$ of independent binomial variables, such that $\vX_i \sim \Bin(\Gamma,k/n)$. These variables can accurately describe the local behaviour of how many infected individuals belong to test $a_i$. We define $\cE_{\Gamma}$ to be the event that the overall number of edges containing infected individuals is correct, i.e.,  
\begin{align}
\cE_{\Gamma}=\left\lbrace\sum_{i=1}^{m}\vX_i=k\Delta \right\rbrace. \label{event_1}
\end{align}
Claim~\ref{Prob_nr_inf_cor} implies that $\pr\bc{\cE_{\Gamma}} = \Omega( (n \Delta)^{-{1/2}} )$. In addition, we have the following. 

\begin{lemma}\label{Eq_Dist}
\geb{The sequence $(\vY_1,...,\vY_n)$ is identically distributed with $(\vX_1,...,\vX_n)$ given the event $\cE_{\Gamma}$}
\end{lemma}
\begin{IEEEproof}

By the definition of $\vY_i$, we find for any $(y_i)_i$ satisfying $\sum_{i}y_i = k\Delta$ that
\begin{align*}
& \Pr(\vY_i=y_i, \forall i\in [m] ) \\& \qquad = \binom{k\Delta}{y_1,...,y_m}\binom{(n-k)\Delta}{\Gamma-y_1,...,\Gamma-y_m}\binom{n\Delta}{\Gamma,...,\Gamma}^{-1} \\ & \qquad = \frac{\prod_{i=1}^m \binom{\Gamma}{y_i}}{\binom{n\Delta}{k\Delta}}.
\end{align*}
where the equality follows by rewriting in terms of factorials and simplifying. 
Furthermore, given $\sum_i x_i = k\Delta$, we have
\begin{align*}
    &\Pr(\vX_i=x_i, \forall i\in [m] |\cE_{\Gamma}) \\ & \qquad =  \prod_{i = 1}^m \binom{\Gamma}{x_i} (k/n)^{x_i} (1 - k/n)^{\Gamma - x_i}  \bc{\Pr \bc{\cE_{\Gamma}}}^{-1}.
\end{align*}
Now, for two sequences $(y_i)_{i \in [m]}$ and $(y'_i)_{i \in [m]}$ such that $\sum_{i=1}^m y_i = \sum_{i=1}^m y'_i = k\Delta$, we obtain 
$$\frac{\Pr(\forall i\in [m]:\geb{\vY_i}=y_i)}{\Pr(\forall i\in [m]:\geb{\vY_i}=y'_i)}=\prod_{i=1}^{m}\frac{\binom{\Gamma}{y_i}}{\binom{\Gamma}{y'_i}}=\frac{\Pr(\forall i\in [m]:\geb{\vX_i}=y_i|\cE_{\Gamma})}{\Pr(\forall i\in [m]:\geb{\vX_i}=y'_i|\cE_{\Gamma})}.$$
This implies the lemma.
\end{IEEEproof}
Thus, \geb{similarly to the analysis following Lemma \ref{lem_x_y_delta},} we are able to carry out all necessary calculations with respect to $(\vX_1,...,\vX_n)$ and transfer the results to the original pooling scheme. 
For the next step, we need to get a handle on the number of positive and negative tests occurring in this setting.
Let $\vm_0 = \vm_0(\cGg, \SIGMA)$ be the number of tests that render a negative result, and let $\vm_1 = \vm_1(\cGg, \SIGMA)$ be the number of tests that render a positive result. Then $\vm_0$ and $\vm_1$ are highly concentrated around their means as follows. 
\begin{lemma}\label{Lem_m0_gamma}
With probability $1-o(n^{-2})$, we have 
$$ \vm_0  = \bc{1 + n^{-\Omega(1)}} m\bc{1- k/n}^{\Gamma} $$ and $$ \vm_1  = \bc{1 + n^{-\Omega(1)}} m \bc{ 1 - \bc{1- k/n}^{\Gamma}}.$$
\end{lemma}
\begin{IEEEproof}
Recalling the definitions of $(\vY_i)_i$ and $(\vX_i)_i$ from \eqref{event_1}, we have
$$ \vm_0 = \sum_{i=1}^m \vecone \cbc{ \vY_i = 0 },$$ 
and we further denote by 
    $$\vm'_0=\sum_{i=1}^m \mathbb{1}\{\vX_i=0\} \qquad \text{and} \qquad \vm_1' = m - \vm_0'$$ 
the number of negative and positive tests as modelled by the family of independent binomial variables $(\vX_i)_i.$ Clearly, as the $\vX_i$ are mutually independent, $$\Erw \brk{\vm_1'} = m \cdot \bc{ 1 - \Pr \bc{\Bin(\Gamma,k/n)=0}}= m \bc{1 - \bc{1-\frac{k}{n}}^\Gamma}.$$
\geb{Observing that $\Erw \brk{\vm_1'} = \Theta(k)$ (since $m = \Theta(n)$ due to $\Gamma = \Theta(1)$),} the Chernoff bound (\Lem~\ref{Lem_Chernoff}) guarantees that $$\Pr\bc{ \abs{\vm'_1 - \Erw(\vm'_1) } > \sqrt{k} \log(n)\mid \Gamma } = o(n^{-10})$$
and, \geb{similar to the proof of Lemma \ref{lem_m0_gdelta},} by combining \Lem~ \ref{Eq_Dist} with Claim~\ref{Prob_nr_inf_cor}, we obtain
$$ \Pr\bc{\abs{\vm_1 - \Erw(\vm'_1)} > \sqrt{k} \log(n)\mid \Gamma} = o(n^{-8}).$$
Thus, the first part of the lemma follows. The second part is immediate, as $\vm_0 + \vm_1 = m.$
\end{IEEEproof}
The above-mentioned naive calculation (assuming independence) can now be rigorously justified, and we can establish the sizes of the disguised individuals \whp~as follows.

\begin{lemma}\label{size_01+} Given $n$ and $k = n^\theta$ as well as $\Gamma = \Theta(1)$ \mhk{and $\Delta \ge 2$}, we have \whp~that \mhk{$\abs{\zeroplus(\cGg)} = o(k)$}.
\end{lemma}
\begin{IEEEproof}
By the definition of $\cGg$ via the configuration model, \Lem~\ref{Lem_m0_gamma} guarantees that the total number of edges connected to a positive test is, with probability at least $1 - o(n^{-2})$, given by
\begin{align}
     \vm_1 \Gamma = \bc{ 1 + O\bc{n^{-\Omega(1)}}} m \Gamma\bc{1-\bc{1- k/n}^{\Gamma}}. \label{Eq_pos_halfedges}
\end{align}
Let $x$ be an uninfected individual. We can calculate the probability of $x$ belonging to $\zeroplus(\cGg)$ \geb{(i.e., being disguised and uninfected)} as follows: Each of the $\Delta = \Theta(1)$ edges\footnote{\geb{By counting degrees, we have $n\Delta=m\Gamma$, so the assumption $\Gamma=\Theta(1)$ leads to $m=\Theta(n\Delta)$.  Since $\Delta$ is integer-valued and we are considering $m > 0$ and $m \le n$ (otherwise, individual testing would be preferred), it follows that $\Delta = \Theta(1)$.}}
that are mapped to $x$ in the configuration model have to be connected to a positive test. Thus, by \eqref{Eq_pos_halfedges} along with Claim~\ref{stirling_applied}, we obtain
\begin{align*}
\Pr & \bc{ x \in \zeroplus(\cGg) \mid x \in \zero(\cGg),\vm_1} \\ & \qquad = \binom{\vm_1 \Gamma}{\Delta} \binom{m\Gamma}{\Delta}^{-1}  = \bc{ 1 + O\bc{n^{-\Omega(1)}}} \bc{1-\bc{1- k/n}^{\Gamma} }^{\Delta} \\ & \qquad = \mhk{O \bc{ \bc{ \frac{k}{n} }^{\Delta} }}.
\end{align*}
Therefore,
\begin{align}
    \Erw \brk{ \abs{\zeroplus(\cGg)} } = \mhk{  O \bc{ (n - k) \bc{ \frac{k}{n} }^{\Delta}} } = \mhk{ O \bc{ k \bc{ \frac{k}{n} }^{\Delta - 1}}  = o(k) }. \label{Eq_FirstMoment}
\end{align}
\mhk{Combining \eqref{Eq_FirstMoment}, $\Delta \geq 2$, and Markov's inequality, we obtain the assertion of \Lem~\ref{size_01+}.}
\end{IEEEproof}
Next, we define the event 
\begin{align}
    \cF_{\Gamma} = \cbc{\vm_1 = (1 + o(1)) m \bc{ 1 - (1 - k/n)^{\Gamma} }} \cap \cbc{ \abs{\zeroplus( \cGg )} = \mhk{o(k)} }, \label{Eq_F}
\end{align}
in which the number of positive tests and disguised uninfected individuals behave as expected.
By \Lem s \ref{Lem_m0_gamma} and \ref{size_01+}, we have $\Pr \bc{ \cF_{\Gamma} } \geq 1 - o(1)$. We assume without loss of generality that the first  $\vm_1$ tests render a positive result.

Letting $$\cD_{\Gamma} = \cbc{ \sum_{i=1}^{\vm_1} \vH_i^{1} = k \Delta, \quad \sum_{i=1}^{\vm_1} \vH_i^{0+} = \abs{\zeroplus(\cGg)} \Delta } $$ be the event that $\vH = \sum_{i=1}^{\vm_1}\vH_i$ equals its expectation, we have the following analog of Corollary \ref{Cor_DistEqual}.  

\begin{claim}\label{eq_dist2}
The distribution of $\vR_i$ equals the distribution of $\vH_i$ given $\cD_{\Gamma}$ and $\Gamma$, and furthermore, $\Pr(\cD_{\Gamma})=\Omega(n^{-1})$.
\end{claim}

\begin{IEEEproof}[Proof of Claim~\ref{eq_dist2}]
Let $(r_i)_{i\in [\vm_1]}$ be a sequence such that $r_i=(r_i^{1},r_i^{0+},r_i^{0-})$ and $\sum_{i} r_i^{1} = k \Delta, \sum_{i} r_i^{0+} = \abs{\zeroplus(\cGg)} \Delta$, and $r_i^{0-} = \Gamma - r_i^{1} - r_i^{0+} $. Let
\begin{align*}
S_1 & = k\Delta, \qquad S_{0+} = \Delta \abs{\zeroplus(\cGg)}  \qquad \text{and} \\
S_{0-} & =n\Delta-n\Delta(1-(1-k/n)^{\Gamma})-k\Delta.
\end{align*}
By the definition of $\vR_i$, we have
\begin{align*}
    \Pr & (\forall i\in [\vm_1]:\vR_i=r_i\mid \abs{\zeroplus(\cGg)},\vm_1) \\ & =\frac{\binom{S_1}{r_1^1...r_{\vm_1}^1}\binom{S_{0+}}{r_1^{0+}...r_{\vm_1^{0+}}}\binom{S_{0-}}{\Gamma-r_1^{1}-r_1^{0+}...\Gamma-r_{\vm_1}^{1}-r_{\vm_1}^{0+}}}{\binom{n\Delta}{\Gamma,...,\Gamma}} \\ & =\binom{n\Delta}{S_1,S_{0+},S_{0-}}^{-1}\prod_{i=1}^{\vm_1}\binom{\Gamma}{r_i^{1},r_i^{0+},r_i^{0-}}.
\end{align*}
Letting $(r^{'}_i)_{i\in [\vm_1]}$ be a second sequence as above, it follows that
\begin{align}\label{Eq_ratioY}\frac{\Pr(\forall i\in [\vm_1]:\vR_i=y_i\mid \abs{\zeroplus(\cGg)},\vm_1)}{\Pr(\forall i\in [\vm_1]:\vR_i=y^{'}_i\mid \abs{\zeroplus(\cGg)},\vm_1)}=\prod_{i=1}^{\vm_1}\frac{\binom{\Gamma}{r^{1}_i r^{0+}_i r^{0-}_i}}{\binom{\Gamma}{(r')^{1}_i (r')^{0+}_i (r')^{0-}_i}}.\end{align}
Furthermore, by the definition of $\vX$, we have 
\begin{align}
    & \frac{\Pr(\forall i\in [\vm_1]:\vH_i=r_i\mid \abs{\zeroplus(\cGg)},\vm_1,\cD_{\Gamma})}{\Pr(\forall i\in [\vm_1]:\vH_i=r^{'}_i\mid \abs{\zeroplus(\cGg)},\vm_1,\cD_{\Gamma})} \notag \\
     & = \frac{ \bc{\frac{k}{n}}^{\sum_{i=1}^{\vm_1} r_i^1}  \bc{\frac{\abs{\zeroplus(\cGg)}}{n}}^{\sum_{i=1}^{\vm_1} r_i^{0+}} \bc{ \frac{n - k - \abs{\zeroplus(\cGg)}}{n} }^{\sum_{i=1}^{\vm_1} r_i^{0-}} }{\bc{\frac{k}{n}}^{\sum_{i=1}^{\vm_1} (r')_i^1}  \bc{\frac{\abs{\zeroplus(\cGg)}}{n}}^{\sum_{i=1}^{\vm_1} (r')_i^{0+}} \bc{ \frac{n - k - \abs{\zeroplus(\cGg)}}{n} }^{\sum_{i=1}^{\vm_1} (r')_i^{0-}} } \notag \\ & \qquad \cdot \prod_{i=1}^{\vm_1} \frac{ \binom{\Gamma}{r_i^{1},r_i^{0+},r_i^{0-}} }{\binom{\Gamma}{(r')_i^{1},(r')_i^{0+},)(r')_i^{0-}}}  =  \prod_{i=1}^{\vm_1} \frac{\binom{\Gamma}{r_i^{1},r_i^{0+},r_i^{0-}}}{\binom{\Gamma}{(r')_i^{1},(r')_i^{0+},(r')_i^{0-}}} \label{Eq_X_ratio}.
\end{align}
The \geb{first part of the } claim follows from Equations \eqref{Eq_ratioY} and \eqref{Eq_X_ratio}. \geb{The probability follows by applying Claim~\ref{Prob_nr_inf_cor} for $\Delta=\Theta(1)$}
\end{IEEEproof}

We are interested in the number of positive tests that contain exactly one infected individual and no elements  of $\zeroplus(\cGg)$. Therefore, we define
\begin{align*}
   \vec{B} = \sum_{i=1}^{\vm_1} \vecone \cbc{ \vR_i^{1} + \vR_i^{0+} = 1} \quad \text{and} \quad \vec{B}' =  \sum_{i=1}^{\vm_1} \vecone \cbc{ \vH_i^{1} + \vH_i^{0+} = 1}.
\end{align*}
\begin{claim}\label{claim_B}
We have \whp~that

$$\vec B\leq \Delta k\bc{1-\geb{O\bc{\Gamma n^{-(1-\theta)}}}}$$
\end{claim}

\begin{IEEEproof}[Proof of Claim~\ref{claim_B}]
We use Claim \ref{eq_dist2} to simulate $\vec{B}$ through independent random variables as in $\vec B'$. Since $\vec{B}'$ is a sum of independent multinomial variables, we obtain its expectation by applying \eqref{Eq_F}, \Lem~\ref{Lem_Stirling_Approx} and Bayes Theorem:
\begin{align}
     \Erw & \brk{ \vec{B}' \mid  \abs{\zeroplus( \cGg )},\vm_1 } \notag \\ & = \sum_{i = 1}^{ \vm_1} \Pr \bc{ \vH_i = (1,0,\Gamma - 1)\mid  \abs{\zeroplus( \cGg )} }\nonumber\\
     &= \vm_1 \Gamma \frac{k/n \cdot \bc{ 1 - (k + \abs{\zeroplus(\cGg)})/n }^{\Gamma - 1}}{1 - (1 -k/n)^\Gamma} \nonumber\\
     &=\bc{1 + O\bc{\Gamma \frac{k}{n}}}\vm_1 \bc{ 1 - \frac{k + \abs{\zeroplus(\cGg)}}{n} }^{\Gamma - 1},
\end{align}
where the last step follows from \Lem~\ref{Bernoulli} and $\Gamma = \Theta(1)$.
Conditioning on $\cF_{\Gamma}$ defined in \eqref{Eq_F}, we obtain 
\begin{align}
     \Erw  & \brk{ \vec{B}' \mid   \cF_{\Gamma} } \notag \\ &= \geb{\bc{1 + O\bc{\frac{ \Gamma k}{n}}}} \frac{m \Gamma k}{n} \cdot \bc{ 1 - \frac{k + {o(k)} -{O\bc{n^{-\Omega(1)}}}}{n} }^{\Gamma-1} \notag \\
     &= \bc{1 + O\bc{\frac{\Gamma k}{n}}} \frac{m \Gamma k}{n} \cdot \bc{ 1 - (\Gamma - 1) \bc{\frac{k + {o(k)} -\geb{O\bc{n^{-\Omega(1)}}}}{n} }  }\notag \\
     &= \bc{1 + O\bc{\frac{\Gamma k}{n}}} \Delta k \bc{ 1 - (\Gamma - 1)  n^{-(1 - \theta)}-\mhk{o \bc{ n^{-(1 - \theta)} } }}\notag\\  
     & =\Delta k\bc{1+O\bc{\Gamma n^{-(1-\theta)}}}, \label{Eq_CondExpA}
\end{align}
where the first line uses Lemma~\ref{size_01+}, the second line uses Claim~\ref{Bernoulli} , and we additionally recall that $k=n^{\theta}$,  $\Delta=\frac{m\Gamma}{n}$, and $\Gamma = \Theta(1)$.
Moreover, since $\vec{B}'$ is a binomial random variable, the Chernoff bound  (\Lem~\ref{Lem_Chernoff}) 
yield with probability $o(n^{-10})$ that
    $$ \vec B' \leq \Delta k\bc{1+O\bc{\Gamma n^{-(1-\theta)}}}. $$
Thus, \geb{similar to the proof of Lemma \ref{lem_m0_gdelta}}, by Claim \ref{eq_dist2} we have \whp~that
\begin{align}\label{SizeB2}
   & \vec B \leq \Delta k\bc{1+O\bc{\Gamma n^{-(1-\theta)}}}.
\end{align}
\end{IEEEproof}

We are now in a position to characterize $\vA=\abs{ \one(\cG) \setminus \oneminusminus(\cGg) }$. 
\begin{claim}
\label{Claim_ExpectationA_Stirling}Given  $\vec B\leq \Delta k\bc{1 - O\bc{\Gamma n^{-(1-\theta)}}}$, we have for some constant $C > 0$ that 
\begin{align}\label{First_moment_A}
  \Erw \brk{ \vA\mid \vec{B}, \cF_{\Gamma} } &= k \binom{k\Delta - \vec{B}}{\Delta} \binom{k \Delta}{\Delta}^{-1} \leq k (C\cdot\Gamma)^{\Delta}n^{-(1-\theta)\Delta}
\end{align}
\end{claim}

\begin{IEEEproof}[Proof of Claim~\ref{Claim_ExpectationA_Stirling}]
The combinatorial expression follows by adding $k$ probabilities, one per defective item.  Each probability is the probability that an infected individual does not belong to $\oneminusminus$, which equals the probability that all of its $\Delta$ connections are disjoint from the $k\Delta - \vec B$ connections to tests in which it would have been the only infected individual with no disguised uninfected individuals.
The assertion then follows by combining the assumption $\vec B\leq \Delta k\bc{1 - O\bc{\Gamma n^{-(1-\theta)}}}$ with Claim~\ref{stirling_applied}.
\end{IEEEproof}

\begin{IEEEproof}[Proof of \Lem~\ref{size_1--}]
We distinguish between $\theta / (1 - \theta) \not \in \ZZ$ and $\theta / (1 - \theta) = T \in \ZZ$, and recall $m_{\DD}$ from \eqref{Eq_mdd} with $\Delta = \max \cbc{2,1+\floor{\frac{\theta}{1-\theta}} }$.  For simplicity, we assume that the inequality  $m \ge m_{\DD}$ holds with equality, but the general case is analogous.

 \noindent \textbf{Case A: $\theta / (1 - \theta) \not \in \ZZ$.} In this case, for $m = m_{\DD}$, we have $\Delta = \max \cbc{2,1+\floor{\frac{\theta}{1-\theta}} } =  \max \cbc{2,\ceil{\theta/(1 - \theta)} }$. 
 We distinguish the two cases $\theta < 1/2$ and $\theta > 1/2$ as follows:
 \begin{itemize}
     \item \textbf{Case A1: $\theta > 1/2$.} 
    In this case, we have $\Delta = \ceil{\theta/(1 - \theta)}$.
    Defining $ \eta = \theta - (1 - \theta) \cdot \ceil{\theta/(1 - \theta)} < 0$, 
    using \eqref{First_moment_A} and $\Gamma, \Delta = \Theta(1)$, we find
    \begin{align}
        \label{Eq_EW-A} \Erw\brk{\vA\mid \vec{B}, \cF_{\Gamma}} \leq O(1) n^{\theta - (1 - \theta) \cdot \ceil{\theta/(1 - \theta)} } = O(n^{\eta}).
    \end{align} 
    \item \textbf{Case A2: $\theta < 1/2$.} In this case, we have $\Delta = 2$, and hence
    \begin{align}
        \label{Eq_EW-A2} \Erw \brk{\vA\mid \vec{B}, \cF_{\Gamma} } \leq O(1) \Gamma^{\Delta} n^{3 \theta - 2 } \leq o(1).
    \end{align} 
 \end{itemize}
\noindent \textbf{Case B: $\theta / (1 - \theta) = T \in \ZZ$.}  Again, we distinguish the cases $\theta = 1/2$ and $\theta > 1/2$:
\begin{itemize}
    \item \textbf{Case B1: $\theta > 1/2$.} We have $\Delta = T + 1$, so by  
    \eqref{First_moment_A} and $\Gamma, \Delta = \Theta(1)$, we find
    \begin{align}
        \label{Eq_EW-A3} \Erw\brk{\vA\mid \vec{B}, \cF_{\Gamma}} \leq O(1) n^{\theta - (1 - \theta) \cdot (T + 1) } = O(n^{-(1 - \theta)}),
    \end{align} 
    where the last step uses $1-\theta \ge \theta$ and $T > 1$.
    \item \textbf{Case B2: $\theta = 1/2$.} We have $\Delta = 2$, and hence 
    \begin{align}
        \label{Eq_EW-A4} \Erw[\vA\mid \vec{B}, \cF_{\Gamma}] \leq O(n^{-1/2}).
    \end{align} 
\end{itemize}

Combining \eqref{Eq_EW-A}--\eqref{Eq_EW-A4} with Markov's inequality and the fact that $\cF_{\Gamma}$ occurs \whp, we deduce that $\vA = 0$ \whp, completing the proof of \Lem~\ref{size_1--}.
\end{IEEEproof}

Theorem ~\ref{Thm_DDg} now follows directly by combining \Lem~\ref{size_1--} and Claim \ref{claim_dd_individualtypes}.  
So far, we have addressed the case where the test design is formed using the configuration model, and showed that the \DD-algorithm is optimal in this regime if applied to the random regular pooling scheme $\cGg$.  However, the preceding analysis does not provide a tight bound for the matching-based design.

\subsection{Algorithmic feasibility II: Matching-based model}\label{algo_sparse}
Recall from from Section~\ref{pool_gamma} that the matching-based model with parameter $\gamma$ is denoted by $\cGg^*$.  \geb{While the \DD~ algorithm does not appear to be optimal in this case, it turns out that turning to SCOMP (a slight refinement of DD) suffices for optimality.}

\begin{theorem}
\label{thm_dd_gamma_sparse_optimal}
If $m \ge 2 n/(\Gamma + 1)$ and $0 < \theta < 1/2$, then \whp~\SCOMP~ recovers $\SIGMA$  from $\cGg^*$ and $\hat \SIGMA$.
\end{theorem}
\subsubsection{Proof of \Thm~\ref{thm_dd_gamma_sparse_optimal}}
{We prove the theorem for $m = 2 n/(\Gamma + 1)$ (which implies $\gamma = \frac{2}{\Gamma+1}n$), but the more general case follows analogously; intuitively, a higher number of tests can only help.}
 We analyse the \DD~ algorithm on $\cGg^*$ in two steps, starting with the regular part of the graph. Denote by $\cGg^{*,r}$ the $(\Gamma-1, 2)$ regular part, in which we select $n-\gamma$ individuals and pool them into two tests each. Denote by $\SIGMA[\cGg^{*,r}]$ and $\hat \SIGMA[\cGg^{*,r}]$ the infection status vector and outcome vector resulting from the regular part alone. 
\begin{lemma}
\label{lem_dd_sparse_on_regular_part}
If $m \geq 2 n/(\Gamma+1)$, then \whp~\DD~ recovers $\SIGMA[\cGg^{*,r}]$ from $(\cGg^{*,r}, \hat \SIGMA[\cGg^{*,r}])$ correctly.
\end{lemma}
\begin{IEEEproof}
    This follows from \Thm~\ref{Thm_DDg}, as $\cGg^{*,r}$ is identically distributed with $\cG_{\Gamma-1}$ therein. With $\gamma = \frac{2}{\Gamma + 1}n$ individuals removed from the population, we have $n' = \frac{\Gamma - 1}{\Gamma + 1} n$ individuals being tested in $\cGg^{*,r}$. Thus, we require at most $m' = 2 \frac{n'}{\Gamma-1} = 2 \frac{n}{\Gamma + 1}$ tests in order for \DD~to succeed \whp~on $\cGg^{*,r}$.
\end{IEEEproof}

It remains to handle the second step, and specifically, argue that after adding the $\gamma = 2 \frac{n}{\Gamma + 1}$ individuals (one to each test) \geb{we can guarantee the success of SCOMP.}
We denote by $k'$ the number of infected individuals under the remaining $n'$ individuals, and let $\theta'$ be the value such that $k' = \Theta((n')^{\theta'})$, which is well-defined due to the following.

\begin{claim}
\label{claim_theta_not_changing}
Under the preceding setup, we have \whp~that $\theta' = \theta$.
\end{claim}
\begin{IEEEproof}
As we remove $\gamma = \frac{2}{\Gamma+1}n$ individuals randomly, the number of infected individuals in the remaining part is a hypergeometrically distributed random variable $\vec K' \sim H \bc{n, k, n'}$. Thus, the Chernoff bound for the hypergeometric distribution guarantees \whp~that $$\vec K' = (1 + o(1)) k n'/n = (1 + o(1)) \frac{\Gamma - 1}{\Gamma + 1} k, $$
and the assertion follows.
\end{IEEEproof}
In the second step, we analyse the remaining part of the graph, in which the $\gamma$ remaining individuals are placed into one test each.  To do so, the following lemma turns out to be useful.

\geb{
\begin{lemma} \label{lem:dist4}
    Under the matching-based model $\cGg^*$ with $\theta < \frac{1}{2}$, it holds \whp~that there are no two infected individuals within distance 4 in the graph.
\end{lemma}
\begin{IEEEproof}
    By construction, it holds with probability one that $\cGg^*$ has individual-degree at most two, and test-degree at most $\Gamma = \Theta(1)$.  Hence, all degrees are bounded.  This means that for any given individual $x$, the set of individuals $x'$ with $\dist(x,x') \le 4$ has size $O(1)$.  For any two individuals $x$ and $x'$, the probability of both being infected is $O((k/n)^2)$, and a union bound over the $O(n)$ possible pairs with $\dist(x,x') \le 4$ increases this probability to $O(n (k/n)^2)$.  The assumption $\theta < \frac{1}{2}$ implies that $k = o(\sqrt n)$, and thus, we have $O(n (k/n)^2) = o(1)$, which establishes the lemma.
\end{IEEEproof}
}

We now combine the preceding lemmas to establish the success of the \DD~algorithm.

\begin{lemma}
\label{lem_dd_sparse_on_regular_complete}
Conditioned on the \DD~algorithm recovering $\SIGMA[\cGg^{*,r}]$ from $(\cGg^{*,r}, \hat \SIGMA[\cGg^{*,r}])$, and on all infected individuals having pairwise distance exceeding 4, it holds with conditional probability one that the \geb{SCOMP} algorithm recovers $\SIGMA$ from $(\cGg^*, \hat \SIGMA)$.
\end{lemma}
\begin{IEEEproof}	
        By the construction of $\cGg^*$, there are $\gamma = \frac{2}{\Gamma + 1} n $ individuals added to $\cGg^{*,r}$ to produce $\cGg^*$. Denote the set of these individuals by $X = \cbc{x_1 \ldots x_\gamma}$. As $\gamma \leq m$, there is a matching from $X$ to the the $m$ tests. 
         
         Having assumed success on the regular part $\cGg^{*,r}$, we only need to show that the newly added individuals in $X$ are also correctly identified, and additionally do not impact the identifications in $\cGg^{*,r}$.  \geb{Recall from Claim \ref{claim_dd_individualtypes} that DD succeeds if and only if all infected individuals are easy infected (i.e., are in $\oneminusminus(\cGg^*)$), and recall also that the success of DD implies the success of SCOMP \cite{Aldridge_2017a}.} We distinguish four different cases, which are illustrated in Figure \ref{fig_fourcases}.

        \begin{figure*}
        \begin{minipage}[t]{0.19 \linewidth}
        \begin{tikzpicture}
        \node[] (C1) at (3,-4) {\underline{A1}};
        \node[shape=rectangle,draw=black, minimum width = 0.5cm, minimum height = 0.5cm] (T1) at (3,0) {0};
        \node[shape=circle,draw=black, minimum width = .5cm, minimum height = .5cm] (O11) at (2,-1) {0};
        \node[shape=circle,draw=black, minimum width = .5cm, minimum height = .5cm] (O21) at (4,-1) {0};
        \node[shape=circle,draw=black, fill = blue!20, minimum width = .5cm, minimum height = .5cm] (O31) at (3,1) {0};
        
        \node[shape=rectangle,draw=black, minimum width = .5cm, minimum height = .5cm] (P31) at (2,-2) {};
        \node[shape=rectangle,draw=black, minimum width = .5cm, minimum height = .5cm] (P32) at (4,-2) {};
        
        \node[shape=circle,draw=black, minimum width = .5cm, minimum height = .5cm] (P41) at (2,-3) {};
        \node[shape=circle,draw=black, minimum width = .5cm, minimum height = .5cm] (P42) at (4,-3) {};
        
        \path [-](T1) edge node[] {} (O11);
        \path [-](T1) edge node[] {} (O21);
        \path [-, dashed, color = blue](T1) edge node[] {} (O31);
        \path [-](P31) edge node[] {} (O11);
        \path [-](P32) edge node[] {} (O21);
        \path [-](P32) edge node[] {} (P42);
        \path [-](P31) edge node[] {} (P41);
        \path [-](P32) edge node[] {} (P42);
        \end{tikzpicture}
        \end{minipage}
        \begin{minipage}[t]{0.19 \linewidth}
        \begin{tikzpicture}
        \node[] (C1) at (3,-4) {\underline{A2}};
        \node[shape=rectangle,draw=black, minimum width = 0.5cm, minimum height = 0.5cm] (T1) at (3,0) {1};
        \node[shape=circle,draw=black, fill = yellow!50, minimum width = .5cm, minimum height = .5cm] (O11) at (2,-1) {0};
        \node[shape=circle,draw=black, fill = yellow!50, minimum width = .5cm, minimum height = .5cm] (O21) at (4,-1) {0};
        \node[shape=circle,draw=black, fill = blue!20, minimum width = .5cm, minimum height = .5cm] (O31) at (3,1) {1};
        
        \node[shape=rectangle,draw=black, minimum width = .5cm, minimum height = .5cm] (P31) at (2,-2) {};
        \node[shape=rectangle,draw=black, minimum width = .5cm, minimum height = .5cm] (P32) at (4,-2) {};
        
        \node[shape=circle,draw=black, fill = red!20, minimum width = .5cm, minimum height = .5cm] (P41) at (2,-3) {};
        \node[shape=circle,draw=black, fill = red!20, minimum width = .5cm, minimum height = .5cm] (P42) at (4,-3) {};
        
        \path [-](T1) edge node[] {} (O11);
        \path [-](T1) edge node[] {} (O21);
        \path [-, dashed, color = blue](T1) edge node[] {} (O31);
        \path [-](P31) edge node[] {} (O11);
        \path [-](P32) edge node[] {} (O21);
        \path [-](P32) edge node[] {} (P42);
        \path [-](P31) edge node[] {} (P41);
        \path [-](P32) edge node[] {} (P42);
        \end{tikzpicture}
        \end{minipage}
        \begin{minipage}[t]{0.19 \linewidth}
        \begin{tikzpicture}
        \node[] (C1) at (3,-4) {\underline{B1}};
        \node[shape=rectangle,draw=black, minimum width = 0.5cm, minimum height = 0.5cm] (T1) at (3,0) {1};
        \node[shape=circle,draw=black, fill=yellow!50, minimum width = .5cm, minimum height = .5cm] (O11) at (2,-1) {1};
        \node[shape=circle,draw=black, minimum width = .5cm, minimum height = .5cm] (O21) at (4,-1) {0};
        \node[shape=circle,draw=black, fill = blue!20, minimum width = .5cm, minimum height = .5cm] (O31) at (3,1) {1};
        
        \node[shape=rectangle,draw=black, minimum width = .5cm, minimum height = .5cm] (P31) at (2,-2) {};
        \node[shape=rectangle,draw=black, minimum width = .5cm, minimum height = .5cm] (P32) at (4,-2) {};
        
        \node[shape=circle,draw=black, minimum width = .5cm, minimum height = .5cm] (P41) at (2,-3) {};
        \node[shape=circle,draw=black, minimum width = .5cm, minimum height = .5cm] (P42) at (4,-3) {};
        
        \path [-](T1) edge node[] {} (O11);
        \path [-](T1) edge node[] {} (O21);
        \path [-, dashed, color = blue](T1) edge node[] {} (O31);
        \path [-](P31) edge node[] {} (O11);
        \path [-](P32) edge node[] {} (O21);
        \path [-](P32) edge node[] {} (P42);
        \path [-](P31) edge node[] {} (P41);
        \path [-](P32) edge node[] {} (P42);
        \end{tikzpicture}
        \end{minipage}
         \begin{minipage}[t]{0.19 \linewidth}
        \begin{tikzpicture}
        \node[] (C1) at (3,-5) {\underline{B2 (i)}};
        \node[shape=rectangle,draw=black, minimum width = 0.5cm, minimum height = 0.5cm] (T1) at (3,0) {1};
        \node[shape=circle,draw=black, fill = yellow!50, minimum width = .5cm, minimum height = .5cm] (O11) at (2,-1) {1};
        \node[shape=circle,draw=black, fill = yellow!50, minimum width = .5cm, minimum height = .5cm] (O21) at (4,-1) {0};
        \node[shape=circle,draw=black, fill = blue!20, minimum width = .5cm, minimum height = .5cm] (O31) at (3,1) {0};
        
        \node[shape=rectangle,draw=black, minimum width = .5cm, minimum height = .5cm] (P31) at (2,-2) {1};
        \node[shape=rectangle,draw=black, minimum width = .5cm, minimum height = .5cm] (P32) at (4,-2) {};
        
        \node[shape=rectangle,draw=black, minimum width = .5cm, minimum height = .5cm] (P51) at (2,-4) {0};
        \node[shape=rectangle,draw=black, minimum width = .5cm, minimum height = .5cm] (P52) at (4,-4) {};
        
        \node[shape=circle,draw=black, minimum width = .5cm, minimum height = .5cm] (P41) at (2,-3) {0};
        \node[shape=circle,draw=black, minimum width = .5cm, minimum height = .5cm] (P42) at (4,-3) {};
        
        \path [-](T1) edge node[] {} (O11);
        \path [-](T1) edge node[] {} (O21);
        \path [-, dashed, color = blue](T1) edge node[] {} (O31);
        \path [-](P31) edge node[] {} (O11);
        \path [-](P32) edge node[] {} (O21);
        \path [-](P32) edge node[] {} (P42);
        \path [-](P31) edge node[] {} (P41);
        \path [-](P32) edge node[] {} (P42);
        
        \path [-](P51) edge node[] {} (P41);
        \path [-](P52) edge node[] {} (P42);
        \end{tikzpicture}
        \end{minipage}
         \begin{minipage}[t]{0.19 \linewidth}
        \begin{tikzpicture}
        \node[] (C1) at (3,-5) {\underline{B2 (ii)}};
        \node[shape=rectangle,draw=black, fill = red!30, minimum width = 0.5cm, minimum height = 0.5cm] (T1) at (3,0) {1};
        \node[shape=circle,draw=black, fill = yellow!50, minimum width = .5cm, minimum height = .5cm] (O11) at (2,-1) {1};
        \node[shape=circle,draw=black, fill = yellow!50, minimum width = .5cm, minimum height = .5cm] (O21) at (4,-1) {0};
        \node[shape=circle,draw=black, fill = blue!20, minimum width = .5cm, minimum height = .5cm] (O31) at (3,1) {0};
        
        \node[shape=rectangle,draw=black, fill = red!30, minimum width = .5cm, minimum height = .5cm] (P31) at (2,-2) {1};
        \node[shape=rectangle,draw=black, minimum width = .5cm, minimum height = .5cm] (P32) at (4,-2) {};
        
        \node[shape=circle,draw=black, fill = blue!20, minimum width = .5cm, minimum height = .5cm] (N2) at (3,-2) {0};
        
        \node[shape=rectangle,draw=black, minimum width = .5cm, minimum height = .5cm] (P51) at (2,-4) {0};
        \node[shape=rectangle,draw=black, minimum width = .5cm, minimum height = .5cm] (P52) at (4,-4) {};
        
        \node[shape=circle,draw=black, minimum width = .5cm, minimum height = .5cm] (P41) at (2,-3) {0};
        \node[shape=circle,draw=black, minimum width = .5cm, minimum height = .5cm] (P42) at (4,-3) {};
        
        \path [-](T1) edge node[] {} (O11);
        \path [-](T1) edge node[] {} (O21);
        \path [-, dashed, color = blue](T1) edge node[] {} (O31);
        \path [-, dashed, color = blue](P31) edge node[] {} (N2);
        \path [-](P31) edge node[] {} (O11);
        \path [-](P32) edge node[] {} (O21);
        \path [-](P32) edge node[] {} (P42);
        \path [-](P31) edge node[] {} (P41);
        \path [-](P32) edge node[] {} (P42);
        
        \path [-](P51) edge node[] {} (P41);
        \path [-](P52) edge node[] {} (P42);
        \end{tikzpicture}
        \end{minipage}
        \caption{The cases considered in our analysis.  Round vertices are items, and square vertices are tests. The blue vertex is added in the second step of construction of $\cG_{\Gamma}^*$, and the labels inside the vertices indicate the defectivity status or test outcome after adding the blue vertex. 
        In case A1, recovery is clearly possible if and only if the same is true the remainder of the graph. In case A2, the yellow vertices may, in principle, no longer be identifiable as definite non-defectives.  This happens if and only if the corresponding red individual is infected, which in turn implies a length-4 path between defectives, contradicting a high-probability event that we show. 
        In case B1, there is a path of length two from the infected blue individual to the infected yellow individual, which is again a contradiction. In case B2(i), the infected yellow vertex can still be recovered as it is element of $\oneminusminus$ in the regular part and will be recovered successfully during the first two steps of \SCOMP. In case B2(ii), the two red tests could either be explained by the yellow infected individual or by the two blue (uninfected) individuals, and due to its greedy selection rule, \SCOMP~ declares the yellow individual as infected and the blue individuals as uninfected. }
        \label{fig_fourcases}
        \end{figure*}
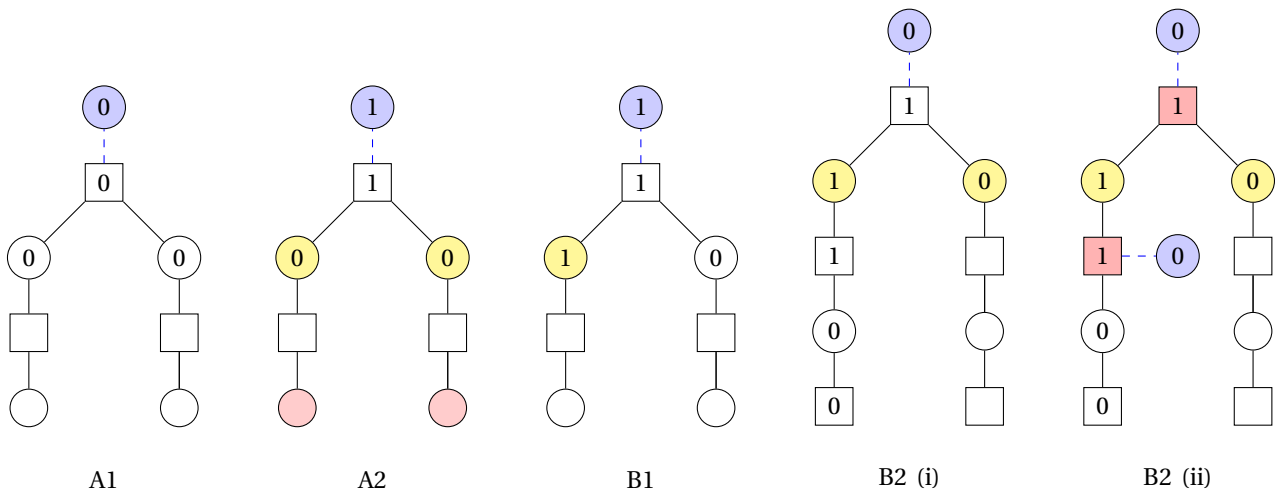
        
        \noindent \textbf{Case A: Connecting to a negative test.} Suppose that an individual $x \in X$ connects to a (previously) negative test $a$.  Then, for all $y \in \partial_{ \cGg^{*, r}}(a)$ we have $y \in \zerominus(\cGg^{*, r})$. 
        \begin{itemize}
            \item \textbf{Case A-1: $\SIGMA_x = 0$.} \geb{If $x$ is uninfected and connects to a negative test, then the test remains negative.  It follows immediately that $x \in \zerominus(\cGg^*)$ (i.e., $x$ is easy uninfected), which further implies that all other individuals in the test that were previously easy uninfected or easy infected in $\cGg^{*, r}$ remain so in $\cGg^*$, as desired.}
            \item \textbf{Case A-2: $\SIGMA_x = 1$.} \geb{In this case, we have $\hat \SIGMA_a({ \cGg^{*, r}}) = 0$ but $\hat \SIGMA_a(\cGg^*) = 1$.  To maintain success, we need to show that all $y \in \partial_{ \cGg^{*, r}}(a)$ (which were previously easy uninfected) remain easy uninfected in $\cGg^*$; this implies both that previous decisions are not affected, and that the decision for $x$ is correct due to $x \in \oneminusminus(\cGg^*)$.  To establish that each $y \in \partial_{ \cGg^{*, r}}(a)$ is easy uninfected, we argue that the second test that $y$ belongs to is negative.  Indeed, suppose for contradiction that $y$ is in another positive test $a'$ with an infected individual $x'$.  Then, there is a path of length 4 in $\cGg^*$ from $x$ to $a$ to $y$ to $a'$ to $x'$, and this contradicts Lemma \ref{lem:dist4}.}
        \end{itemize}
        \noindent \textbf{Case B: Connecting to a positive test.} Suppose that an individual $x \in X$ connects to a (previously) positive test $a$. Therefore, there exists at least one $y \in V_1({ \cGg^{*, r}}) \cap \partial_{ \cGg^{*, r}}(a)$. As \DD~succeeds on $\cGg^{*, r}$ by assumption, we have  $y \in \oneminusminus({ \cGg^{*, r}})$.
        \begin{itemize}
            \item \textbf{Case B-1: $\SIGMA_x = 1$.} \geb{This case does not occur, because it implies a length-2 path from $x$ to $y$, both of which are infected, in contradiction with the lemma assumption.}
            \item \textbf{Case B-2: $\SIGMA_x = 0$.} \geb{Since the first two steps of SCOMP (Algorithm \ref{dd_algorithm}) never make mistakes, the only way that an error can occur in this case is that (i) $x$ is added in some step of the final (sequential greedy) step, or (ii) $y \notin \oneminusminus(\cGg^*)$ and $y$ fails to be chosen throughout the final step.  We argue that neither of these events occur.  To see this, first note that in $\cGg^{*,r}$, $y$ is not only part of $\oneminusminus({ \cGg^{*, r}})$ because of $a$, but also because the second test that $y$ belongs to consists only of $y$ and individuals from $\zerominus(\cGg^*)$:  If this were not the case, then we could create a path from $y$ to another infected individual using a path of length at most 4.  We then have the following:
            \begin{itemize}
                \item If $y \in \oneminusminus(\cGg^*)$ then $y$ is trivially decoded correctly, and $x$ is certainly not added in the final step (since its only test is already explained).
                \item If $y \notin \oneminusminus(\cGg^*)$ then the two tests containing $y$ are unexplained at the start of the final step.  Due to the above-established property of both of these tests leading to $y \in \oneminusminus({ \cGg^{*, r}})$ in the regular part, we have that in $\cGg^*$, only $y$ and/or the newly added elements of $X$ can explain these two tests.  But since $y$ explains both of them, but the elements of $X$ can only explain one each (since their degree is one), it is clearly $y$ (and not $x$) that will be chosen, as desired.
            \end{itemize}}
            
        \end{itemize}
\end{IEEEproof}

We now have all the ingredients to prove \Thm~\ref{thm_dd_gamma_sparse_optimal}.

\begin{IEEEproof}[Proof of \Thm~\ref{thm_dd_gamma_sparse_optimal}]
By construction, $\cGg^*$ consists of $n$ individuals and $m = 2 n / (\Gamma + 1)$ tests. By \Lem~\ref{lem_dd_sparse_on_regular_part}, this $m$ suffices for \DD~to succeed \whp~on the regular part of $\cGg^*$ (i.e., on $\cGg^{*,r}$).  \geb{In addition, Lemma \ref{lem:dist4} gives the convenient distance-4 property \whp, and \Lem~\ref{lem_dd_sparse_on_regular_complete} guarantees that the preceding two findings suffice to ensure that SCOMP infers $\SIGMA$ correctly from $\cGg^*$ and $\hat \SIGMA$.  Hence, the theorem follows.}
\end{IEEEproof}

\subsection{Putting the pieces together} \label{sec:pieces}
\Thm~\ref{Thm_DDg} proves that \DD~ succeeds on the bi-regular graph $\cGg$ created by the configuration model using $\max \cbc{2,1+\floor{\frac{\theta}{1-\theta}} } $ tests, \geb{and hence so does SCOMP \cite{Aldridge_2017a}}. Furthermore, as \Thm~\ref{thm_dd_gamma_sparse_optimal} shows, for $\theta < 1/2$, $\frac{2 n}{\Gamma + 1}$ tests suffice employing $\cGg^*$ \geb{and using SCOMP.}


Finally, we show that the results of \Thm ~ \ref{Thm_DDg} and \Thm~\ref{thm_dd_gamma_sparse_optimal} combine to match the information-theoretic lower bound \eqref{minf_gamma}, i.e., $\max \cbc{ \bc{1 + \floor{\frac{\theta}{1-\theta}}} \frac{n}{\Gamma}, 2 \frac{n}{\Gamma + 1}}$. On the one hand, for $\theta<\frac{1}{2}$, the lower bound simplifies to the desired quantity $\frac{2 n}{\Gamma+1}$ due to the fact that $\floor{\frac{\theta}{1-\theta}}=0$ in this regime, and $\frac{2}{\Gamma+1} \ge \frac{1}{\Gamma}$ for $\Gamma \ge 1$.  On the other hand, if $\theta \ge \frac{1}{2}$ then we have $\floor{\frac{\theta}{1-\theta}} \ge 1$, and so the maximum in the lower bound is achieved by the first term (since $\frac{1}{\Gamma} \ge \frac{1}{\Gamma+1}$), thus again matching the upper bound.  Hence, the \geb{SCOMP} algorithm is information-theoretically optimal when used with the pooling scheme $\tilde{\cGg}$.

\section{Adaptive Group Testing with $\Delta$-Divisible Individuals} \label{sec:adaptive}

In this section, we turn to adaptive testing strategies in the case of $\Delta$-divisible individuals, and demonstrate that in certain cases the number of tests can be reduced significantly.

\subsection{Converse}

Recall that the converse bound proved in Theorem \ref{thm:converse} already considered adaptive test designs.  Thus, any adaptive strategy fails \whp when $m \le (1-\epsilon) \eul^{-1} \Delta k^{1+\frac{(1-\theta)}{\Delta \theta}}$ for fixed $\epsilon > 0$.

%
%
\subsection{Algorithm}

We present an algorithm that can be viewed as an analog of Hwang's binary splitting algorithm \cite{Hwang_1972}, instead using {\em non-binary} splitting in order to ensure that each item is in at most $\Delta$ tests.  Like with Hwang's algorithm, we assume that the size $k$ of the infected set is known.  In the case case that only an upper bound $k_{\rm max} \ge k$ is known, the same analysis and results apply with $k_{\rm max}$ in place of $k$.  
However, such bounds may somewhat loose, and care should be taken in using initial tests to estimate $k$ as an initial step (e.g., see \cite{Dam10,Fal16,Bshouty_2018}), as this may use a significant portion of the $\Delta$ budget.  For clarity, we only consider the case of known $k$ in this section, and leave the case of unknown $k$ to future work (see also \cite{tan_2020} for some initial findings).

\subsubsection{Recovering the infected Set}
\setlength{\algomargin}{0pt}
Our adaptive algorithm is described in Algorithm \ref{alg:adaptive_algo},
\begin{algorithm}[t]
    \begin{algorithmic}[1]
        \REQUIRE Number of individuals $n$, number of infected individuals $k$, and divisibility of each individual $\Delta$
        \STATE Initialise $\widetilde{n}\leftarrow \big(\frac{n}{k}\big)^{\frac{\Delta-1}{\Delta}}$ and the estimate $\widehat{\mathcal{K}}\leftarrow\emptyset$
        \STATE Arbitrarily group the $n$ individuals into $n/\widetilde{n}$ groups of size $\widetilde{n}$
        \STATE Test each group and discard any that return negative
        \STATE Label the remaining groups incrementally as $G^{(0)}_j$, where $j=1,2,\dots$
        \FOR{$i=1$ to $\Delta-1$}{
        \FOR{each group $G^{(i-1)}_j$ from the previous stage}{
        \STATE Arbitrarily group all individuals in $G^{(i-1)}_j$ into $\widetilde{n}^{1/(\Delta-1)}$ sub-groups of size $\widetilde{n}^{1-i/(\Delta-1)}$
        \STATE Test each sub-group and discard any that return a negative outcome
        \STATE Label the remaining sub-groups incrementally as $G^{(i)}_j$
        }\ENDFOR
        }\ENDFOR
        \STATE Add the individuals from all of the remaining singleton groups $G^{(\Delta-1)}_j$ to $\widehat{\mathcal{K}}$
        \RETURN $\widehat{\mathcal{K}}$
    \end{algorithmic}

    \caption{Adaptive algorithm for $\Delta$-divisible individuals \label{alg:adaptive_algo}}
\end{algorithm}
where we assume for simplicity that $\big(\frac{n}{k}\big)^{1/\Delta}$ is an integer.\footnote{Note that we assume $k=o(n)$ and $\Delta=o\big(\log\big(\frac{n}{k}\big)\big)$, meaning that $\big(\frac{n}{k}\big)^{1/\Delta} \to \infty$. Hence, the effect of rounding is asymptotically negligible, and is accounted for by the $1+o(1)$ term in Theorem \ref{thm:gamma_upperbound_ours}.} Using Algorithm \ref{alg:adaptive_algo}, we have the following theorem, which is proved throughout the remainder of the subsection.  We define
\begin{equation}
    \mada(\Delta) = \Delta k^{1+\frac{1-\theta}{\theta \Delta}}.
\end{equation}

\begin{theorem} \label{thm:gamma_upperbound_ours}
    For $\Delta=o(\log n)$ and $k = n^{\theta}$ with $\theta \in (0,1)$, the adaptive algorithm in Algorithm \ref{alg:adaptive_algo} tests each individual at most $\Delta$ times and uses at most $\mada(\Delta)(1+o(1))$ tests to recover the infected set exactly with zero error probability.
\end{theorem}

\begin{IEEEproof}
    Similar to Hwang's generalised binary splitting algorithm \cite{Hwang_1972}, the idea behind the parameter $\widetilde{n}$ in Algorithm \ref{alg:adaptive_algo} is that when $k$ becomes large, having large groups during the initial splitting stage is wasteful, as it results in each test having a high probability of being positive (not very informative). Hence, we want to find the appropriate group sizes that result in more informative tests to minimise the number of tests. 
    \begin{figure}[t] 
      \centering
      \includegraphics[scale=0.35]{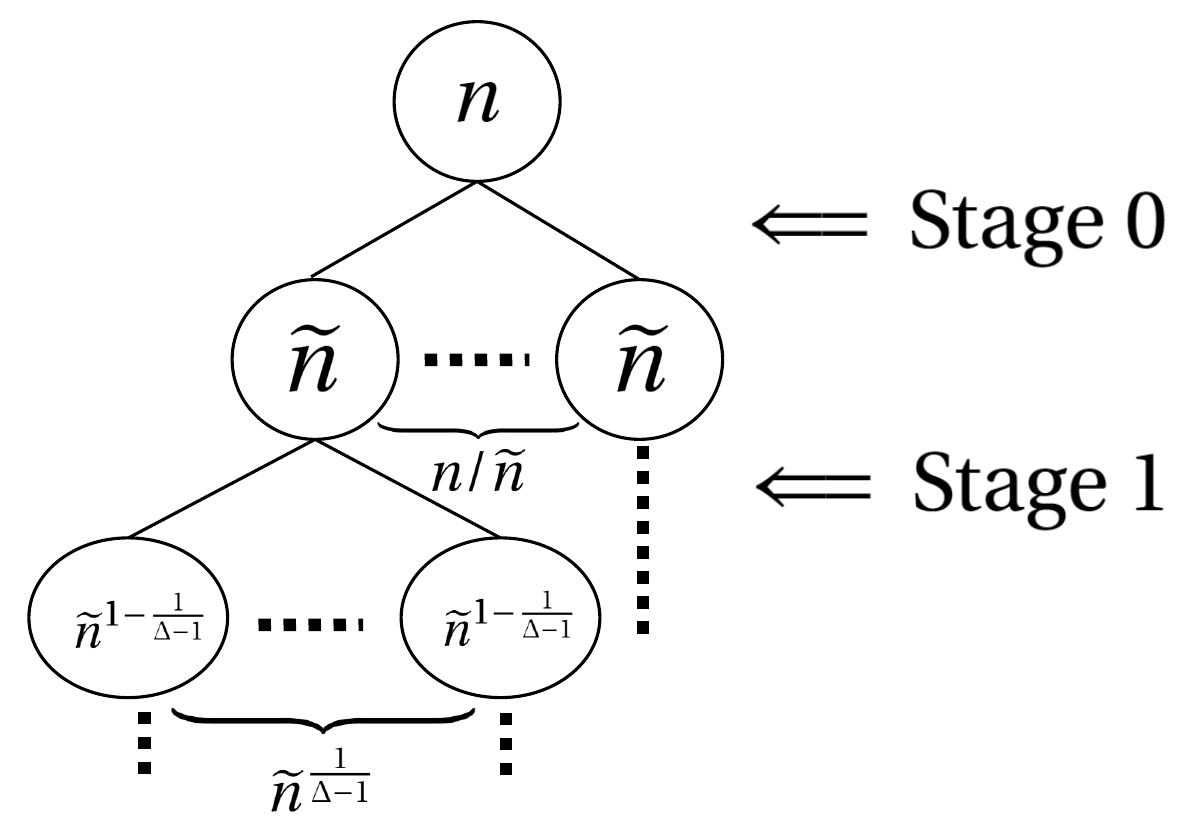}
      \caption{Visualization of splitting in the adaptive algorithm.}
      \label{fig:adaptive_algo}
    \end{figure}
    Each stage (outermost for-loop in Algorithm \ref{alg:adaptive_algo}) here refers to the process where all groups of the same sizes are split into smaller groups (e.g., see Figure \ref{fig:adaptive_algo}). We let $\widetilde{n}$ be the group size at the initial splitting stage of the algorithm. The algorithm first tests $n/\widetilde{n}$ groups of size $\widetilde{n}$ each,\footnote{Note that $n/\widetilde{n}$ is an integer for our chosen $\widetilde{n}$ below, which gives $\frac{n}{\widetilde{n}}=k\big(\frac{n}{k}\big)^{1/\Delta}$, and $\big(\frac{n}{k}\big)^{1/\Delta}$ was already assumed to be an integer.} then steadily decrease the sizes of each group down the stages: $\widetilde{n}\rightarrow \widetilde{n}^{1-1/(\Delta-1)}\rightarrow \widetilde{n}^{1-2/(\Delta-1)}\rightarrow\dots\rightarrow 1$ (see Figure \ref{fig:adaptive_algo}). Hence, we have $n/\widetilde{n}$ groups in the initial splitting and $\widetilde{n}^{\frac{1}{\Delta-1}}$ groups in all subsequent splits. 
    
    With the above observations, we can derive an upper bound on the total number of tests needed. We have $n/\widetilde{n}$ tests in the first stage. Since we have $k$ infected and split into $\widetilde{n}^{\frac{1}{\Delta-1}}$ sub-groups in subsequent stages, the number of smaller groups that each stage can produce is at most $k\widetilde{n}^{\frac{1}{\Delta-1}}$. This implies that the number of tests conducted at each stage is at most $k\widetilde{n}^{\frac{1}{\Delta-1}}$, giving the following bound on $m$:
    \begin{align}
        m &\leq \frac{n}{\widetilde{n}}+(\Delta-1)k\widetilde{n}^{\frac{1}{\Delta-1}} \label{eq:upper_bound}.
    \end{align}
    We optimise with respect to $\widetilde{n}$ by differentiating the upper bound and setting it to zero. This gives $\widetilde{n}=\big(\frac{n}{k}\big)^{\frac{\Delta-1}{\Delta}}=n^{\frac{(1-\theta)(\Delta-1)}{\Delta}}$, and substituting $\widetilde{n}=\big(\frac{n}{k}\big)^{\frac{\Delta-1}{\Delta}}$ into the general upper bound in \eqref{eq:upper_bound} gives the following upper bound:
    \begin{align}
        m\leq\frac{n}{(n/k)^{\frac{\Delta-1}{\Delta}}}+(\Delta-1)k\bigg(\Big(\frac{n}{k}\Big)^{\frac{\Delta-1}{\Delta}}\bigg)^{\frac{1}{\Delta-1}}
        =\Delta k\Big(\frac{n}{k}\Big)^{\frac{1}{\Delta}}= \Delta k^{1+\frac{1-\theta}{\theta \Delta}}. \label{eq:adaptive_upper_bound}
    \end{align}    
\end{IEEEproof}
\textit{Comparisons:} We observe that $\mada(\Delta)$ matches the universal lower bound in Theorem \ref{thm:converse} to within a factor of $\eul$ for all $\theta \in (0,1)$.  For $\theta < \frac{1}{2}$, we have $\mada(\Delta) = \mDD(\Delta) = \Delta k^{1+\frac{(1-\theta)}{\Delta \theta}}$, meaning that the best known bounds for the adaptive and non-adaptive settings are identical (though the adaptive algorithm attains {\em zero} error probability).  In contrast, for $\theta > \frac{1}{2}$, we have $\mDD(\Delta) = \Delta k^{1+\frac{1}{\Delta}}$ and $\mada(\Delta) = \Delta k^{1+\frac{(1-\theta)}{\Delta \theta}}$.  The former is significantly higher, and Theorem \ref{thm_inf_theory_non_ada} reveals that this limitation is inherent to {\em any} non-adaptive test design and algorithm.  Hence, for $\theta > \frac{1}{2}$, there is a significant gap between the number of tests required by adaptive and non-adaptive algorithms.

\section{Adaptive Group Testing with $\Gamma$-Sized Tests}

\begin{algorithm}[t]
    \begin{algorithmic}[1]
        \REQUIRE Number of individuals $n$, number of infected individuals $k$, and test size restriction $\Gamma$
        \STATE Initialize infected set $\mathcal{K}\leftarrow\emptyset$
        \STATE Randomly group $n$ individuals into $n/\Gamma$ groups of size $\Gamma$
        \FOR{each group $G_i$ where $i\in\mathbb{Z}:i\in[1,n/\Gamma]$}{
        \WHILE{testing $G_i$ returns a positive outcome}{
        \STATE run Algorithm \ref{alg:binary_splitting} on a copy of $G_i$, and add its one infected individual output $k^*$ into $\mathcal{K}$
        \STATE $G_i\leftarrow G_i\setminus\{k^*\}$
        }\ENDWHILE
        }\ENDFOR
        \RETURN $\mathcal{K}$
    \end{algorithmic}
    \caption{Adaptive algorithm for $\Gamma$-sparse tests \label{alg:adaptive_algo_rho}}
\end{algorithm}

\begin{algorithm}[t]
    \begin{algorithmic}[1]
        \REQUIRE a group of individuals $\tilde{G}$
        \WHILE{$\tilde{G}_i$ consists of multiple individuals}{
            \STATE Pick half of the individuals in $\tilde{G}$ and call this set $\tilde{G}'$. Perform a single test on $\tilde{G}'$.
            \STATE If the test is positive, set $\tilde{G} \leftarrow \tilde{G}'$.  Otherwise, set $\tilde{G} \leftarrow \tilde{G} \setminus \tilde{G}'$.
        }\ENDWHILE
        \RETURN single individual in $\tilde{G}$
    \end{algorithmic}
    \caption{Binary splitting \label{alg:binary_splitting}}
\end{algorithm}

Our adaptive algorithm with $\Gamma$-sparse tests, shown in Algorithm \ref{alg:adaptive_algo_rho}, is again a modification of Hwang's generalised binary splitting algorithm \cite{Hwang_1972}, where we initially divide the $n$ individuals into $\frac{n}{\Gamma}$ groups of size $\Gamma$, instead of $k$ groups of size $\frac{n}{k}$ as in the original algorithm.

Our main result is stated as follows, in which we define 
\begin{equation}
    \mada(\Gamma) = \frac{n}{\Gamma} + k \mathrm{log}_2 \Gamma.
\end{equation}

\begin{theorem} \label{thm_gamma_adaptive}
    For any $\Gamma = o\big(\frac{n}{k}\big)$, Algorithm~\ref{alg:adaptive_algo_rho} outputs the correct configuration of infection statuses with probability one, while using at most $\mada(\Gamma) (1+o(1))$ tests, each containing at most $\Gamma$ items.
\end{theorem}
\begin{IEEEproof}
    Let $k_i$ be the number of infected individuals in each of the initial $\frac{n}{\Gamma}$ groups. Note that since $\Gamma= o\big(\frac{n}{k}\big)$ implies $k= o(\frac{n}{\Gamma})$, most groups will not have a infected individual.  
    In the binary splitting stage of the algorithm, we can round the halves in either direction if they are not an integer. Hence, for each of the initial $\frac{n}{\Gamma}$ groups, we take at most $\lceil\mathrm{log}_2\Gamma\rceil$ adaptive tests to find a infected individual, or one test to confirm that there are no infected individuals. Therefore, for each of the initial $\frac{n}{\Gamma}$ groups, we need $\max\{1,k_i\mathrm{log}_2\Gamma+O(k_i)\}$ tests to find $k_i$ infected individuals. Summing across all $\frac{n}{\Gamma}$ groups, we need a total of $m=\sum_{i=1}^{n/\Gamma}\max\{1,k_i\mathrm{log}_2\Gamma+O(k_i)\}$ tests. This has the following upper bound:
    \begin{align}
        m&\leq\frac{n}{\Gamma}+k\mathrm{log}_2\Gamma+O(k) \stackrel{(a)}{=}\frac{n}{\Gamma}(1+o(1))+k\mathrm{log}_2\Gamma \notag \\ &= \mada(\Gamma) (1+o(1)),
    \end{align}
    where (a) uses $k= o\big(\frac{n}{\Gamma}\big)$. 
\end{IEEEproof}

If we slightly strengthen the requirement $\Gamma = o\big(\frac{n}{k}\big)$ to $\Gamma= o\big(\frac{n}{k\log(n/k)}\big)$ (which, in particular, includes the regime $\Gamma = \big(\frac{n}{k}\big)^{1-\Omega(1)}$ studied in \cite{Gandikota_2016}), then we have $\frac{n}{\Gamma} = \omega\big(k\log\big(\frac{n}{k}\big)\big)$ and hence $\frac{n}{\Gamma} = \omega(k \log \Gamma)$. Thus, we obtain
\begin{align}
    \mada(\Gamma) = \frac{n}{\Gamma}(1+o(1)).
\end{align}
This simplified upper bound is tight, due the simple fact that $\frac{n}{\Gamma}(1-o(1))$ tests (of size at most $\Gamma$) are needed just to test a fraction $1-o(1)$ of the items at least once each (which is a minimal requirement for recovering $\SIGMA$ w.h.p.).  Formally, this argument reveals the following.

\begin{theorem} \label{thm:gamma_simple_converse}
    In the setup of \hspace{1.5mm}$\Gamma$-sparse tests with $k = n^{\theta}$ for some $\theta \in (0,1)$, any (possibly adaptive) group testing procedure that recovers $\SIGMA$ w.h.p.~must use at least $\frac{n}{\Gamma}(1 - o(1))$ tests.
\end{theorem}

\section{Auxiliary Results}

The following variant of the Chernoff bound is convenient to work with (e.g., see \cite[Sec.~4.1]{Mot10}).

\begin{figure*}[t]
    \captionsetup{margin=0.01cm}

    \begin{center}
	    \includegraphics[width=\textwidth]{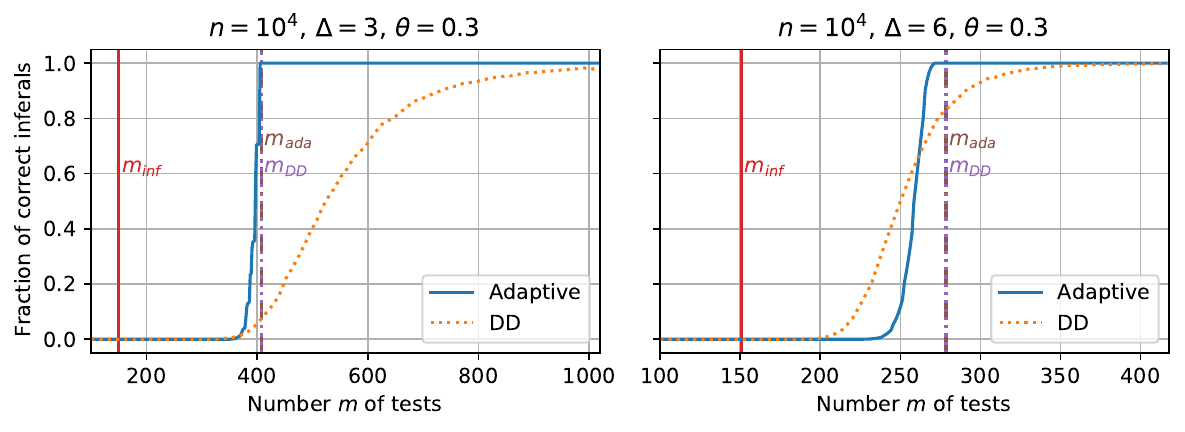}
	\end{center}
	
	\vspace{-1.5em}
	\caption{
	    Performance of adaptive and non-adaptive $\Delta$-divisible algorithms as function of number of tests.
	}
	\label{fig:success_delta}
\end{figure*}

\begin{figure*}[t]
    \captionsetup{margin=0.01cm}

    \begin{center}
	    \includegraphics[width=\textwidth]{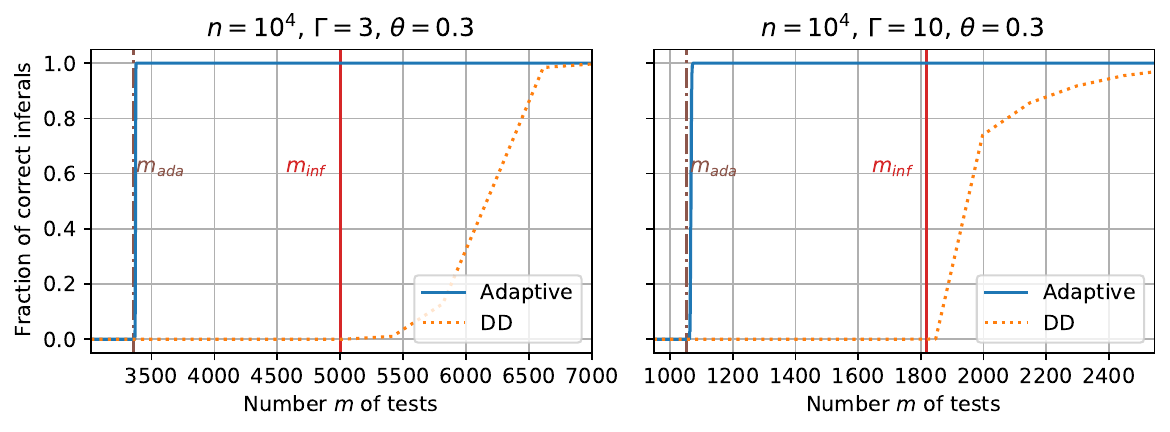}
	\end{center}
	
	\vspace{-1.5em}
	\caption{
	    Performance of adaptive and non-adaptive $\Gamma$-sparse algorithms as function of number of tests.
	}
	\label{fig:success_gamma}
\end{figure*}
    
\begin{lemma}[Multiplicative Chernoff Bound] \label{Lem_Chernoff}
Let $\vX_1, \ldots, \vX_n$ be independent random variables such that $0 \leq \vX_i \leq 1$ a.s., and fix $\delta \in (0,1)$.  Then, we have
\begin{align*}
    \pr{( | \vX - \Erw{[\vX]}| \geq \delta \Erw{[\vX]} )} \leq 2 \exp{ (- \delta^2 \Erw{[\vX]} / 3)}.
\end{align*}
\end{lemma}

\begin{lemma}[Stirling Approximation, \cite{Maria1965}]
\label{Lem_Stirling_Approx} We have for $n \to \infty$ that $$ n! = (1 + O(1/n)) \sqrt{2 \pi n} n^n \exp\bc{-n}.$$
\end{lemma}

\begin{claim}
\label{stirling_applied}
Let $n > 0$, $\Delta = \log^{O(1)} n$ be integers, and let $\alpha \in (0,1)$. Then $$ \binom{\alpha n}{\Delta} \binom{n}{\Delta}^{-1} = \bc{ 1 + O \bc{ n^{-\Omega(1)} } } \alpha^{\Delta}.$$  
\end{claim}
\begin{IEEEproof}
By definition, we have
\begin{align}
    \notag \frac{\binom{\alpha n}{\Delta}} {\binom{n}{\Delta}} &= \frac{ (\alpha n)! (n-\Delta)! }{ n! (\alpha n - \Delta)! }.
\end{align}
Hence, applying \Lem~\ref{Lem_Stirling_Approx} on each factor yields
{\small
\begin{align}
    \label{base_approx} \frac{\binom{\alpha n}{\Delta}} {\binom{n}{\Delta}} &= (1 + O(n^{-1})) \exp \bc{ - \alpha n + (n - \Delta) - (\alpha n - \Delta) - n } \notag \\ & \qquad \cdot (\alpha n)^{\alpha n} (n - \Delta)^{n - \Delta} (\alpha n - \Delta)^{- (\alpha n - \Delta)} n^{-n} \sqrt{ \frac{(\alpha n) (n - \Delta)}{ n( \alpha n - \Delta) } }. 
\end{align}
}
As $\Delta = \log^{O(1)} n $, we find that \eqref{base_approx} equals
\begin{align}
    \label{base_approx2}\frac{\binom{\alpha n}{\Delta}} {\binom{n}{\Delta}} &= \bc{1 + O \bc{ n^{-\Omega(1)} }} ( \alpha n)^{\alpha n} n ^{n} (\alpha n)^{- (\alpha n - \Delta)} n^{-n} \notag \\ & = \bc{1 + O \bc{ n^{-\Omega(1)} }} \alpha^\Delta,
\end{align}
and the assertion follows.
\end{IEEEproof}

We also use the following direct consequence of the binomial expansion.
\begin{claim}\label{Bernoulli}
For any real number $x\geq -1$ and any integer $t\geq 0$ the following holds:
$$(1+x)^t = 1 + tx + O(t^2 x^2).$$
\end{claim} 

Finally, we state the following useful result relating to Stirling's approximation and the local limit theorem.

\begin{claim}\label{Prob_nr_inf_cor}[Appendix B1 of \cite{Coja_2019}]
For any $m,\Delta\in \mathbb{N},\theta\in (0,1),k\sim n^{\theta}$, let $\left(\vX_i\right)_{i\in [m]}$ denote a sequence of independent $\Bin(\Gamma_i,k/n)$ and define 
$$\cE=\cbc{\sum_{i\in[m]} \vX_i=k\Delta}.$$
Then, we have $\Pr \bc{\cE}=\Omega(1/\sqrt{n\Delta}).$
\end{claim}

\section{Simulations}\label{simulation}

\noindent In Figures \ref{fig:success_delta} and \ref{fig:success_gamma}, we compare our theoretical findings to empirical results obtained as follows: 
\begin{itemize}
\item In the non-adaptive case, we fix the number of individuals~$n$, the infection parameter~$\theta$, and, depending on the setup considered, the individual degree~$\Delta$ or test degree~$\Gamma$.
We vary the number~$m$ of tests (x-axis), and simulate $10^4$ independent trials per parameter set.
\DD's performance (y-axis) is reported as the fraction of simulations per parameter point that inferred the infected set without errors.
\item In the adaptive case, we cannot directly control the number of tests~$m$ {\em a priori}.
Instead, we fix the same parameter set as in the non-adaptive case, and carry out $10^6$ simulations.
We then report the cumulative distribution of tests required, i.e., the y-value corresponding to some~$m$ is given as the fraction of runs that required at most~$m$ tests.
\end{itemize}
We observe that the empirical results are consistent with our theoretical thresholds in all cases.  The adaptive testing strategies show a particularly rapid transition at $\mada(\Delta)$ and $\mada(\Gamma)$ respectively. 
We find that the non-adaptive \DD{} algorithm requires more tests in comparison to the adaptive schemes, and has a much broader range of transient behaviour.  This suggests that {\em convergence rates} to the first-order asymptotic threshold may reveal an even wider gap between adaptive and non-adaptive designs, in analogy with studies of channel coding \cite{polyanskiy2011feedback}.
Note that the change of slope in Figure~\ref{fig:success_gamma} (right) at $m{=}2000$ is due to rounding of $\Delta$.

\section{Conclusion}    
  
We have studied the information-theoretic and algorithmic thresholds of group testing with constraints on the number of items-per-test or test-per-item.  For $\Delta$-divisible items, we proved that at least for $\Delta = \omega(1)$, the DD algorithm is asymptotically optimal for $\theta > \frac{1}{2}$, and is optimal to within a factor of $\eul$ for all $\theta \in (0,1)$, thus significantly improving on existing bounds for the COMP algorithm having suboptimal scaling laws.  For $\Gamma$-sized tests with $\Gamma = \Theta(1)$, we improved on both the best known upper bounds and lower bounds, established a precise threshold for all $\theta \in (0,1)$, and introduced a new randomised test design for $\theta > \frac{1}{2}$.  In both settings, we additionally provided near-optimal adaptive algorithms, and demonstrated a strict gap between the number of tests for adaptive and non-adaptive designs in broad scaling regimes.

\section*{Acknowledgments}

OG was funded by DFG CO 646/3. MHK was partially funded by Stiftung Polytechnische Gesellschaft and DFG FOR 2975. OP was supported by the DFG (Grant PA 3513/1-1) and the London School of Economics and Political Science. MP was funded by ME 2088/4-2 and ME 2088/5-1 (DFG FOR 2975). JS was funded by an NUS Early Career Research Award.

\bibliographystyle{IEEEtran}
\bibliography{bibliography}

\begin{IEEEbiographynophoto}{Oliver Gebhard}
studied Mathematics and Economics at Goethe University Frankfurt and University of Toronto. Currently, he is a PhD student under the supervision of Amin Coja-Oghlan.
\end{IEEEbiographynophoto}
\begin{IEEEbiographynophoto}{Max Hahn-Klimroth}
is PostDoc at TU Dortmund University. He studied Mathematics and Computer Science at Goethe-University Frankfurt and obtained his PhD in Mathematics under the supervision of Amin Coja-Oghlan (Goethe-University Frankfurt) and Yury Person (TU Ilmenau).
\end{IEEEbiographynophoto}
\begin{IEEEbiographynophoto}{Olaf Parczyk}
studied Mathematics at the Free University of Berlin and obtained his PhD at Goethe University Frankfurt under the supervision of Yury Person.
He was a PostDoc at Technical University Ilmenau and the London School of Economics and Political Science.
\end{IEEEbiographynophoto}
\begin{IEEEbiographynophoto}{Manuel Penschuck}
studied Computer Science at Goethe-University Frankfurt. He received his PhD in Computer Science from Goethe-University Frankfurt under the supervision of Ulrich Meyer.
\end{IEEEbiographynophoto}
\begin{IEEEbiographynophoto}{Maurice Rolvien}
studied Mathematics at Johannes-Gutenberg University Mainz and Goethe-University Frankfurt. He is currently a PhD student in Mathematics under the supervision of Amin Coja-Oghlan.
\end{IEEEbiographynophoto}
\begin{IEEEbiographynophoto}{Jonathan Scarlett}
    (S'14 -- M'15) received 
    the B.Eng.~degree in electrical engineering and the B.Sci.~degree in 
    computer science from the University of Melbourne, Australia. 
    From October 2011 to August 2014, he
    was a Ph.D. student in the Signal Processing and Communications Group
    at the University of Cambridge, United Kingdom. From September 2014 to
    September 2017, he was post-doctoral researcher with the Laboratory for
    Information and Inference Systems at the \'Ecole Polytechnique F\'ed\'erale
    de Lausanne, Switzerland. Since January 2018, he has been an assistant
    professor in the Department of Computer Science and Department of Mathematics,
    National University of Singapore. His research interests are in
    the areas of information theory, machine learning, signal processing, and
    high-dimensional statistics. He received the Singapore National Research Foundation (NRF)
    fellowship, and the NUS Presidential Young Professorship.
\end{IEEEbiographynophoto}
\begin{IEEEbiographynophoto}{Nelvin Tan}
    received the B.Comp.~degree in computer science and statistics from the National University of Singapore, in 2021. He is currently pursuing the Ph.D.~degree from the Signal Processing and Communications Group in the Department of Engineering, University of Cambridge.  His research interests include information theory and statistical learning. 
\end{IEEEbiographynophoto}

\end{document}